\documentclass[11pt]{now}

\usepackage[linesnumbered,ruled,vlined]{algorithm2e} 
\usepackage{multirow}
\usepackage{textcomp}
\usepackage[breaklinks]{hyperref}
\usepackage{mathtools,cancel}
\usepackage{amsmath}  
\usepackage{amssymb}  
\usepackage{bbm} 
\usepackage{bm} 
\usepackage{colortbl} 
\usepackage{booktabs}  
\usepackage{xcolor}

\hypersetup{
    bookmarks=true,         
    colorlinks, linkcolor=blue, citecolor=blue, urlcolor=black
}

\newcommand{\rr}{\raggedright}
\newcommand{\tn}{\tabularnewline}

\begin{document}

\isbn{xxxxxxxxxxx}

\DOI{xxxxxx}


\abstract{This article discusses the theory, model, implementation and performance of a combinatorial fuzzy-binary and-or (FBAR) algorithm for lossless data compression (LDC) and decompression (LDD) on 8-bit characters. A combinatorial pairwise flags is utilized as new zero/nonzero, impure/pure bit-pair operators, where their combination forms a 4D hypercube to compress a sequence of bytes. The compressed sequence is stored in a grid file of constant size. Decompression is by using a fixed size translation table (\textbf{TT}) to access the grid file during I/O data conversions. Compared to other LDC algorithms, double-efficient (DE) entropies denoting 50\% compressions with reasonable bitrates were observed. Double-extending the usage of the \textbf{TT} component in code, exhibits a Universal Predictability via its negative growth of entropy for LDCs $>$ 87.5\% compression, quite significant for scaling databases and network communications. This algorithm is novel in encryption, binary, fuzzy and information-theoretic methods such as probability. Therefore, information theorists, computer scientists and engineers may find the algorithm useful for its logic and applications. }

\articletitle{A Universal 4D Model for Double-Efficient Lossless Data Compressions}

\authorname1{Philip Baback Alipour}
\affiliation1{University of Victoria}
\author1address2ndline{Dept. of Electrical and Computer Engineering, University of Victoria, }
\author1city{Victoria, }
\author1zip{B.C. V8W 3P6}
\author1country{Canada}
\author1email{phibal12@uvic.ca}

\relax \input sample.jtl 
\volume{xx}
\issue{xx}
\copyrightowner{xxxxxxxxx}
\pubyear{xxxx}

\maketitle

\cleardoublepage \pagenumbering{roman}

\begingroup
\hypersetup{linkcolor=black}
\tableofcontents
\endgroup

\clearpage

\setcounter{page}{0}
\pagenumbering{arabic}

\newtheorem{theorem}{Theorem}[chapter]
\newtheorem{definition}{Definition}[chapter]
\newtheorem{axiom}{Axiom}[chapter]
\newtheorem{paradox}{Paradox}[chapter]
\newtheorem{hypothesis}{Hypothesis}[chapter]
\newtheorem{lemma}{Lemma}[chapter]
\newtheorem{proposition}{Proposition}[chapter]
\newtheorem{assumption}{Assumption}[chapter]
\newtheorem{corollary}{Corollary}[chapter]
\newtheorem{remark}{Remark}[chapter]
\newtheorem{solution}{Solution}[chapter]
\newtheorem{example}{Example}[chapter]
\newtheorem{algo}{Algorithm}[chapter]

\chapter{Introduction}

One of the greatest inventions made in Computer Science, as a building-block for its logical premise was Boolean Algebra, by the well-known mathematician, G. Boole (1815-1864). Its foundation on Boolean operators enlightened further, the great mathematician C. E. Shannon (1916-2001). In 1938, this leading scholar, with reference to Boolean operators~\cite{boole}, managed to show how electric circuits with relays were a suitable model for Boolean logic~\cite{shannon40}. Hence, a model for Boolean logic, as a sequence of 0's and 1's, constituted binary~\cite{shannon93}. From there, he measured information by quantifying the involved \emph{uncertainty} to \emph{predict} a random value, also known as \emph{entropy}. He thus inducted this new entropy with codeword to compress data, losslessly. During this venture of computational science in progress, another mathematician came up with fuzzy sets theory, L. A. Zadeh (1921-present), resulting fuzzy logic with its algorithmic constructs and applications~\cite{ zadeh96, zadeh65}.

In this paper, we put all of these scholars' findings into one \emph{logic synthesis}. Coding this combinatorial logic by \emph{biquaternions}~\cite{hamilton1}, \emph{self-contains} any randomness occurring in a 4D field, delivering a \emph{universal predictability}. Contrary to the notion of randomness, which states: ``the more random, i.e. unpredictable and unstructured the variable is, the larger its entropy"~\cite{hyv, papoulis}, by ``self-containing" the random variable, we then stipulate

\begin{hypothesis}\label{hypo1}
The more random a biquaternion field contains i.e. unpredictable and unstructured the variable in a 4D subspace containment, the smaller its entropy.
\end{hypothesis}

In other words, containing complexity, like a scalable cannon containing a cannonball before ejection, allows complexity's dynamic vectors to remain in \emph{containment} to \emph{relatively reach} the end drop coordinates as a unified result, quite akin to the complexity of all of our universe's randomness \emph{contained} in a dot (or a unifying equation, comparatively~\cite{green}). If the ``complexity vectors" are unleashed from any application, obviously, uncertainty or randomness is emerged.

According to Shannon, ``a long string of repeating characters has an entropy rate of 0, since every character is predictable"~\cite{shannon93}, whereas Hypothesis~\ref{hypo1} self-contains any randomness coming from a string of non-redundant characters, attaining an entropy of 0 bits per character (bpc). If achieved, Hypothesis~\ref{hypo1}, for an observer of the variable, delivers a Universal Predictability theorem:

\begin{theorem} \label{theo0::1}
As the field's entropy grows negatively, i.e. becoming smaller and smaller, its curve gives an observer of the information variable a predictable output.
\end{theorem}

To prove this ``containment of information variable," from Hypothesis~\ref{hypo1}, resulting Theorem~\ref{theo0::1}, we have no need to minimize multi-level logic. In fact, we need to combine logic states correctly using standard and custom operators to obtain \emph{losslessness}. The ``variable containment," is later indicated as $y \in xx'$, which further involves \emph{fuzzy binary and-or operators} to confine the output $y$ content representing the input $xx'$ content. This is introduced as FBAR logic, entailing its fixed compression entropy for its information products throughout the following sections.

\section{Overview}

This paper aims to introduce FBAR logic, apply it to information in a model, causing data compression. The compression model is constructed after introducing the theory of FBAR. From there, its usage and implementation in code are discussed. Furthermore, a clarification between model representation and logic is established for both, the FBAR algorithm and its \emph{double-efficient} (DE) \emph{input}/\emph{output} (I/O) evaluation. The evaluation on the algorithm's efficiency is conditioned by conducting two steps:

\begin{enumerate}[(1)]
\item data compaction and compression processes, using a new bit-flag encoding technique for a lossless data compression (LDC),

\item validating data at the other end with the bit-flag decoding technique for a successful lossless data decompression (LDD).
\end{enumerate}

We introduce FBAR logic from its theoretical premise relative to model construction. We further implement the model for a successful LDC and LDD. The general use of the algorithm is aimed for current machines, and its advanced usage denoting maximum DE-LDCs for future generation computers.

This article is organized as follows: Section~\ref{sect2} gives background information on FBAR model, and its universality compared to other algorithms. It concludes with Subsection~\ref{sect3} introducing FBAR synthesis with expected outcomes. Section~\ref{sect4} focuses on FBAR LDC/LDD theory, model and structure. It introduces FBAR test on data by model components, functions, operators, proofs and theorems. Section~\ref{sect5} presents implementation. Section~\ref{sect6} presents the main contribution made in this work. Subsection~\ref{sect7} describes the experiment on DE performance including results. Subsection~\ref{sect8} onward, end the paper with costs, future work and conclusions.

\chapter{The Origin of FBAR Logic}\label{sect2}

In this section, we review a wide range of existing mathematical theories that are relevant to the foundation of FBAR logic, its model structure arising in lossless data compressions. We also introduce the universal model with a universal equation applicable to LDC algorithms, both in theory and in practice, to perform double-efficient compression as well as communication. Throughout the monograph, the coding theory subsection formulating the four-dimensional model, employs bivector operators to manipulate data symmetrically in the memory's finite field. That is done with real and imaginary parts of bit-state revolutions as high-level 1, or low-level 0 signals, where data is circularly partitioned and stored in the field. We express such operations in form of integrals denoting bivector codes. The memory field equations are integrable when data compaction, compression and four-dimensional field partitioning are both complex and real during communication. \label{sect2.1}

\section{Motivation and Related Work} \label{sect2.1.1}

We at first questioned the actual randomness behavior coming from regular LDC algorithms in their compression products. No matter how highly ranked and capable in compressing data observed on dictionary-based LDCs e.g., LZW, LZ77, WinRK, FreeArc~\cite{bergmans, ziv}, they still remain probabilistic for different input types~\cite{sayood}. These algorithms are mainly based on repeated symbols within data content~\cite{shannon93, mackay}. For example, a compressed output with a \emph{string}\,=\,\texttt{[16a]bc} is interpreted by the algorithm as \texttt{aaaaaaaaaaaaaabc} when decompressed (assuming this was the original data). The \emph{length} of the input string is 16\,B, and for the compressed version is 7\,B, thus we say a 56.25\% compression has occurred. We assess its entropy as Shannon-type inequality, since it minimally involves two mutual random variables~\cite{dembo, makarychev} for the recurring symbols in context.

For such random behavior performed by LDC algorithms sold on the market, the question was whether it would be possible to somehow confine randomness whilst LDC operations occur. This statement motivated the concept of combining the well-known logics to address randomness, both in theory and in practice.

In modern machines, each ASCII character entry from a set of $\geq 2^{7\, {\rm bit}}$code groups, occupies 8 bits or more of space, in which, each bit is either, a low-state or high-state logic. These logic states in combination, build up a character information or their corresponding symbol~\cite{maini, murdocca}. To perform the least probability of logic operations, there must be a definite relatedness between binary logic and its in-between states of low and high for each corresponding symbol. In FBAR logic, this could be recognized at its lowest layers of binary logic between AND and OR operations. Once these operators with negation are applied to original data, manipulating a byte length of pure bits e.g., `\texttt{11111111}' to obtain original data, 8 bits of 0's and 1's is therefore transmitted. This is possible if bivector operators manipulate data in a 4D subspace $\mathbbm{R}^4$~\cite{lanczos}, with a minimally 4 fuzzy bits, thereby, 2 pairwise bits producing compressed data. This \emph{encoding-decoding method} further gives a compression on 2-byte inputs as a reversible 1-byte output, denoting a DE-transmission. This transmission, suggests the relatedness implementation or proof of all logics in FBAR model and relationships.

\section{Relatedness of Logic Types} \label{sect2.1.2}

The relatedness for each character entry on a binary construct is presented by the logical consequence~\cite{zalta} from different models: fuzzy logic~\cite{zadeh96, zadeh65, zalta}, binary, and transitive closure~\cite{jacas, shukla}. By making this uniformity, FBAR logic is emerged. This logic is possible when packets of Boolean values per character are updated and abstracted into relative states of fuzzy and pairwise logic. When we conceive {F, B, AND/OR}, each, as a separate field in calculus, we also conclude that each has its own founder, i.e., chronologically: Boole (1848)~\cite{boole}, Shannon (1948)~\cite{shannon48} and Zadeh (1965)~\cite{zadeh65}. Therefore, for establishing a combinatorial logic model, we question that:

\begin{itemize}
\item Why not uniting the binary part with the highly-probable states of pairwise logic via fuzzy logic?
\item Is there a way to assimilate the discrete version F, B, AND/OR, into one unified version of all, FBAR?
\item Would this unification lead to more probability or else, in terms of predictability?
\item If predictable, what is the importance of it, compared to random states of codeword results?
\end{itemize}

To address each question, it is essential to establish FBAR logic in a combinatorial sense. In essence, the information models known in Information Theory, must be brought into a standard logical foundation as FBAR, representing their logic states combination, computation, information products and application, respectively.

\section{The Foundation of FBAR Model and Logic} \label{sect2.2}

\subsection{Logarithmic and Algorithmic Premise} \label{sect2.0.2}

Here onward, we use Table~\ref{tab1} notations and definitions. For subspace fields, to store, compress and decompress data, we adapt and refer our main findings to Hamilton (1853) \cite{hamilton1}, Conway (1911) \cite{conway}, Lanczos (1949) \cite{lanczos}, Bowen (1982) \cite{bowen}, Girard (1984) \cite{girard},  Lidl and Niederreiter (1997) \cite{lidl}, and Coxeter \emph{et al.} (2006) \cite{coxeter}. Moreover, the algorithmic premise for our algorithms is formulated on the logarithmic preference of information metric as log base 2, which measures any binary content for a character communicated in a message. Foremost, the premise to achieve self-containment on any information input, is to mathematically elaborate on this compression theorem:

\begin{theorem} \label{theo0::2}
Any probability $P$ on information variable $y$ is $1$, if $y$ as a single-character output is contained within the binary intersection limits of its input $xx'$.
\end{theorem}

\begin{table}
\begin{center}
\begin{minipage}{\textwidth}
\caption{Notations and terminology for LDC operations.\label{tab1}}
{\footnotesize \begin{tabular}{@{}clcc@{}} \toprule
Notation & Short definition & Example \rr \tn \hline  \\
$\mathcal{C}_r$ & Data compression ratio & 2:1 compression  \\
$\mathcal{C}$ & Compressed data; compression & $\mathcal{C}_{-1}> \mathcal{C}_n \ , \ n \in \mathbbm{N}$\\
$\mathcal{C}'$ & Decompressed data; decompression & $\mathcal{C}  \stackrel{{\rm out}\times {\rm ref}}{\longrightarrow} \mathcal{C}'$  \\
$H$  & Entropy rate in e.g., Shannon systems & $H_\mathbbm{A} > H_{\wedge\vee(b)}$\\
$x_i$  & A bit, byte or character by scale, where $i \in \mathbbm{N}$ & $\{x_1x_2x_3\dots \}$\\
$y$  & Product of a function, or output & $f(x) = y$\\
$\mathcal{s}$  &  \multirow{2}{5.9cm}{A sequence of an entailed complement $xx'$ (see $\therefore$), or just concatenated values of $x_i$  } & $\mathcal{s}_{\, \rm in}= \mathcal{s}(x) = x_1\!+\!x_2\!+\! x_3$ \\ &  & $ +\ldots=\{x_1x_2x_3\ldots \}\!=\mathcal{s}_{\, \rm out}$\\
$\ell$  & \multirow{1}{5.5cm}{Length function on field, string, time, etc. }  & $\ell \ (xx') = 16$ bits \\
$\infty$  &    \multirow{1}{5.5cm}{Infinity; continuous flow of I/O data, in \emph{measure theory}~\cite{bartle} measured by chars in the flow. $\emptyset$ denotes a null set} &  ${\rm if} \ \mathcal{U}_{{\rm I/O}}= \{\infty\}  \ {\rm then} $ \\ $\emptyset$ & & $ \left(\{\infty\} -  \mathcal{s}(x)\right)^{c} $ \\ & & $= \emptyset \cup \mathcal{s}(x)  = \mathcal{s}(x)$\\
$\mathbbm{R}^{2^n}$ & A $2^n\!\rm D$-product space with a topology of & \multirow{2}{3.2cm}{$\forall xx' \in \mathbbm{R}^{2^n} \, ; \, n=2\, ,$ \\ $\ xx' \mapsto \{y\}_{i,j,k,l\in\mathbbm{R}^4} \ ,$ \\ $\therefore \hat{{\bm v}}_y = \frac{1}{\sqrt{2^n}}xx'\, = 1\mathrm{B}$ } \\  & mapping bit-pairs of input characters into  \\ & subspace partitions, where $n=2$\,characters. \\ $\hat{{\bm v}}$ & A spatial unit vector = 1\,bit, 1\,byte, etc. & $= \left(\frac{1}{\sqrt{2}}x_i \, \frac{1}{\sqrt{2}}x_j \, \frac{1}{\sqrt{2}}x'_k \, \frac{1}{\sqrt{2}}x'_l\right)$ \\
$\mathbf{e}_{ij}$  & A unit bivector for bit-pair mappings &  $\forall \mathbf{e}_{ij} \in \mathbbm{C}\ell_4\mathbbm{R}^{4} \, ; \, \mathbf{e}^{2}_{12} = -1 $ \\
$A$  & An array for the residing bits in memory & \multirow{2}{3.4cm}{if  $\beta \vdash x_i = 0 \ $ then , $ \ A_{1\times n}  =[000\ldots0] $} \\
$\vdash$ & A sequent; derived from; yields \ldots &  \\
$\beta$  & \multirow{1}{5cm}{Binary value or sequence, where  $\forall \beta \in f(x) = x \rightarrow y = b$} & \multirow{2}{3.2cm}{if $\beta \vdash x_i = 0$ and $x'_i = 1$ $\therefore \beta = 01010101$ }\\
& & & \\
$\wedge \, , \, \bigcap$  & Logical AND; for sets as Intersection & $1\wedge0 = 0 \ , 1\wedge 1 = 1$\\
$\vee \, , \, \bigcup$  & Logical OR; for sets as Union & $1\vee0 = 1 \ ,0\vee 0 = 0$ \\
$\leftrightarrow$  & \multirow{1}{5cm}{Bi-conditional between states or logic; if and only if; iff} & \multirow{2}{3.2cm}{$x \leftrightarrow y \equiv $ $(x \rightarrow y)\wedge (y \rightarrow x)$}\\
& & & \\
$\equiv$  & Equivalence; identical to \ldots & $2 \ {\rm chars} \equiv 16 \ {\rm bits} $\\
$\therefore $  & Logical deduction; therefore \ldots & \multirow{2}{3.2cm}{ if $\{x_1x'_1\}=\{\$2\%1\}$ , $  \therefore x_1=\$2 \ , \ x'_1=\%1$} \\
& & & \\
\fbox{\scriptsize{\emph{component}}} & \multirow{1}{5.5cm}{Algorithm component as an I/O object, \textbf{P} as a program with filter, \textbf{G} as a grid file, \textbf{TT }as a translation table} & $\mathcal{s} \stackrel{\rm in}{\rightarrow} \fbox{\bf{P}} \stackrel{\rm out}{\rightarrow} \mathcal{C}$ \\
& & & \\
& & & \\
$\otimes \, $  & \multirow{1}{5.5cm}{Strong conjunction on array values; matrix vector or finite field product} & \multirow{2}{3.2cm}{$\{8 \, {\rm bits}\}\otimes \left[
                                                                                                            \begin{array}{ccc}
                                                                                                              1 & 0 & 0 \\
                                                                                                              1 & 2 & 0 \\
                                                                                                              1 & 2 & 3 \\
                                                                                                            \end{array}
                                                                                                          \right]
= \left[
                                                                                                            \begin{array}{ccc}
                                                                                                              8 \, {\rm bits} & 0 & 0 \\
                                                                                                              8 \, {\rm bits} & 8 \, {\rm bits} & 0 \\
                                                                                                              8 \, {\rm bits} & 8 \, {\rm bits} & 8 \, {\rm bits} \\
                                                                                                            \end{array}
                                                                                                          \right]$ = \{6 bytes\}}\\
& & & \\
& & & \\
& & & \\
& & & \\
& & & \\
& & & \\
& & & \\
& & & \\\hline
\end{tabular}}
\tablefootnote{a}{These notations are used in defining LDC operations between algorithmic components, model and logic. Those notations that are not listed here, are defined throughout the text, or in the ending section `Notations and Acronyms', before `References'.}
\tablefootnote{b}{Some notations imply bivectors~\cite{lounesto}, entropy and complexity (Sections~\ref{sect4.3}, \ref{sect5}-\ref{sect7}).}
\end{minipage}
\end{center}
\end{table}

Theorem \ref{theo0::2} lays out the foundation of self-containing $xx'$ as $y$ in preserving all probability $p(xx')\rightarrow P(y)\rightarrow1$ counts, against any ``surprisal" as a highly improbable outcome $p(xx')\rightarrow0$ or uncertainty $u\rightarrow\infty$~\cite{tribus}. The current goal is to ``self-contain" $xx'$ within the limits of \emph{self-information} $I$ measure on $y$.

Now consider the definition revisited by Bush (2010)~\cite{bush} on ``self-information" as:``a measure of the information content associated with the outcome of a random variable." Further, ``the measure of self-information is positive and additive." In contrast, as we prove in Section~\ref{sect4}, \emph{the measure of self-containment is positive and conjunctive, but not additive}. So, any binary content as a given input is partitioned in its \emph{dual space output}~\cite{lounesto} when contained by 4D \emph{bivector operators}, or

\begin{definition} \label{def2.2} Self-information containment is associative in binary states of a given input, preserving its equally combined additive and conjunctive function using 4D bivector code operators, returning a constant size output stored in an array $A$.
\end{definition}

This associativity between logic states in $A$, returns an information constant as data$_{\rm in}$ in entanglement or a  bivector DE-coding. The coding objective is to put all logic states of an information input into two places at once as a unique address in $A$. The array stores an event as an output character $y$ denoting two original events as input characters $xx'$. In essence, suppose event $\Gamma=y$ content is composed of two mutually independent events $\Theta $ as $x$ content, and $\Lambda$ as $x'$ content. The amount of information when $\Gamma$ is ``communicated" equals the \emph{combination} of the amounts of information at the communication of event $\Theta $ and event $\Lambda $, simultaneously. \\

\noindent {\bf Coding objective.} The strong conjunction `$\otimes$' of $xx'$ contents must satisfy their binary ``combination": $x$ as $\geq$ 1\,B intersected with $x'$ as $\geq$ 1\,B, gives a $y=\frac{1}{\sqrt{2^n}}xx'=\frac{1}{\sqrt{2}}x\frac{1}{\sqrt{2}}x'=\frac{1}{2}xx'$ content size in a locally-compact space $\mathbb{R}^{N}$ on $A$. Taking into account that $y$ would be the compressed content of $xx'$, located in a specific subfield address. The subfield as $\mathbf{F}_{xx'}$ is a dual space partitioning $x$ and $x'$ bit-pair values into 4 bivector dimensions, hence the notion $\frac{1}{\sqrt{2^n}}xx'$, partitioning $N=4$ into $2^n$. The field's range is of single-byte addresses or rows $r$ available to store $y$, such that
 \begin{equation*}I(\Theta\cap\Lambda)\cap \max \left\{\mathbf{F}_{xx'}\right\}= \left\{I(\Theta)+ I(\Lambda)\right\} \otimes A_{r\times N} = I(\Gamma) \otimes A_y=I\left(A_\Gamma \right) \end{equation*} This relation portrays Definition~\ref{def2.2}, and its rows or address limit is established by

\begin{lemma}\label{lemma1.4}
A single-byte input $x$ has $2^8=256$, $\{0, 1\}$ bit combinations, $r$ rows 0 to 255. Thus, for a 2-byte input $xx'$, we possess $k=2^{(8+8)}=65536$ combinations as the maximum range of its finite field $\mathbf{F}_{xx'}\in\mathbb{R}^4$ addresses, building an array $A_{k\times4}$.
\end{lemma}

The lemma holds even if all pairs of bytes input, $xx'$, in the order of $2, 4, 8, \ldots, 2^n$ as the number of characters are compressed into 1 byte output, $y$. Once a \emph{translation} on the \emph{intersection of combinations} are decoded, a lossless decompression is gainable from the given \emph{Coxeter order}~\cite{coxeter} In:Out as $xx'$:$y$= 2:1, 4:1, 8:1, $\ldots, 2^n$:1. Henceforth, this ordered sequence of ratios is called ``double-efficiency" for any transmitted data as compression relative to its successful decompression. Reference to Lemma~\ref{lemma1.4} in aim of proving Theorem~\ref{theo0::1}, we then propose

\begin{proposition} \label{prop0}
If a binary intersection of $x$ and $x'$ by a 4D bit-flag function $\varphi$ produces $y$, translating the intersected bivector combinations by a $\vartheta$ function, conversely produces $xx'$. These flags have a physical space occupation of $\left(
        \begin{array}{cccc}
          h^2\mathbf{e}_{12}^2 & h^2\mathbf{e}_{34}^2 \\
        \end{array}
      \right)$.
\end{proposition}

We investigate this new type of logic directed to one LDC-DD algorithm:

\begin{algo} \label{algo0}
\begin{align*}&\forall P(y)\subset P(xx')\, ; \, \mathrm{if}~\varphi(xx') = \left\{ \log_2 65536  \cap  \log_2  65536 \right\} = 8\,{\rm bits} = y \\ &\in 65536 \left(
        \begin{array}{cccc}
          h^2\mathbf{e}_{12}^2 & h^2\mathbf{e}_{34}^2
        \end{array}
              \right) = \left(
              \begin{array}{cccc}
          2\,{\rm bits} & 2\,{\rm bits} & 2\,{\rm bits} & 2\,{\rm bits} \\
        \end{array}
        \right)_{64 \, \rm K}
,  \\ &\mathrm{then} \, P\left(\vartheta(y)\right) \rightarrow P(xx')  =\vartheta \left(\log_2 65536 \cap \log_2 65536\right) \rightarrow \left\{\log_2 65536 \, \cup \right. \\ &\log_2 65536\} = 16\times16 \times 16 \times 16 bit-flags \rightarrow 16 \text {bits original data.} \\ &\text{Thus, probability} P\left(\vartheta(y)\right) \rightarrow P(xx')= 1.\end{align*}
\end{algo}

Proposition~\ref{prop0} gives a predictable outcome for all probability scenarios on the compressed $y$ with a $P$ representing $xx'$ contents for a lossless decompression. Predictability is achieved only if the \emph{complete algorithm}, Algorithm~\ref{algo0}, runs according to its $\cap$ and $\cup$ operations. It should \emph{configure bit-flag} $\varphi$ and \emph{execute code translation} $\vartheta$ function with relevant operators on $xx'$ and $y$. The general use of $\cap$ and $\cup$ operators are expressed by the universal FBAR equation in Section~\ref{sect2.1.3}.

What we mean about ``universal predictability" is that irrespective of the number of inputs, the output is predictable before logic state combinations. As a result, \emph{invariant entropies of higher order} (\emph{negentropy}~\cite{hyv}) with more complex coding, becomes predictable. We synthesize all of the presented logics into this combinatorial logic via logical operators as categorized in Section~\ref{sect3}.

\section{A Universal FBAR Coding Model and Equation} \label{sect2.1.3}

We first commence with an assumption

\begin{assumption}
Let an information input $x$ to our machine lie in the interval $[0, 1]$. Assume function $f$ operates on $x$ between its logic states as binary, otherwise fuzzy, producing $x$ with a new value. Let the machine produce this value via standard logical operators as: and $\wedge$, or $\vee$, union $\cup$, intersection $\cap$, and negation $\neg$.
\end{assumption}

\noindent and then we define its universal relation, once proven in terms of

\begin{definition}
Relation $\mathfrak{R}$ is a universal relation, if presented with union $\cup$, and intersection $\cap$, between fuzzy set $\tilde{\mathcal{A}}$ and binary set $\mathcal{A}$ simultaneously. If $\mathfrak{R}$\emph{=}$\cup$, $\mathcal{A} \mathfrak{R} \tilde{\mathcal{A}}$ succeeds in many pairwise states $\mathcal{A}|\tilde{\mathcal{A}}$. Conversely, if $\mathfrak{R}$\emph{=}$\cap$, $\{\mathcal{A}|\tilde{\mathcal{A}}\}\mathfrak{R}\tilde{\mathcal{A}}$ succeeds in binary states 0\,or\,1. Thus, a fuzzy-binary function $\Phi$ for all $\mathfrak{R}$'s is given by \begin{equation}\label{eq1}\begin{split} \Phi_{\wedge\!\vee} (x) &=  \mathcal{A}   \mathfrak{R}  \tilde{\mathcal{A}}  \mathfrak{R}  \left\{\mathcal{A}|\tilde{\mathcal{A}}\right\}\mathfrak{R} \mathbbm{C}\ell_4\left(\mathbbm{R}^{2^n}\right) \\ &\equiv \{0, 1\} \leftrightarrow \{  [0 , 1] \} \leftrightarrow \{  00 , 01, 10, 11 \}\ , \ n = |\mathcal{A}\cup \tilde{\mathcal{A}}| \leq 4 \end{split}
\end{equation}
where both finite sets $\mathcal{A}$ and $\tilde{\mathcal{A}}$ membership values are contained by dual carriers as bivectors in a partitioned real field $\mathbbm{R}^{2^n}$, projecting bit-pair values from a real or complex plane in $\mathbbm{C}(\mathbbm{R})$ with a dimensional length $\ell_i$. Based on the \emph{inclusion-exclusion principle}~\cite{balakrishnan}, $n$ is equal to the number of elements in the union of the $\mathcal{A}$ and $\tilde{\mathcal{A}}$ as the sum of the elements in each set respectively, minus the number of elements that are in both, or $|\mathcal{A} \cup \tilde{\mathcal{A}}|=|\mathcal{A}|+|\tilde{\mathcal{A}}|-| \mathcal{A}\cap\tilde{\mathcal{A}} |$.
\end{definition}

Furthermore, from the well-known scholars, we plug the latter definition into their scalar bivector definitions, hence deducing

\begin{definition} \label{def2.7}
In Eq.~(\ref{eq1}), a scalar element $h = \sqrt{-1}$ by Conway (1911) \cite{conway}, its dual and quadruple forms, $h^2= -1$ and $h^4=1$, respectively satisfy combinatorial operators during 2$^n$D and 4D bit-pair spatial partitioning and projections. The 2$^n$D type projections are of Coxeter group in the order 2, 4, 8, \ldots by Coexeter \emph{et al.} (2006) \cite{coxeter}, and configure matrix vertices to store, compress and decompress data in a hypercube.
\end{definition}

Now, we begin the proof of the universal model by a theorem,

\begin{theorem}
Relation $\mathfrak{R}$ is universal, iff $\mathfrak{R}=\{\cup , \cap\}$ on all logic states stored in a dual space, yield from a 2$^n$D $\leftrightarrow$ 4D bivector field, where $n$ is the possible number of pairable states. This produces a combinatorial fuzzy-binary and-or equation.
\end{theorem}

The following equations prove the relatedness of all probable logics from one side of Eq.~(\ref{eq1}) to another:

\begin{proof} The fuzzy unit in its membership function $\mu(x)$ with a numerical range covering the interval [0, 1] operating on all possible values, gives minimally $n\geq3$ possible states~\cite{passino, zadeh65}. Binary, however, in its set is discrete for a possible 0 or 1. Let for $\mu(x)$, fuzzy membership degrees close to 0 converge to 0, and those close to 1 converge to 1 as an independent state with a periodic projection (an integral), stored onto a closed surface of 2D planes for a bivector decision $\hat{{\bm\mu}}$. This decision is derived from such input values projected into a dual space forming a hypercube. That is the space $\bigwedge^2\mathbbm{R}^4$ dual to itself in $\mathbbm{C}\ell_4(\mathbbm{R})$~by Lounesto (2001) \cite{lounesto}, where every plane of data (in form of bitpairs) is orthogonal to all vectors in its dual space. The projection for a given bit is done by $2^n$-bivectors, partitioning the input as bitpairs into the space. Now having a \emph{dual space output}, by using the Pythagoras' theorem, the output covers a projection of $x$ from either set $\mathcal{\tilde A}=\{\approx  0, \approx 1\}$ or $\mathcal{A}=\{1, 0\}$, as a \emph{hypotenuse transformation} $\vartheta$, in terms of
\begin{subequations}
\begin{equation}\begin{split} \vartheta\left( \hat{{\bm\mu}}_{\mathcal{A}, \tilde{\mathcal{A}}}  \right) &= \hat{{\bm\mu}}_{x} \sqrt{\left( h_{12}\mathbf{e}_{12} + h_{34}\mathbf{e}_{34} \right)^2} \equiv \left\{\, \hat{{\bm\mu}}_{x} \| \mathbf{e}\| = \hat{{\bm\mu}}_{x} \sqrt{\mathbf{e} \cdot \mathbf{e}}\,\right\} \\
&=\hat{{\bm\mu}}_{x}\sqrt{h^2_{12}\mathbf{e}_{12} \mathbf{e}_{12} + h^4_{1234}\mathbf{e}_{12} \mathbf{e}_{34} + h^4_{1234}\mathbf{e}_{34} \mathbf{e}_{12} + h^2\mathbf{e}_{34} \mathbf{e}_{34}} \\ &=\hat{{\bm\mu}}_{x}\sqrt{-\mathbf{e}^2_{12} +2\mathbf{e}^2_{1234}-\mathbf{e}^2_{34}} = \hat{{\bm\mu}}_{x}\sqrt{2 +2\mathbf{e}^2_{1234}}\\
&=1\sqrt{4} = 2 \ {\rm bit \ states \ per \ \hat{\mu} \ vector}\vdash \left\{\mu_{\mathcal{\tilde A}}, {\mu_{\mathcal{A}}}\right\} , \\
\therefore \forall x \in  &\int_{S_A} \! \! \hat{{\bm\mu}}_{x} \, {\rm d} \, \alpha \left\{
              \begin{array}{cccc}
                 \lim_{x \to \min(x)}f\left(\hat{{\bm\mu}}_{\mathcal{A}\cap\tilde{\mathcal{A}}}(x)\right) = \left(0\leftarrow \{\approx 0\} \right)_{\mathbf{e}_{ij}} = 0 \ , \\
                \lim_{x \to \max(x)}f\left(\hat{{\bm\mu}}_{\mathcal{A}\cup\tilde{\mathcal{A}}}(x)\right) = \left(1\leftarrow \{\approx 1\} \right)_{\mathbf{e}_{ij}} =1 \ \ \\
              \end{array}
            \right. \end{split}  \label{eq:2a}\end{equation}

\noindent where $\mathbf{e}^2_{ij}= -1$, and its product ${\rm d} \alpha$, is the area element of array surface $S_A$, occupied by a bit or $x$, thus its full frequency occupation is 2-bit states and converges to a $\pm 2k\pi$ radian of the projections made onto hypercubic planes (lattices). We obtain this by coding a \emph{surface-volume integration}, modeled in Fig.~\ref{fig1}. It shows that a \emph{compression hypercube} $Q_y$ containing encoded data is formed like a \emph{tesseract}~\cite{bowen}, when at least $\theta=4\pi$ radians occur. It minimally contains bit-pair values, and maximally 2\,B, or a closed pair of characters as an $xx'$ message, stored into two places at once. Hence, by a \emph{cylinder method}~\cite{anton}, we then elicit a combinatorial integral
\begin{equation*}\begin{split}
V_{Q_y}&=\hat{{\bm\mu}}\!\int_{S\rightarrow V_A} \! \! \! \! \! \! \! \vartheta\left( x, \mathbf{e}_{ij}\right) \, {\rm d} \, \alpha = 2\pi \!\int_{\frac{x}{\ell_4}}^{x'} \rho_{x}\phi_{x}\, {\rm d} \, x = 2\pi\! \int_{\mathbf{e}_{12}}^{\mathbf{e}_{34}} \!x \left|\vartheta\left( \hat{{\bm\mu}}_{x}\right)\right| \, {\rm d} \, x \\ &= \Delta {\bf e}\left( \circ, {\bf s} \right)_A = {\hat{{\bm\mu}}_{x}\mathrm{d}S_{\circ}} \!\rightarrow\hat{{\bm\mu}}_{x}\mathrm{d} V_{\circ}\!=  {\mathrm{d} \rho \hat{{\bm\rho}} + \rho \mathrm{d} \theta \hat{{\bm\theta}} + \mathrm{d} \phi \hat{{\bm\phi}} }  \rightarrow \rho \mathrm{d} \rho \mathrm{d} \theta \mathrm{d} \phi \\  &= \tfrac{1}{2}\pi \hat{{\bm\mu}}_{x} \| \mathbf{e}\| \! \rightarrow 2\pi \hat{{\bm\mu}}_{x} \| \mathbf{e}\|^2
\\ &= S_{\bullet} \mathop{\stackrel{\cap \,}{\longrightarrow}}_{\pm \pi} V_{\bullet}\!=\! \left( \!\tfrac{\sqrt{4}\pi}{2}\cap \right. \!\!\tfrac{4\pi\sqrt{4}}{\pm\pi} \!\cap \left.\! \stackrel{2\pi\sqrt{4}^{2}}{} \!\right) \! \approx\!  (3, 25] \cap \!  \left|\tfrac{25}{\pm\pi}\right|\! \end{split}\label{eqs:2}\end{equation*}\vspace{-6pt} \begin{equation}=\!(3,8] \, {\rm bits} \in\! A_{\mathbbm{R}^4} \ \ \ \ \ \ \ \ \ \ \ \ \ \ \ \ \ \ \ \ \ \ \ \ \ \ \ \ \ \ \ \ \ \ \ \ \ \ \ \ \ \ \ \ \ \ \ \
\label{eq:2b}\end{equation}\end{subequations}

\noindent where $\rho_x$ is the area radius equal to the binary length of input $x$, in this case, quantified as a planar binary sum inclusion $\Phi_{\wedge}\sum\beta(x)=[2, 4)$ bits, and $\phi_x$ is the input projection equal to the magnitude of bivectors ${\bf e}_{ij}$ from Eq.~(\ref{eq:2a}), in this case $\phi_x = \| \mathbf{e}\| $. Line element $\mathbf{s}$ by code integrates the $\rho$ and $\phi$ quantities to form a cube $Q$ with volume $V$ by the bivectors when traversed, or, ${\bf s}\Delta {\bf e} = {\bf s}\left| {{\bf e}_{34}-{\bf e}_{12} } \right|$. Orthogonal vectors $\hat{{\bm\rho}}, \hat{{\bm\theta}}$, $\hat{{\bm\phi}}$, via $\hat{{\bm\mu}}$, denote $[2, 4]$ dimensions to store and sort data into 1 or more empty vertices $\circ$, of array $A$. The left integral result denotes an occupiable surface $S_\circ >$ 3 bits of $x$ via $\rho_x$, projected onto $\circ_x$ (stored as $\bullet_x$ via $\phi_x$) producing $S_\bullet$ before forming $V_{Q_{x}}$via ${\bf e}_{ij}$. The overlapping results via $\cap$, denote a compressed volume of filled vertices, or $V_{\bullet}=\! \left|\tfrac{25}{\pm\pi}\right|\!\approx 8 $ bits as: two cubes having 16 vertices built by the bivectors containing two input characters in a simultaneous $\pm\pi$-communication inside a big cube as 0's and 1's entanglement. This cube self-contains the subcubes in its subfield, a decodable $xx'$ data packed into a $y$ as its 8 outer-vertices (bits), in total 24 co-intersecting vertices involved. The 4D model is illustrated in Fig.~\ref{fig1}.

\begin{figure}
\begin{center}
\includegraphics[scale=2.35]{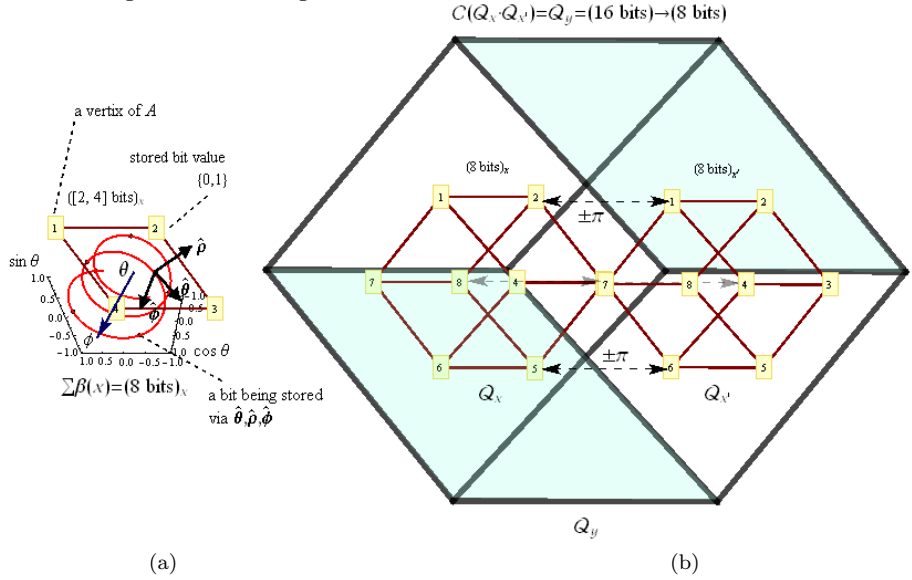}
\footnotesize (a) \hspace{2.5in} (b)
\end{center}\vspace{-12pt}
\caption{Illustrates a geometric model of bit-pair projections over $\theta$ forming a side-by-side tesseract. (a) shows a projection from the $\theta$ plane with a growing revolution that forms a volume containing 8 bits; (b) shows two sub-cubes denoting a bit-set of at least 16 bits under compression within an 8-bit cube. The equation above the hypercube represents this compression. The two-input $\pm\pi$-communication denoting entangled bit-set states, is viable for superdense coding operators to store data into a single \emph{qubit}~\protect\cite{bennett}, as our future quantum compression model. \label{fig1}}
\end{figure}

Further associating both fuzzy results with binary, each as an independent state, in total builds four simultaneous bit-pairs (each subcube of 8 bits), thus giving \begin{subequations}\begin{align}
\Phi_{\vee} (x) &= \left(\forall x = 0 \in \mathcal{A} \cup \tilde{\mathcal{A}}\left|V_{Q_x}\right. \right) \vee \left(\forall x = 1 \in \mathcal{A} \cup \tilde{\mathcal{A}}\left|V_{Q_x}\right. \right) \notag \\ & \left. = \{0, 1\} \rightarrow \{  [0 , 1] \} \rightarrow \{ \ 00 , 01, 10, 11 \} \right.  \,
\label{eq:3a}\end{align}

\noindent This is pairwise logic for  many possibly contained (compressed) values of fuzzy as well as binary, and in its default set covers 8 states. By using a fuzzy unit on each bit-pair, we abstract the pairwise version to binary, which is an inverse process, or
\begin{align}
\Phi_{\wedge} (x)&=\left(\forall x=0|\{0,1\} \wedge \forall x=1|\{0,1\} \in \left\{\mathcal{A}|\tilde{\mathcal{A}}\right\} \cap \tilde{\mathcal{A}}\left|V_{Q_x}\!\rightarrow S_{Q_x}\right.\right) \notag \\ & \left. = \{0, 1\} \leftarrow \{  [0 , 1] \} \leftarrow \{ \ 00 , 01, 10, 11 \}\, \right.
\label{eq:3b} \end{align}

\noindent The \emph{fuzzy-binary and} function $\Phi_{\wedge}$ uses logical operators \emph{and-or}, \emph{negate} and \emph{closures} e.g., \emph{transitive closure}~\cite{jacas} to generate crisp logic. Combining Eqs.~(\ref{eq:3a}) and (\ref{eq:3b}), further outputs a \emph{fuzzy-binary and-or} $\Phi_{\wedge\!\vee}$ or Eq.~(\ref{eq1}), such that \begin{equation}
\left\{\left(\Phi_{\vee} \rightarrow \Phi_{\wedge} \right) (x)\right\} \wedge \left\{\left(\Phi_{\vee} \leftarrow \Phi_{\wedge} \right) (x)\right\} = \Phi_{\wedge  \! \vee} (x) \left|S_{Q_x}\leftrightarrow V_{Q_x}\leftrightarrow V_{Q_y}\right. \
\end{equation}

\noindent where $\Phi_{\wedge  \! \vee} (x)$ usage in the hypercube model, appears valid in its compression ratio \begin{align}
\mathcal{C}_r (x) &=\frac{\mathcal{C}(Q_x \cdot Q_{x'})}{Q_x + Q_{x'}}=\frac{Q_y}{\mathcal{C}'(Q_y)}= \Phi_{\wedge  \! \vee} \sum \beta (x)\vdash|xx'|\leftrightarrow|y| \notag \\
&=(16 \, {\rm bits})\leftrightarrow(8 \, {\rm bits})= (2 \, {\rm B})\leftrightarrow(1 \, {\rm B}) \notag \\
&= 2\!:\!\!1 \ \text{compression.}
\end{align}
\label{eqs:3}\end{subequations}
\end{proof}

\begin{remark}
Fuzzy convergence between maximum and minimum of $x \in [0,1]$ implies to many-valued logic~\cite{gottwald}, now in abstraction by $\wedge ,\vee$ operators.
\end{remark}

\noindent and

\begin{remark}
In recognition of Eqs.~(\ref{eqs:2})-(\ref{eqs:3}), the
$\mathfrak{R}$ relationships in Eq.~(\ref{eq1}), denote an encoded-decoded, compressed-decompressed data projected into a hypercube, or conversely from its dual space of at least two subcubes on either end of the equation.
\end{remark}
\section{Compression Products Aimed by the FBAR Algorithm} \label{sect2.1.4}

To deliver a DE-transmission, our first approach was to study textual samples with binary constructs for a lossless compression, similar to the approach made by Shannon (1948)~\cite{shannon48} on English alphabet. Of course, with a main difference. We used standard characters in ASCII, with their 256 bit-byte combinations ($2^8$ bits, Lemma~\ref{lemma1.4}) on both binary and text. The resultant logic could be employed in the order of integer multiplications. For example, the RGB colors satisfy a huge number of possible combinations i.e., a 3-table consisting $256_{\rm R}\times256_{\rm G}\times256_{\rm B} \approx 16.7$ million combinations. This integer in turn supports other data types or non-English spoken languages (Unicode tables). In this paper, however, the primary scope for the first version of FBAR is $256\times256 = 65536$ combinations. Its future versions, hypothetically grow toward much greater numbers beyond a 3-table, i.e., a 4-table for a $\mathcal{C}_r$ = 8:1 or 87.5\% compression as $65536^{4 \: \text{tables}}$ = 16 exabytes (EB). The latter is convenient for managing \emph{very large databases} $\geq$ 1\,TB. Such hypotheses are discussed in our future work section, Section~\ref{sect8}, and translation table in Section~\ref{sect5}.

\section{FBAR Synthesis}\label{sect3}
Let an algorithm synthesize any logic state known to quantify information. To quantify, we need operators that operate on logic states. Those operators would be:

\begin{enumerate}[(1)]
\item Boolean logic: Boolean operators as \emph{and}, \emph{or}, and \emph{not} or \emph{negate}. Boole (1848)~\cite{boole}
\item Fuzzy logic: Fuzzy operators as \emph{fuzzy-and}, \emph{fuzzy-or}, and \emph{not} or \emph{negate}. Zadeh (1965)~\cite{zadeh65}
\item Fuzzy-binary and-or (FBAR) combinatorial logic: synthesize all two above as \emph{fuzzy-binary-and}, \emph{fuzzy-binary-or}, and \emph{not} or \emph{negate}. \label{synthesis3} This gives a lossless
\begin{enumerate}[(i.)]
\item dynamic FBAR encoding and compression,
\item static FBAR encoding and compression,
\item FBAR decoding and decompression. Alipour and Ali (2010) \cite{alipour10}
\end{enumerate}
\end{enumerate}

In this paper, we focus on 3.ii and 3.iii methods used in the algorithm to achieve ``universal predictability." This could only happen if both methods are conducted in terms of FBAR \emph{pairwise logic}: synthesize all three logics via \emph{and-or} and \emph{negate} operators on pairs of 0, near 0, 1, and near 1 states to \emph{minimally transmit} 8 bits, or

\begin{enumerate}[(1)]
 \item * a possible combination of 8 bits or $\{00, 01, 10, 11\}$ denoting 4 possible bit-flag combinations (1-bit operators) on the four bit pairs (1-character input)
\begin{enumerate}[(i.)]
 \item FBAR bit-flag operators \textbf{z} for zero, \textbf{n} for negate, \textbf{i} as impurity, and \textbf{p} as purity operators on each pair of bits (definitions are given in Section~\ref{sect4.3})
\item Employing \textbf{znip} operators in code, their minimum dimensional intersection as $\{\textbf{z}, \textbf{n} , \textbf{i} , \textbf{p}\}\times \{\textbf{z}, \textbf{n} , \textbf{i} , \textbf{p} \}$ allows such ``transmissions" occur.
\end{enumerate}

This in total gives 4 \emph{fuzzy-binary} or {\it fb} bits (definitions are given in Section~\ref{sect4.3})
 \item * Thus, a minimum of 8\,bits is transmitted via 4\,bits. Since we need to represent data through standard 8-bit characters, the 4\emph{fb} bits is concatenated with another 4\emph{fb} bits, thus a pair of characters or 16 bits input are transmitted via 8\emph{fb} bits output. This is an FBAR compression product.
\end{enumerate}

\subsection{Expected Outcomes} \label{sect3.1}
From (2)*, DE entropies are emerged denoting 50\% and higher fixed compressions, regardless of the number of inputs given to the algorithm. This is theorized before our FBAR technique in Section~\ref{sect4}, such as \textbf{znip} bivector operators are introduced and benefited from their \emph{product elements} and \emph{conjunctive normal form} (CNF) \emph{conversions}~\cite{cockle, jackson}, thereby tested and employed within the algorithm in Section~\ref{sect5}. Also, by referring to \emph{Hamming distance} and \emph{binary coding}~\cite{hamming, sloane}, a preamble on \emph{prefix coding} to encode and decode data by operators is formulated. Higher fixed compressions are hypothesized for a negative entropy growth, once the minimum degree of a DE-compression or 50\% is proven in theory and in practice. In Section~\ref{sect5}, we demonstrate these entropies by proving the minimum degree of our DE-technique relative to its 4D-model of I/O data. We then prove decompression by satisfying the decoding method employed for the compressed data. Decompression is done by referring to byte addresses in the compressed file denoting the original input in the last two DE possibilities, (1)* and (2)*. We evaluate all of the hypotheses in the FBAR technique to demonstrate the validity of our concept. The efficiencies of compression in Theorem~\ref{theo0::1}, are further evaluated through complexity measures on the algorithm's code with bitrate performance. This quantity is measured for LDC \emph{temporal} and \emph{spatial} operations during I/O data access and process between compression and decompression states.

\chapter{FBAR Compression Theory}\label{sect4}
In this section, we formulate the FBAR synthesis in form of theorems and proofs aligned with its foundation from Section~\ref{sect2.2}. Then, a 4D compression-decompression model is constructed by using FBAR operators and conversion functions on I/O data, as an improvement to the universal model from Section~\ref{sect2.1.3}. The model should exhibit DE predictable values. It also encloses the DE values in form of bits, from a compressed form, to its decodable or decompressed form, in a lossless manner.

\section{Reversible FBAR Compression Theorem and Proof} \label{sect4.1}

A reversible FBAR compression theorem begins with an assumption:

\begin{assumption} \label{assum1}
Suppose for every $x$ character input we have a righthand character $x'$, outputting a sequence $\mathcal{s}= (xx')$. We assume our machine compiler compiles data on 8-bit words. We also assume $x$ and $x'$ are from the ASCII table with a range of 0-255 characters. Let also any sequence of words be quantified by a length function $\ell$. Thus, the length of $\mathcal{s}$ in bits is $\ell(\mathcal{s})=$16 bits or 2 bytes ASCII.
\end{assumption}

Using Assumption~\ref{assum1} for a logical consequence, we hence submit a bit manipulation theorem on the given sequence $\mathcal{s}$, in terms of

\begin{theorem}\label{theo1}
Let the machine store a 2-byte binary input $xx'$, as information. Once we manipulate a pure byte sequence $\beta =$\texttt{11111111} to obtain $xx'$ by four single bit-flags, one $y$ character is produced denoting the $xx'$ content, equal to 8 bits. \end{theorem}

Theorem~\ref{theo1} is analogous to the notion of Hamming distance~\cite{hamming, murdocca}: ``the minimum number of substitutions required to change one string to another." However, our case has no relevance to Hamming error-correction characteristics, and concerns \emph{the number of fuzzy pairwise bit substitutions}. From Assumption~\ref{assum1}, the notion of storing any sequence $\mathcal{s}$ in Theorem~\ref{theo1} becomes valid in terms of an array quantifying the contents of $\mathcal{s}$, or by convention

\begin{remark}\label{remark1}
The machine stores data in form of an array $A$ on sequence $\mathcal{s}$. Either $x$ or $x'$ from $\mathcal{s}$, is of ASCII type measured in bpc, or entropy rate $H$.
\end{remark}

Hence, a bivector product $\mathbf{e}_{ij}$ on subfields by Lounesto (2001)~\cite{lounesto}, via its dual scalar element $h$ from Definition~\ref{def2.7}, self-contains and quantifies input $xx'$ in terms of

\begin{definition} \label{def1}
The $y$ character is stored in one of the rows $r$ in array $A_{r\times4}$, where $r$ satisfies a possible number of ASCII combinations for $xx'$. For all $y$, an $x\times x'$ combination produces $r= 256 \times 256 = 65536$ rows with a subspace scalar of $h^2\mathbf{e}^2_{ij}$.
\end{definition}

The $r$ for $xx'$ by Definition~\ref{def1}, further shows the following to be true, if and only if, an interactive proof on FBAR logic is presentable (Section~\ref{sect4.1.1}). Hence

\begin{proposition} \label{prop4.5}
The $y$ in $A$, holds single bit-flags that occupy the four columns $i$, $j$, $k$ and $l$, as biquaternion products from Proposition~\ref{prop0}, in the $r\times4$ array, or
\begin{subequations}\begin{align}
A_{r\times4} = A_{r\times (i \ j \ k \ l)} \ , \ \  (\stackrel{1}{i} \ \stackrel{2}{j} \ \stackrel{3}{k} \ \stackrel{4}{l}) &= h^2 \left( \begin{array}{cccc}
           \mathbf{e}_{1} & \mathbf{e}_{2}  & \mathbf{e}_{3} & \mathbf{e}_{4} \\
        \end{array}
      \right)^2 \notag \\ &= \left(
        \begin{array}{cccc}
           h^2\mathbf{e}_{12}^2 & h^2\mathbf{e}_{34}^2 \\
        \end{array}
      \right) \end{align} \noindent where the bit-flags field is displayed by \begin{equation}
\therefore xx' \longrightarrow y \in \{r\times (i \ j \ k \ l)\} = 65536 \times (1_x1_{x'} \ 1_x1_{x'} \ 1_x1_{x'} \ 1_x1_{x'}) \end{equation}\noindent and dimensionally measured by length $\ell$ as \begin{equation}
\therefore \ell\left(A_{r\times4}\right) = 65536 \times (2\,{\rm bits} \ 2 \, {\rm bits} \ 2 \, {\rm bits } \ 2 \, {\rm bits}) \end{equation}\noindent holding input data $xx'$, by a $y$ in the same field, in terms of    \begin{equation}
\therefore \ell\left(A_{r\times4}\right) =  65536 \  \times 8 \ \mathrm{bits} = 64 \ {\rm kilobytes}
\end{equation} \label{eq:4} \end{subequations}  where $y = 8 \ \mathrm{bits} \in A_{r\times4}=64 \, {\rm K}$, affirming that $A$ is static. If the bivector $\mathbf{e}_{12}^2 = -1$ for a pre-occupying character $x$, then its dual scalar is $h^2=-1$, otherwise $h^2=1$ for the post-occupying character $x'$ by $\mathbf{e}_{34}^2 = 1$. Both conditions determine the subspace property on each $xx'$ input as a superposing pair under compression. Thus, the compression products are orthogonally projective, positive and non-commutative.
\end{proposition}

\begin{proof} \label{proof1}
Suppose a $\varphi$ symbol denotes bit-flags for all $\mathbf{e}_{1234}$ in the $r\times4$ array. According to Assumption~\ref{assum1}, for the number of ASCII combinations on $r$, a total of $4\times 4 \times 4\times 4 = 256$ on $x$, and 256 on $x'$, satisfies 65,536 unique flag combinations. Its unit vector $\hat{{\bm v}}_\varphi$ whose coordinates are in one of the $1\times4$ array dimensions, has a length of 1\,bit with a scalar occupation. Thus, the $y$ character is stored in the $xx'$ intersection $\cap$, where $\varphi$ values meet. This gives $y$ a different content not equal to $\varphi$, but representing the exact location of $xx'$ in the sequence as well as content when $\varphi$ flags are \emph{translated}. We create a static \emph{translation table} to decode these flags based on where the $y$ character is stored i.e. the \emph{address} with a \emph{reference point}, or
\begin{subequations}
\begin{align} &\underbrace{\bigg{(}\underbrace{\overbrace{(\overbrace{4}^{\left(
                                                                \begin{array}{cccc}
                                                                  1_x & 1_x & 1_x & 1_x \\
                                                                \end{array}
                                                              \right)
} \times \overbrace{4}^{\left(
                                                                \begin{array}{cccc}
                                                                  1_{x'} & 1_{x'} & 1_{x'} & 1_{x'} \\
                                                                \end{array}
                                                              \right)}}^{8 \, {\rm bits} } )}_{1\mathrm{st} \, 16 \, \text{combinations}} \cap \underbrace{(\overbrace{4\times 4}^{8 \, {\rm  bits}})}_{2\mathrm{nd} \, 16}\bigg{)} }_{y \, \text {address}\,\cdots}  \notag \\
                                                              &\underbrace{\cap \bigg{(}\underbrace{(\overbrace{4\times 4}^{8 \, {\rm  bits}})}_{3\mathrm{rd} \, 16} \cap \underbrace{(\overbrace{4\times 4}^{8 \, {\rm  bits}})}_{4\mathrm{th} \, 16}\bigg{)}}_{\cdots \, y \, \text {address}}= 8 \,  \mathrm{bits} \label{eq:5a} \end{align}
\text stored as character $y$ in one of the 65,536 rows (\emph{prefix addresses}) representing
one of the four-dimensional combinations. These combinations are either 1st, 2nd,
3rd or 4th 16 combinations of bit-flags, for $xx'$. Equation~(\ref{eq:5a}) validates Proposition~\ref{prop4.5}
equations for specific address and flags configuration. Specifically,
\begin{align}
 &\forall \varphi \in A_{i,j,k,l} | \, \hat{{\bm v}}_\varphi = (1,0,0,0) \vee (0,1,0,0) \vee (0,0,1,0) \vee (0,0,0,1) \ , \
\\ &\text{such that}  \notag \\
&\forall y \in A_{1\times4} \subset A_{r\times4} \ | \ \varphi \cdot A_y = \lambda_{xx'} + \varrho_{xx'}= (i, j, k, l)_{xx'} + (i, j, k ,l)_{xx'}\label{eq:5c}\end{align}
\end{subequations}

So, the translation table is a precalculated (prefix) rows-by-columns file on bit-flags, giving a reference point for the stored character $y$. The reference is a specific bit-flag combination from the 65,536 possible rows, constituting the $y$ address. The bit-flags set $\varphi$, represents all 16 bits content of $xx'$, by manipulating a pure binary sequence $\beta=$ \texttt{1111 1111} recursively. This is shown in Eqs.~(\ref{eq:8}). The byte is manipulated to obtain the binary content of $xx'$. From Eq.~(\ref{eq:5c}), let this manipulation start with the left-most bit to the right-most bit, operating on the left-byte $\lambda$ and the right-byte $\varrho$ of sequence $\mathcal{s}$. This gives a $y$ product on the $xx'$ input, and is expressed by the following compression function:
\begin{align}
 &g\circ f:\mathcal{s}_{\mathrm{in}}\rightarrow A \notag\\
 &xx' \mapsto g(f(xx')) = f(y, \varphi) \, ,\notag\\
 &f(y, \varphi)= y \times \varphi_{xx'} = y (\varphi_x + \varphi_{x'}) = y_{\varphi_x} + y_{\varphi_{x'}} = y_{\varphi_{xx'}}
\end{align}

Let $f$ be a function composition that maps the contents of $\mathcal{s}$ to the contents of array $A$, holding the same $xx'$ contents via bit-flags $\varphi$. The bit-flags are occupied in form of character $y$. A two-variable function $f(y, \varphi)$ expresses the $y$ character carrying flags in array $A$, or its respective field ${\bf F}_y$. Its length is specified by
\begin{align}
 \ell(y_\varphi) = \varphi A_y =&\!\left[\left(
   \begin{array}{cccc}
    1_i & 1_j & 1_k & 1_l \\
  \end{array}
\right)_x \cdot \left ( \begin{array}{cccc}
    1_i & 1_j & 1_k & 1_l \\
  \end{array}
\right)_x \right] \notag \\ &+ \left[\left(
   \begin{array}{cccc}
    1_i & 1_j & 1_k & 1_l \\
  \end{array}
\right)_{x'} \cdot \left(
   \begin{array}{cccc}
    1_i & 1_j & 1_k & 1_l \\
  \end{array}
\right)_{x'}\right] \notag \\
= & \, \pmb\mho (x) + \pmb\mho (x')=\hat{{\bm\mu}}_x\|\mathbf{e}\|^2\left(2\,\hat{{\bm\mu}}_{\beta_x}+2\,\hat{{\bm\mu}}_{\beta_{x'}}\right) \notag \\ =& \, 4\,{\pmb\mho} + 4\,{\pmb\mho} = 8 \, \mathrm{bits} \label{eq:7}
\end{align}

In Eq.~(\ref{eq:7}), the dual $\pmb\mho$ manipulation method is of bivector type $|\mathbf{e}_{ij}|$ or $|\mathbf{e}_{kl}|$ from the norm $\|\mathbf{e}\|^2$ in Eqs.~(\ref{eq:2b}), and so its product $\mathbf{e} \cdot \mathbf{e}$, is of a data-decoding process. The method is derived to alter the partitioned bit-pair values from $\beta$ in ${\bf F}_y$ to obtain original data, as if its values are of \emph{Pythagorean identity}~\cite{leff} such that $8\,\hat{{\bm\mu}}_\beta =\mathop{\sum}_{n=1}^{8}\mathcal{s}(\sin^2\theta + \cos^2\theta)_n = 11111111\,$ stored as a byte, or
\begin{subequations}
\begin{align}
\begin{array}{c}
\hspace{-1.2cm} \left\{4 \, {\pmb\mho} \, {\rm   manipulations \ on \ } y = 8\,\hat{{\bm\mu}}_\beta = 11111111 = \beta  {\rm \ to \ obtain \ }x \  + \right.  \\
\left. 4 \, {\pmb\mho} \, {\rm manipulations \ on \ } y = 8\,\hat{{\bm\mu}}_\beta = 11111111=\beta{\rm \ to \ obtain  \ } x' \right\} = 8 \, \mho  \ ,
\end{array}
\end{align}
\begin{align}
\therefore \, &\varphi \left(\overbrace{\underbrace{(11)}}_{1 {\rm bit}_i \times 1 {\rm bit}_i}^{\frac{1}{4}y =  y_i} \overbrace{\underbrace{(11)}}_{1 {\rm bit}_j \times 1 {\rm bit}_j}^{\frac{1}{4}y =  y_j}\overbrace{\underbrace{(11)}}_{1 {\rm bit}_k \times 1 {\rm bit}_k}^{\frac{1}{4}y =  y_k} \overbrace{\underbrace{(11)}}_{1 {\rm bit}_l \times 1 {\rm bit}_l}^{\frac{1}{4}y =  y_l}\right)_{xx'} \!\!= y_{xx'} \notag \\ &=  8 \ {\rm bits} \to \underbrace{xx'}_{16 {\rm bits}}
\end{align}
\label{eq:8}\end{subequations}

Thus, to store more $y$ characters in the field of rows, $r\times4$, we establish a finite field $\mathbf{F}_y$ with $n$ elements. Therefore, $\mathbf{F}_y=\{y_1, y_2, \ldots, y_n\}$, represents a compression of sumset $\sum \mathcal{s}=(x_1x'_1 + x_2x'_2 + \ldots + x_mx'_m)$, achieving
\begin{equation}
\ell\left(\mathbf{F}_y\right) = \ell\left(\sum\limits_{i = 1}^m {x_i x'_i}\right)_{\mathrm{in}}\!\!\!\! \stackrel{\sum\limits_{r = 1}^{65536} A_{r\times4}}{\longrightarrow } \! \ell\left(\sum\limits_{i = 1}^n {y_i}\right)_{\mathrm{out}} \!= \frac{1}{2}\sum \mathcal{s}
\label{eq:9} \end{equation}

Equations~(\ref{eq:8}) and (\ref{eq:9}), show a fixed degree of double-efficiency over $y$, generating a decodable 50\% compression.
\end{proof}

\begin{proposition} \label{prop4.6}
The addressability of any original data is self-embedded in a grid file with 65,536 addresses, Eqs.~(\ref{eq:4}). For each double-character $xx'$ input, one specific address is occupied by a character $y$ output, Eqs.~(\ref{eq:8}). Translating the occupied address via a table whose row content returns original content, is by translating bit-flag combinations on a pure byte $\beta = \,$\texttt{11111111} obtaining $xx'$ from Eqs.~(\ref{eq:7})-(\ref{eq:8}).
\end{proposition}
\subsection{Interactive Proof} \label{sect4.1.1}
\begin{proof}
\label{proof2}
Suppose by default, we establish an ASCII combination on an $xx'$ input. The total number of combinations is 256 ASCII characters for $x$, and 256 ASCII characters for $x'$. Thus, $f(y) = 256 \, x \times 256 \, x' = 65536$ $xx'$. This gives an intersection of the combinations in total 65,536 8-bit addresses (64\,KB). We prove the intersections using logic and subspace topology on array $A$. The combinations are integrable in space where bits reside as stored and then manipulated. Let a compact Hausdorff space~\cite{scarborough} contain a pure byte sequence $\beta=$ \texttt{11111111} for manipulation to obtain $xx'$. The manipulation as a filter is done at a target point where space is locally compact. This results in compacted bits by associating all possible fuzzy pairwise bit manipulations, using bitwise operators OR \texttt{|}, and AND \texttt{\&} in code. Therefore, the manipulation `\texttt{11111111}' for a 2-byte content $xx'$ is `$\texttt{11111111} + \texttt{11111111}$'. The association of manipulation is via bit-flags giving left-byte intersected with right-byte. This association of two 8-bit sets gives an 8-bit output. Interactively, $\forall \ xx' \mapsto \{y\}_{i,j,k,l\in\mathbb{R}^4}$ we prove:
\begin{subequations}\begin{align}
&xx' \stackrel {\mathrm{store}}{\longrightarrow} A_y =\frac{xx'}{\|\mathbf{e}\|} = \frac{xx'}{\sqrt{\mathbf{e}\cdot \mathbf{e}}} = \frac{xx'}{\sqrt{4}} = \left(\frac{1}{\sqrt{2}}x_i \, \frac{1}{\sqrt{2}}x_j \, \frac{1}{\sqrt{2}}x'_k \, \frac{1}{\sqrt{2}}x'_l\right)  \\
\therefore \ &xx' \stackrel {\mathrm{store}}{\longrightarrow} \fbox{\bf{P}}  \stackrel {\mathrm{filter}}{\longrightarrow} f(y, \varphi) = \sqrt{|e_{12}|^2+ |e_{34}|^2}= \sqrt{\frac{1}{2}x^2 + \frac{1}{2}x'^2} = 8 \, {\rm bits}
\end{align}
\text Using the law of associativity in logic, the manipulative bits for $xx'$ appear as
\begin{align}
\therefore \, &\left(\{8\,\mathrm{bits}_x\} \in \underbrace{i \times j }_{2 \, {\rm dimensions}} \ni \{8\,\mathrm{bits}_x\} \right) \cap  \notag \\ & \left(\{8\,\mathrm{bits}_{x'}\} \in \underbrace{k \times l}_{2 \, {\rm dimensions}} \ni \{8\,\mathrm{bits}_{x'}\}\right)= \pmb\mho(\beta) \notag
 \\ &= 8 \ {\rm dual \ \mho \ manipulations} = 8 \, {\rm bits} = y \times \varphi_{xx'} \in \mathbb{R}^4
\label{eq:10c}\end{align}
\text and the address of $y$ for $xx'$ is found via 1-bit flags $\varphi$ operating on $\beta$, such that
\begin{align}
\therefore \ell(\varphi_{xx'}) &= \left(\{\underbrace{4_{\varphi \forall x}}_{4 \, {\rm flags}}\} \in i \times j \ni \{\underbrace{4_{\varphi \forall x}}_{4 \, {\rm flags}}\} \right) \cap \left(\{\underbrace{4_{\varphi \forall x'}}_{4 \, {\rm flags}}\} \in k \times l \ni \{\underbrace{4_{\varphi \forall x'}}_{4 \, {\rm flags}}\}\right) \notag \\ &= 4 \, {\rm bits}
\end{align}
\text where a storage field ${\bf F} $ covering all bit-flags $\varphi$ is quantified in terms of
\begin{equation*}
 \ell\left({\bf F}_\varphi \right) = \ell\left(\prod\limits_{e = 1}^4 4 \left(
                               \begin{array}{cccc}
                                 1_i & 1_j & 1_k & 1_l \\
                               \end{array}
                             \right)_e\right)
 = \prod\limits_{e = 1}^4 \left(4_i+4_j+4_k+4_l\right)_e
 \end{equation*}
 \begin{equation}
 \ \ \ = 16_1 \times 16_2 \times 16_3 \times 16_4 = {16^4{\rm  bits}} = 64 \, {\rm kilobytes}
 \label{eq:10d} \end{equation}
\label{eq:10}\end{subequations}
So, $i \times j \times k \times l$ represents a binary address as a four-dimensional flag or a byte in a spatial field $\mathbb{R}^4$ topology. Therefore, $xx'$ is distributed in 4 dimensions $i, j, k$ and $l$ by storing 1 $y$ character in the corresponding row. Now for the ``storage field," suppose we create a \emph{grid file} $\fbox{\textbf{G}}$ as a portable file with an empty space of 65,536 $i \times j \times k \times l$ rows, satisfying all possible combinations. This grid in specification should cover the $y$ field or ${\bf F}_y$ as well as bit-flag combinations. The $y$ field has a limit to store 1 to $n$, of $y$ characters corresponding to more inputs of $\mathcal{s}$. Let this be a sumset $\sum \mathcal{s}$ from Eqs.~(\ref{eq:8}). Thus, we specify these possible address combinations with a multi-sum on the available dimensions of 1-bit flags $\varphi$, covering field ${\bf F}_y$, or
\begin{equation}
A_y  \cap A_\varphi = {\bf F}_y \times {\bf F}_\varphi  = {\bf F}_y \otimes \sum\limits_{1\leqslant i,j,k,l\leqslant 16} \varphi_{xx'} =  \fbox{\textbf{G}} \label{eq:11}
\end{equation}

This, specifically builds our grid file multiplied with the sequences input $\mathcal{s}$ transposed matrix, as follows
\begin{subequations}
\begin{align}
 \fbox{\textbf{G}} = & \, \{ y_1 , y_2 , \ldots , y_m \}   \notag \\ &\times
\mathop{\left[
  \begin{array}{cccccccccccccc}
    1 & 2 & 3 & \cdots & 16 & 1 & 1 & 1 & 1 & 1 & 1 & 1 & \cdots & 16 \\
    1 & 1 & 1 & 1 & 1 & 2 & 3 & \cdots & 16 & 1 & 1 & 1 & \cdots & 16 \\
    1 & 1 & 1 & 1 & 1 & 1 & 1 & 1 & 1 & 2 & 3 & \cdots & \cdots & 16 \\
    1 & 1 & 1 & 1 & 1 & 1 & 1 & 1 & 1 & 1 & 1 & 1 & \cdots & 16 \\
  \end{array}
\right]}\nolimits_{\sum \mathcal{s}}^{\rm T} \notag
\\
   &= \{ 8 \, {\rm bits} , 8 \, {\rm bits} , \ldots , 8 \, {\rm bits} \} \times 65536_{ijkl} = \mathcal{s}_{\,\rm out}
      \end{align}
Evidently, the way $y_i$ is stored and configured, is later decodable for an LDD, where
\begin{align} \left \{ \mathcal{s}_{\,{\rm out}} \in \left. \fbox{{\bf G}} \, \right| 1 \, y \leq s \leq \ell \left( {\bf F}_y \right)  \right\} \ , \ \ell\left(\fbox{{\bf G}}\right) = \ \ell \left( {\bf F}_y \right) + 64 \, {\rm kilobytes} \label{eq:12c} \end{align} Ergo, by default, we occupy a $y$ = 8 bits for an $xx'$ = 16\,bits, since Eqs.~(\ref{eq:12}) intersected 8-bit flags for $x$ with 8-bit flags for $x'$ in the range of available rows in the grid file. So, if we exemplify a string of characters ``\texttt{resolved}," the sequence $\mathcal{s}$ becomes $\sum \mathcal{s} = x_1x'_1+x_2x'_2+x_3x'_3+x_4x'_4 = \texttt{resolved}$. Thus, the allocated bit-flags with respect to byte addresses display,
\begin{align}\sum \mathcal{s} = 64\,\mathrm{bits} \stackrel{\mathrm{store} \ \mathrm{as}}{\longrightarrow} \mathbf{F}_y &= \{y_1, y_2, y_3, y_4\}\! \subset \! A_{65536\times4} \notag \\ &= 32 \, {\rm bits} \in \|r\|=\frac{[1, 4]}{65536} \ {\rm rows}\ , \notag \end{align}\begin{align}
{\rm where} \ |\mathbf{F}_y| = 4
\ | y_i \in [1\times1\times1\times1 \ , \ 16\times16\times16\times16] \label{eq:12e}
\end{align} \label{eq:12} \end{subequations} \smallskip

\noindent In this case, the rows magnitude $\|r\|$ is up to $ \sqrt{1^2 + 4^2} = \sqrt{17}= 4.12$ out of 65,536 rows, since we have a cardinality of 4$y$'s occupying the array or storage field ${\bf F}_y$, with a specific address to decode 64 original bits. This gives a decodable 50\% compression plus a static size of 64 kilobytes.
\end{proof}

\subsection{FBAR Compression-Decompression Theorems} \label{sect4.2}

Following the Compression Proof \ref{proof2}, we specify an FBAR logic on I/O bit manipulations, delivering an FBAR compression theorem as follows:

\begin{theorem}\label{theo4}
Let the machine store a 2-byte binary input, $xx'$, as information. Once we manipulate a pure byte sequence `\texttt{11111111}' to obtain $xx'$ by bitwise and-or, negate \cite{boole}, and close it by fuzzy transitive closure~\cite{jacas}, one $y$ character is produced denoting the $xx'$ content, equal to 8 bits. \end{theorem}

From Theorem~\ref{theo4}, following Theorem~\ref{theo1} proof, we deduce an FBAR compression $\mathcal{C}$ corollary,

\begin{corollary}\label{corol1}
The four combined operations, pairwise and-or, negate and fuzzy transitive closure on sequence $\mathcal{s}$, give a 50\% fixed compression.
\end{corollary}

From Theorem~\ref{theo4}, we further deduce a complete FBAR compression $\mathcal{C}$ corollary

\begin{corollary} \label{corol2}
The FBAR four combinatorial operations on any sequence $\mathcal{s}$ from $\sum \mathcal{s} = (x_1x'_1 + x_2x'_2 + \ldots + x_mx'_m) $ give the same compression ratio 2:1, or 50\%.
\end{corollary}

Complementing Theorem~\ref{theo4} with a decompression $\mathcal{C}'$ theorem, we deduce a complete FBAR compression-decompression theorem or $\mathcal{CC}'$

\begin{theorem}\label{theo5}
Let a 1-byte $y$ output holding bit-flags of $\varphi$, represent a 2-byte binary input. This gives a 50\% fixed compression. During compression, byte addresses are stored in a static address as (4$\times$4$)\times($4$\times$4) single bit-flags, once accessed and decoded via a translation table on $y$ and $\varphi$, a lossless decompression is obtained.
\end{theorem}

\noindent and obtainable by

\begin{proposition} \label{prop4.11}
Sequence $\mathcal{s}$ is compressed into $y$ whilst reconstructable via bit-flags held by $y$, where $\ell(y) = \ell\left(\frac{xx'}{2}\right)=\frac{\mathcal{s}}{2}$ or 50\% fixed $\mathcal{C}$ on $xx'$ inputs. Once inputs are equally compressed via bit-flags in different locations of $y$, any $\mathcal{s}_{\,\mathrm{in}}$ is losslessly reconstructed. The specific address is where $y$ is stored amongst the 65,536 rows.
\end{proposition}

\section{4D Bit-Flag Model Construction} \label{sect4.3}

\subsection{I/O Operators Construction} \label{sect4.3.1}

To satisfy the practical application of our proof for Theorem~\ref{theo4}, we need to construct the algorithm components with relevant operators to conduct an LDC-DD.

\begin{proof}
Let a grid file component $\fbox{\bf G}$ be constructed according to Eqs.~(\ref{eq:12}), with a translation table $\fbox{\bf TT}$. To construct these two components, we intersect the dimensions of bit-flags $\varphi$ with each other in the $1\times4$ matrix, for each LDC I/O. We partition $\varphi$ according to Eqs.~(\ref{eq:7}), (\ref{eq:10c}) and (\ref{eq:10d}) in $i,j,k,l$ dimensions of $\fbox{\bf G}$, resulting \emph{four bivector operators}. Let those operators be \textbf{z}, \textbf{n}, \textbf{i}, \textbf{p}, where their paired combinational form, {\bf zn} and {\bf ip}, constructs in total the four dimensions. The static range $[\mathbf{zzzz}\ , \ \mathbf{pppp}]$ for $[1 , 65536]$ rows, stores only one dynamic character output $y$, per two-character input $xx'$:
\begin{subequations}\begin{equation}
\begin{split}
&\exists \, xx' =  16 \, {\rm bits} \in i \times j \times k \times l \ni y = 8 \, {\rm bits} \ ,   \\ &\mathrm{iff} \ h^2\mathbf{e}^2_{1234}=\mathbf{znip} \stackrel{\mathrm{operates \ on}}{\longrightarrow} xx' \  \mathrm{for \ compression} \, .\
\end{split}
\end{equation}
Notation $\ni$, here, follows the set membership notation $\in$, symmetrically, such that $y$ is identifiable in the $i,j,k,l$ bivector dimensions. Now, we adapt the grid file range on the rows from Eqs.~(\ref{eq:12}) to bit-flag operators, giving
\begin{align}
&\mathbf{znip} \stackrel{\mathrm{operates \ on}}{\longrightarrow} xx' = y \ , {\rm where} \notag \\
&\mathbf{znip} \in i \times j \times k \times l = [\mathbf{zzzz}\ , \ \mathbf{pppp}] \notag \\
&=[1\times1\times1\times1 \ , \ 16\times16\times16\times16]= [1 \ , \ 65536]
\end{align}
\label{eq:13}\end{subequations} where ${\bf z}$ is a zero or neutral operator, ${\bf n}$ a negate operator, ${\bf i}$ an impurity operator, ${\bf p}$ a purity operator, operating on bits. Suppose by definition, we apply \emph{and-or logic} ($\wedge\vee$) to the ${\bf z}$ operator on bit pairs, thus
\begin{subequations}
\begin{definition} \label{def4.12}
For all ${\bf z}$, ${\bf z}$ on a pair of bits in the $x$ or $x'$ binary, returns the same bits when and-or applied, where this applies to all remaining pairwise combinations, or
\begin{equation}
\mathbf{z}(0\wedge\vee0) = 00, \ \mathbf{z}(1\wedge\vee1) = 11, \ \mathbf{z}(0\wedge\vee1) = 01, \ \mathbf{z}(1\wedge\vee0)=10.
\end{equation}
\end{definition}

Then negation holds good for all pairwise bit combinations

\begin{definition} \label{prop4.13}
For all ${\bf n}$, ${\bf n}$ on a pair of bits in the $x$ or $x'$ binary, negates the bits when and-or applied, where this applies to all remaining pairwise combinations
\begin{align}
&\mathbf{n}(\neg(0\wedge\vee0)) = 11, \ \mathbf{n}(\neg(1\wedge\vee1)) = 00, \ \mathbf{n}(\neg(0\wedge\vee1)) = 10, \notag \\  &\mathbf{n}(\neg(1\wedge\vee0))=01.
\end{align}
\end{definition}

From the laws of Boolean algebra covering $\wedge$, $\vee$ and $\neg$ operators, the defined $\mathbf{z}$ and $\mathbf{n}$ operators in consequence hold good as axioms:

\begin{axiom} \label{axiom1}
Operator $\mathbf{n}$ is antecedent to negate all pairwise bit combinations. The consequent output is always the opposite of the given input.
\end{axiom}

\noindent and

\begin{axiom} \label{axiom2}
Operator $\mathbf{z}$ is antecedent to pass all pairwise bit combinations. The consequent output is always as same as the given input.
\end{axiom}

We can now make the definitions of Eqs.~(\ref{eq:14}) and (\ref{eq:15}) quite explicit as follows:

\begin{definition} \label{def4.14}
For all \textbf{z}, \textbf{z} on any pair of bits in the binary of $x$ or $x'$, returns the same bits, where this applies to all remaining pairwise combinations.
\end{definition}

\noindent and

\begin{definition} \label{def4.15}
For all \textbf{n}, \textbf{n} on any pair of bits in the binary of $x$ or $x'$, returns negated bits, where this applies to all remaining pairwise combinations.
\end{definition}
\label{eq:14}\end{subequations}

Now, suppose by definition, we apply \emph{transitive closure} $\stackrel{\curvearrowright}{\longrightarrow}$, to the ${\bf i}$ operator on bit pairs, acting as ``operates on" from Eqs.~(\ref{eq:13}), such that
\begin{subequations}
\begin{definition}  \label{def4.16}
For all \textbf{i}, bit pairs in the binary of $x$ or $x'$ is either 01 or 10. \textbf{i} closes with 1 for 01, and 0 for 10, or
\begin{equation}
\mathbf{i}(01)=0 \stackrel{\curvearrowright}{\longrightarrow} 1 = 1, \ \ \ \mathbf{i}(10)=1 \stackrel{\curvearrowright}{\longrightarrow} 0 = 0. \label{pdxeq:1}
\end{equation}
\end{definition}

\noindent and if applied to the ${\bf p}$ operator, we then have

\begin{definition}  \label{prop4.17}
For all \textbf{p}, bit pairs in the binary of $x$ or $x'$ is either 00 or 11. \textbf{p} closes with 1 for 11, and 0 for 00, or
\begin{equation}
\mathbf{p}(11)=1 \stackrel{\curvearrowright}{\longrightarrow} 1 = 1, \ \ \ \mathbf{p}(00)=0 \stackrel{\curvearrowright}{\longrightarrow} 0 = 0.
\label{pdxeq:2} \end{equation}
\end{definition}
\label{eq:15}\end{subequations}

The bit-pair operators from Eqs.~(\ref{pdxeq:1}) and (\ref{pdxeq:2}), generate conflicting binary products in terms of

\begin{paradox} \label{paradox1}
Operator $\mathbf{i}$ is coincident with operator $\mathbf{p}$ in bit-pair products that end with 1 or 0. The consequent output after closure is $\mathbf{i}(10) = \mathbf{p}(00)$, $\mathbf{i}(01) = \mathbf{p}(11)$.
\end{paradox}

\noindent which is in code, addressed by

\begin{solution} \label{sol4.18}
We first consider a pure sequence of bits, `\texttt{11111111}'. We manipulate the bit-pairs in the sequence with \textbf{ip}, then its result by \textbf{zn} combinations for original data $xx'$. Example~(\ref{4.31}) shall prove this.
\end{solution}

Combining the \textbf{z} and \textbf{n} axioms with the \textbf{i} and \textbf{p} solution, further delivers

\begin{definition}  \label{def4.19}
Combining \textbf{z}, \textbf{n}, \textbf{i}, and \textbf{p}, gives a combination of FBAR operators locatable in 4 dimensions making a possible 65,536 \textbf{znip} combinations. Each occupying $y$ in one of the 65,536 rows of the 4 dimensions, represents a \textbf{znip} combination corresponding to an $xx'$ input. These combinations are: \\

\noindent \textbf{ip} as an impure or pure pairwise bits' dimension providing 16 combinations:
\begin{subequations}
\begin{equation}\begin{split}
&{\bf iiii \ iiip \ iipi \ ipii \ piii \ iipp \ ippi \ ppii \ pipi \ ipip \ piip \ ippp} \\
&{\bf pipp \ ppip \ pppi \ pppp}\end{split}\end{equation}
\textbf{zn} as a zero or negate pairwise bits' dimension providing 16 combinations:
\begin{align}
&{\bf zzzz \  zzzn \ zznz \ znzz \ nzzz \ zznn \ znnz \ nnzz \ nznz \ znzn \ nzzn \ znnn} \notag \\
&{\bf nnnz \ nznn \ nnzn \ nnnn}
\end{align}
\end{subequations}
\end{definition}

Now, we establish a static solution using Eq.~(\ref{eq:11}) and the combinations above for an LDC operation

\begin{solution} \label{sol4.20} We build a grid file $\fbox{\bf G}$ based on these available combinations. The combinations constitute the finite field addresses $\mathbf{F}_y \times \mathbf{F}_\varphi$ as our static solution. This field gives a static number of rows and addresses for each compressed character $y_i$.
\end{solution}

\noindent and the dynamic solution continuing the previous LDC operation is,

\begin{solution} \label{sol4.21}
The $\fbox{\bf G}$ grows in file size as a dynamic field of compressed characters $y_i$ in length, when more than 1 character is stored beyond its static range. So, we store the  $1 \, y$ character by the program in one of the 65,536 rows representing original data (two characters) within the intersected columns of the dimensions.
\end{solution}

\noindent and its dynamic solution for an LDD operation, respectively, would be

\begin{solution} \label{sol4.22}
We access $\fbox {\bf G}$ by program code, invoking a comparator subroutine in our code. A $\mathcal{C}'$ is achieved by traversing the $i \times j \times k \times l$ dimensions as the grid's Hamming distance $d$ for each data read between fields $\mathbf{F}_y $ and $\mathbf{F}_{xx'}$ via $\mathbf{F}_{\varphi}$.
\end{solution}

So the question remains that: where should we establish distance $d$ between the prefix encoding and decoding levels of our $\mathcal{C}$ and $\mathcal{C}'$ operations?

To answer this question, we need to succeed in Solution~\ref{sol4.22}. For this, we recall Eqs.~(\ref{eq:10})-(\ref{eq:12}), and consider the Hamming distance $d$ definition by Symonds (2007)~\cite{symonds}, thus measuring our $d$ as follows:

\begin{definition} \label{def4.23} Hamming distance $d$ is measured between $n$-compressed characters stored with a minimum dynamic space of 64\,K = 65,536 static rows of $\fbox {\bf G}$, or as of Eq.~(\ref{eq:12c}), in $\mathbf{F}_{y} + \, 64\,$K, and their decodable 2$n$-input characters in a maximum static space as translation rows and columns in terms of $\{ \mathbf{F}_{xx'},  \mathbf{F}_{r}, \mathbf{F}_{\varphi}, \mathbf{F}_{y} \} = \fbox{\bf TT}$. Therefore, the compressed characters $y_i$ are decoded by flag vectors $\varphi$ as carried in their address of the finite field $\mathbf{F}_y \times \mathbf{F}_\varphi  = \fbox {\bf G}$.
\end{definition}

\noindent Thus

\begin{definition} \label{def4.24}
The number of coefficients in which they may differ in $\mathbf{F}_y$ to return original characters, is equal to the number of bit-pair manipulations $\pmb\mho(\varphi A_y)$ elicited from  Eqs.~(\ref{eq:8}) using {\bf znip} operators based on Eqs.~(\ref{eq:13})-(\ref{eq:15}).  \end{definition}

\noindent The last two definitions deduce the following

\begin{definition} \label{def4.25}
Distance $d$ conserves the finiteness of character I/Os via the code comparator, comparing characters between $\{ \mathbf{F}_{xx'}, \mathbf{F}_{y} \} \subset \fbox{\bf TT}$ and $\{ \mathbf{F}_{y} \} \subset \fbox{\bf G}$, and will not recover any randomness between two or more compressed characters located in $1$ or $> 1$ rows of $\mathbf{F}_y$ in $\fbox{\bf G}$.
\end{definition}

\noindent In addition,

\begin{definition}\label{def4.26} $d$ is 0 if at least 2 compressed $y$'s are stored in the same row address of $\mathbf{F}_y$ in $\fbox {\bf G}$, which represent 4 original redundant $x$'s decoded from $\fbox{\bf TT}$,
\end{definition}

\noindent and \smallskip

\begin{definition} \label{def4.27}
Distance $d$ is $>0$ if compressed characters are located in 2 up to $k$ = 65536 grid rows, representing some original characters redundant, otherwise, all as different in $\fbox{\bf TT}$.
\end{definition}

Upon Definitions~\ref{def4.23}-\ref{def4.27}, the following solution is emerged to satisfy a lossless $\mathcal{C}$-and-$\mathcal{C}'$ scenario of at least two compressed $y$'s stored in the $\fbox{\bf G}$ file component:

\begin{solution}\label{sol4.28} According to Eqs.~(\ref{eq:11}) and (\ref{eq:12}), and by Definitions~\ref{def4.12}-\ref{def4.19}, to decompress data losslessly, the 4 bit-flag operators~\textbf{znip}, operate on $2 \ y$'s in the finite field $\mathbf{F}_y$ in $\fbox {\bf G}$, manipulating the byte in form of bit-pairs, which return $2\ xx'$'s as their original. The original 4-character string is located in a $\fbox{\bf  TT}$ file with all distances prefixed for each pair of $xx'$. This transforms the grid to a distance of 0 when original characters are returned at the $\mathcal{C}'$ phase for each read row $r$. The code comparator compares the unique~\textbf{znip} row-by-column address from the table with the stored character in one of the grid file rows, from end-of-file to the file's header.
\end{solution}

Hence, benefiting from the Hamming distance propositions followed by their proofs in Symonds (2007)~\cite{symonds}, and Eqs.~(\ref{eq:7})-(\ref{eq:9}), the usage of~\textbf{znip} operators gives a maximum distance $d$ between vectors $\varphi_{\{y_1,y_2\}}$ and $\varphi_{{xx'}_1{xx'}_2}$, as $d(\varphi_{\{y_1,y_2\}}, \varphi_{{xx'}_1{xx'}_2}) = 16$ for the number of required bit-pair manipulations (\emph{prefix dual-coding}\,$\mho$). Therefore, as a bonus result, we immediately deduce the following two corollaries:

\begin{corollary} \label{cor4.29}
For two compressed characters $y_1$ and $y_2$, giving $d$ with respect to time $t$, as Hamming rate $R_\mathbbm{H}= d/t$ to search for a string match, results in a distance $d(\{y_1,y_2\}, {xx'}_1{xx'}_2)=0$ during decompression $\mathcal{C}'$. Let this distance be $d'$.
\end{corollary}

Corollary~\ref{cor4.29} for the compressed characters $y_1$ and $y_2$, further gives

\begin{corollary} \label{cor4.30}{\bf Part 1:} The $\varphi_{y_1y_2}$ and $\varphi_{{xx'}_1{xx'}_2}$ vectors in output do not differ in the number of coefficients when time $t$ allows $\pmb\mho (\{y_1,y_2\},$ ${xx'}_1{xx'}_2) = 16\stackrel {\varphi}{\longrightarrow } 0$ grid transformations, such that from Eqs.~(\ref{eq:11})-(\ref{eq:12}), we firstly establish the integral on distance $d$ at time $t$
\begin{subequations}
\begin{align}
& \forall R_\mathbbm{H} \in \mathcal{C} | R_\mathbbm{H} \! = \! \int_{\Delta t} \frac{d}{t} \, {\rm d} \, t \! = \! \frac{\|r\|}{\mathbf{t}}\int_{\min \|r\|_{\bf \emph{\bf{G}}}}^{ d\left(\emph{\bf{G}} , \emph{\bf{TT}}\right)} \! \! d(\{y_1,y_2\}, {xx'}_1{xx'}_2) \, {\rm d} \, d = \! \frac{\Delta d}{\Delta t}  \notag  \\
&\mathrm{where} \ \exists \ \{y_1 , y_2\} \in \mathbf{F}_y \, | \, \ell\left(\mathbf{F}_y\right) = 2 \mathrm{B} + 64 \mathrm{K }= \ell\left( \fbox{\emph{\bf{G}}} \right) \ , \ \mathbf{t} = 1 {\rm s} \ , \, \mathrm{and} \notag \\
& \exists \ \{{xx'}_1 , {xx'}_2\}  \in \mathbf{F}_{xx'} \, | \, \ell \left(\mathbf{F}_{xx'}\right) = 4 \mathrm{B} + 64 \mathrm{K} = \ell \left(\mathbf{F} \in \, \fbox{\emph{\bf{TT}}} \right) \, . \label{eq:17a}\end{align}

\noindent Given the occupied $\varphi_{\{y_1,y_2\}}$ and $\varphi_{{xx'}_1{xx'}_2}$ values in components set $\{ \bf{G} , \bf{TT} \}$ from a maximum number of rows $r$, between 4 original characters in {\fbox {\bf TT}} and 2 compressed $y$'s in {\fbox {\bf G}}, as distance $d$ at time $t$, in virtue of Eq.~(\ref{eq:7}) and Lemma~\ref{lemma1.4}, we deduce
\begin{align}
 \therefore \ & R_\mathbbm{H} =\! \int_{\Delta t} \frac{d}{t} \, {\rm d} \, t =\! \iint_{\frac{y_1}{2^{(8+8)}}}^{\max d} \dot{d} \, {\rm d} \, R= \frac{\|r\|}{\mathbf{t}}\!\int_{\frac{1}{k} \approx \, 0}^{r' \times \pmb\mho(\varphi A_y)}\!\!\! {\rm d} \, \pmb\mho  \ , \notag \end{align} \begin{align} & {\rm where} \, \begin{cases}
                                                                            \,  r' \in  [1, k]_{\emph{\bf{TT}}}  \\
                                                                            \,  t \in [0, \mathbf{t}]
                                                                          \end{cases} \hspace{4cm} \label{eq:17b}\end{align}

Thus, the maximum value of $r'$ rows holding the translation of the 2 compressed characters to their 4 original characters, is 2 rows in {\fbox {\bf TT}}, giving a Hamming rate
\begin{equation}
\therefore R_\mathbbm{H} =  \frac{2.23_{\bf G}}{k_{\bf TT}}\int_{\approx \, 0}^{2_{\bf TT} (4+4) } \frac{r \times \pmb\mho}{t} \, {\rm d} \, t =\frac{\Delta d}{\Delta t} \ge 0.0039 \ {\rm Bps} \ , \label{eq:17c}
\end{equation}

\noindent whereby Definition~\ref{def4.27}, constant $k$ is 65,536 grid rows. Rate $R_\mathbbm{H}$ measured in Bps, is equal to the change of number of $\pmb\mho$ bit-pair manipulations on the compressed data in an array of $r$ rows, relative to their original characters' rows $r'$, at time $t$. \\

\noindent {\bf Part 2:} After applying~\textbf{znip} operators in the Eq.~(\ref{eq:17c}) upper integral limit, denoting a future $\mathcal{C}'$ as stored addresses by $y$, we then compute its future rate returning $xx'$
\begin{align}
\ \exists \mathcal{C}' \ni R'_\mathbbm{H}  = \frac{\|r'\|}{\mathbf{t}} \int_{\max r'_{\bf \emph{\bf{TT}}}}^{ \|r\|_\emph{\bf{G}}}  \! \! \! {\rm d} \, \pmb\mho \notag & = \frac{\sqrt{\sum_{i=1}^{k} i^2}_{\bf G}}{\sqrt{k^2}_{\bf TT}} \int_2^{\frac{2.23}{k}} \frac{\pmb\mho \times r'}{t} \, {\rm d} \, t \notag \\ &= \frac{\Delta d}{\Delta t} \ge 295.6 \, {\rm Bps}  \label{eq:17d} \end{align}

Thus, the total Hamming rate performing $\mathcal{C}$ and $\mathcal{C}'$ is $\sum R_\mathbbm{H} = R'_\mathbbm{H} + R_\mathbbm{H} \in \mathcal{CC}'$. To evaluate Eq.~(\ref{eq:17d}), we measure the total distance $d$ between $\mathcal{C}$ and $\mathcal{C}'$ points via~\textbf{znip} as their $\pmb\mho$ string-match relation. Employing the Pythagoras' theorem gives an imaginary part $\imath=\sqrt{-1}$ for $\mathcal{C}'$, added to its conversed real part from $\mathcal{C}$ as follows \begin{equation*}
\begin{split}
\sum d(\mathcal{C}, \mathcal{C}') & = \pm 2\pi \! \int_{\sqrt{\mathcal{C}}}^{\sqrt{\mathcal{C}'}} \! \sqrt{\mathcal{C}} \: \: {\rm d} \, \mathcal{C}' =  2\pi\sqrt{\mathcal{C}^2-\mathcal{CC}'}= |8.88| \, \mho \ ;   \\
 &\vdash \forall \,\pmb\mho \in \mathbbm{C}\ell_4\left(\mathbbm{R}^{2^{\|\mathbf{e}\|}}\right) \left| \|\mathbf{e}\|\sum d(\mathcal{C}, \mathcal{C}')\right.  \\
 &= \left(\mathcal{C}\left(\sum \mathcal{s}_{\, \rm in} \right)\right)\pmb\mho \left(\mathcal{C}'\left(\sum \mathcal{s}_{\, \rm out} \right)\right)= \|dd'\|= \sqrt{4}\mathcal{C} (d) \\
& \ \ \ \, - \sqrt{4}\mathcal{C}'(d) = 2\mathcal{C} (d) \, + 2\mathcal{C}(d')= 8.88\|\mathbf{e}\| \approx 17.7 \, \mho \ , \end{split}
\end{equation*} therefore
\begin{align}
\|dd'\| & \vdash\left(\mathcal{C}\left(x_1x'_1x_2x'_2 \right)\stackrel{d}{\longrightarrow} \{y_1 , y_2\}\right)\mathbf{znip}\left(\mathcal{C}'\left(\{y_1 , y_2\}\right)\stackrel{d'}{\longrightarrow} \{xx'_1 , xx'_2\}\right) \notag \\
& = \pmb\mho \ \mathbf{znip} \ r' = \pmb\mho (\{y_1,y_2\}, {xx'}_1{xx'}_2) \stackrel {+}{\bigcap}  \pmb\mho ({xx'}_1, {xx'}_2) \ , \ {\rm ideally}  \notag \\
& = \pmb\mho \left(xx'_1 + xx'_2\right) = \pmb\mho (x_1x'_1x_2x'_2) =  \pmb\mho \left(\sum \mathcal{s}_{\, \rm{out}}\right)  = 16 + \imath\imath_{1234}^{\|\mathbf{e}\|} \notag \\ & = 17 \, \mho \ , \notag \end{align} and for the latterly-deduced result, we finally complement
\begin{equation}
\therefore \forall d \, \exists d' \in \Delta d \, | \, \{ \Delta d = (17.7-17)<1\, \mho  \} \Leftrightarrow \{ d \, \pmb\mho \, d' = 16\stackrel {\varphi}{\longrightarrow } 0\, \mho \} \ .  \label{eq:17e}
\end{equation}

\noindent Operator $\stackrel {+}{\bigcap}$ denotes a combinatorial string catenation and intersection of $y_i$ and $xx'_i$ elements, emitted into a (discrete) concrete sequence $\mathcal{s}$. Equations~(\ref{eq:17e}), show that there is no~\textbf{znip} manipulation $\pmb\mho$ to be made on $r' \in \fbox{\textbf{TT}}$ to obtain $xx'$ during $\mathcal{C}'$, except a $d \, \pmb\mho \, d'$ string-search, match and catenate relationship, thus giving $d' = 0\, \mho $.
\label{eq:17}\end{subequations}
\end{corollary}

Corollaries~\ref{cor4.29} and \ref{cor4.30} are possible, if and only if, component {\fbox {\bf TT}} is accessed, thereby char addresses compared by code and their data decoded. Such $\varphi$ vectors agree in all coordinates of the grid's $i \times j \times k \times l$ dimensions standing for an address during FBAR I/O operations. To keep matters simple, here is an example on a single compressed character $y$, spatially returning 2-original characters $xx'$

\begin{example} \label{4.31}
If $xx'=$ \texttt{01000000 00100100} $=@ \$ $, then by default \texttt{11111111} for $y$ is decoded, only when $y$ occupies $i\times j \times k \times l$ dimensions with a unique combination of operators. This combination is ${\bf ippp}\times {\bf niin}($\texttt{11111111}$)=$ \texttt{01111111  00010100} for $i\times j$, and this output intersected with the combination ${\bf znnn}\times {\bf znzz}($\texttt{01111111 00010100}$)=$ \texttt{01000000 00100100} for $k\times l$. The $y$ is stored in one of the rows out of 65,536 possible \textbf{z}, \textbf{n}, \textbf{i}, \textbf{p} combinations, which now returns characters $@ \$ $ = $xx'$ by code.
\end{example}

Therefore, we reconstruct the pair $xx'$ as our output. The output content is now equal to its original.
\end{proof}

Corollary~\ref{cor4.30} assumes continuity on the whole interval quantified as a spatial-temporal type. It further proves Corollary~\ref{cor4.29} via Eqs.~(\ref{eq:17d})-(\ref{eq:17e}), with an output sequence $\mathcal{s}_{\rm \,out}$ from Eqs.~(\ref{eq:9}) and (\ref{eq:12}). The spatial intervals in Eqs.~(\ref{eq:17a})-(\ref{eq:17c}) are delimited by the rows magnitude $\|r\|$ employed from Eqs.~(\ref{eq:12}) in {\fbox {\bf G}}, and its translation type $\|r'\|$ in {\fbox {\bf TT}}, where $r$ builds the maximum distance $d$ as the upper limit of integration. The temporal interval is given by $t$, and as conditioned in Eq.~(\ref{eq:17a}), goes with a maximum ideal time $\mathbf{t} = 1$\,s, i.e., processor(s) and memory being able to handle an occupied space $\ge 64$\,K I/O cases. The \emph{fluxion} $\dot{d}$ in Eq.~(\ref{eq:17b}), covers both temporal and spatial conditions from Eq.~(\ref{eq:17a}), and absorbs any covariance of a great time and distance change (rate $R$) into a small one, resulting their $\Delta$ forms in Eq.~(\ref{eq:17c}). The double integral, in this case, is absorbed into one succeeding integral as the future rate of $\mathcal{C}'$ in Eq.~(\ref{eq:17d}) proportional to the rate given by $\mathcal{C}$. Either rate is measured in bytes per second (Bps). It conveys to the implementation of prefix code and processing for pre-fuzzy bit-pair manipulations, which involves comparing addresses and rows for each set of compressed characters. For example, the minimum 295.6 Bps in Eq.~(\ref{eq:17d}), corresponds to a minimum ASCII table-read requirement as 2$\times$128 translatable characters or 256 Bps to decode compressed data in FBAR. Thus, an extra 39.6 Bps memory allocation is needed to conduct a full $\mathcal{C}'$. The Hamming distance used in Solutions~\ref{sol4.22} and \ref{sol4.28} relative to components $\fbox{\bf G}$ and $\fbox{\bf TT}$, is now subject to model construction for an I/O FBAR operation.

\subsection{Algorithm Components and Model Construction} \label{sect4.3.2}
Using the string sample ``\texttt{resolved}" from Eq.~(\ref{eq:12e}), and considering the {\bf znip} operators used in Example~\ref{4.31}, suppose we establish a \emph{translation table} $\fbox{\bf TT}$ like Table~\ref{tab2} to read $i \times j \times k \times l$ addresses (rows) containing 4\,$y$'s (32 bits) from $\fbox{\bf G}$. The $\fbox{\bf TT}$ file is fixed in size = 65,536 rows, and at least requires two key columns to translate $i \times j \times k \times l$, 1\,$y$ binary content to $xx'$ binary content and vice-versa.

At the $\mathcal{C}$ phase, the program writes $\fbox{\bf G}$ with ceratin characters, known as \emph{occupant chars} as output $y$ in a specific row number. This number must correspond to an address that returns original characters when the occupant char is decompressed.

At the $\mathcal{C}'$ phase, the program reads the grid file contents. From the occupant chars and the row address columns in $\fbox{\bf TT}$ or Table~\ref{tab2}, the program returns \emph{original chars} according to the `original char' column. Occupant chars are those characters residing in the $\fbox{\bf G}$ file. Once the program identifies the occupant char in a particular row number, outputs $xx'$ for that character according to the original char column from the $\fbox{\bf TT}$ file. This file in size is always 8\,MB for any reference point as a bit-flag for an occupant char corresponding to the original file. The matrix vectors and I/O process layout for the example looks like this
\begin{equation}
\begin{array}{*{20}cc}
  \ &  \\
   \{\overbrace{\texttt{resolved}}^{\text {original text}}\}  \stackrel{\text{read}}{\longrightarrow}  &  \fbox{\textbf{P}} \\
 \ & \updownarrow  \\
   \ & \fbox{\textbf{TT}} \\
\end{array}
 \stackrel{\text{write}}{\longrightarrow} \fbox{\bf G} = \{\texttt{a},\texttt{b},\texttt{c},\texttt{d}\} \in \left\langle \mathbbm{R}^4, \varphi \right\rangle
\label{eq:18}\end{equation}

\noindent and the decomposition of the $\fbox {\bf G}$ component after being constructed and written by program $\fbox {\bf P}$, is
\begin{align}
  &\left[
   \begin{array}{c}
     \texttt{a} \\
     \texttt{b} \\
     \texttt{c} \\
     \texttt{d} \\
   \end{array}
 \right] \! \otimes
 \left\langle \overbrace{\left( \begin{smallmatrix} {\bf zizp} \\ {\bf zizp} \\ {\bf zini} \\ {\bf zini} \end{smallmatrix} \right) \! \times \! \left( \begin{smallmatrix}  \bf n p n i  \\ {\bf npzp} \\ {\bf zpzp}  \\ {\bf zizi}   \end{smallmatrix} \right) \! \times \! \left( \begin{smallmatrix}  \bf z i n i  \\ {\bf zini} \\ {\bf zizp}  \\ {\bf zini}  \end{smallmatrix} \right) \! \times \! \left( \begin{smallmatrix}   \bf z i  z i \\ {\bf zpnp}  \\ {\bf zini}  \\ {\bf zinp}  \end{smallmatrix} \right)}^{\left\langle \mathbbm{R}^4, \varphi \right\rangle (\beta)} \right\rangle  =
 \left[
   \begin{array}{ccc}
     \emptyset & \cdots & \emptyset  \\
     \vdots & \ddots \\
     \texttt{d} & & \emptyset \\
     \emptyset & \cdots & \emptyset  \\
     \vdots & \ddots \\
     \texttt{c} & & \emptyset \\
    \emptyset & \cdots & \emptyset  \\
     \vdots & \ddots \\
     \texttt{a} & & \emptyset \\
     \emptyset & \cdots & \emptyset  \\
     \vdots & \ddots \\
     \texttt{b} & & \emptyset \\
   \end{array}
 \right]_{\ell} \notag \\
 &\text{where} \ \ell = 64 \texttt{K} + 4 \texttt{B} \ ,  \label{eq:19} \\
 &\text{and is a measured output denoting original data, such that} \notag
\end{align}
\begin{equation}
\begin{array}{*{20}cc}
  \ &  \\
  \ &  \\
\{\texttt{re},\texttt{so},\texttt{lv},\texttt{ed}\}  \! \longleftarrow \!\!\! & \!\!\!\!\!\! \fbox{\textbf{P}} \\
 \ & \!\!\! \!\!\! \updownarrow  \\
   \ & \!\!\!\!\!\! \fbox{\textbf{TT}} \\
\end{array}
\!\!\!\!\!\! \stackrel{\beta}{\longleftarrow}\left\langle \overbrace{\left( \begin{smallmatrix} 7 \\ 12 \\ 6 \\ 1 \end{smallmatrix} \right), \left( \begin{smallmatrix}  11  \\ 14 \\ 6  \\ 13   \end{smallmatrix} \right), \left( \begin{smallmatrix}  1  \\ 6  \\ 4  \\ 2  \end{smallmatrix} \right) , \left( \begin{smallmatrix}  13 \\ 13  \\ 15  \\ 7  \end{smallmatrix} \right)}^{\rm address} \right\rangle \! = \left\langle \texttt{a}, \texttt{b}, \texttt{c}, \texttt{d} \right\rangle
\label{eq:20}\end{equation}

In Eq.~(\ref{eq:19}), the constructed $\fbox{\bf G}$ by empty values $\emptyset$ with 65,536 rows (64\,K), has now an extra 4 bytes (chars) written to it by $\fbox{\bf P}$. Program $\fbox{\bf P}$ before writing to $\fbox{\bf G}$, accesses $\fbox{\bf TT}$ to write the 4 chars in specific locations denoting original data, given by (\ref{eq:18}) and (\ref{eq:19}). Later, for a decompression, the 4D function $\varphi$ in the $\fbox{\bf P}$ code, manipulates $\beta$ to obtain original data (its binary). This manipulation occurs when $\fbox{\bf P}$ refers to  $\fbox{\bf TT}$. The addresses of these 4 chars are identified in the $\fbox{\bf TT}$ file by $\fbox{\bf P}$ to reconstruct original data according to (\ref{eq:20}).

The left half of the input string `\texttt{resolved}' in Eq.~(\ref{eq:18}), is illustrated by a hypercube in Fig.~\ref{fig2}. The grid file is constructed according to Eqs.~(\ref{eq:12}), as well as ${\fbox {\bf TT}}$ for the code to access bit flags. Program ${\fbox {\bf P}}$ accesses the occupant chars \texttt{a}, \texttt{b}, \texttt{c} and \texttt{d} (known as $y$ in ${\fbox {\bf G}}$) to return the original chars at the $\mathcal{C}'$ phase. This phase is recognizable between components ${\fbox {\bf TT}}$ and ${\fbox {\bf P}}$ relationship `$\updownarrow$' in Eqs.~(\ref{eq:18})-(\ref{eq:20}).

\begin{table}[th]
\begin{center}
\begin{minipage}{\textwidth}
\caption{The FBAR Translation Table\label{tab2}}
{\footnotesize \begin{tabular}{@{}cccc@{}} \toprule
Row \# & Bit-flag address & 95 ASCII characters as ``occupant chars"   & Original char \\
&  &  representing the ``original char" column via   & \\
&  &   the ``bit-flag address" column & \\\hline
1\hphantom{00} & 1x1x1x1 &  abcde\ldots zABCDE \ldots Z123\ldots 0'$\sim$!@\$\ldots$>,<$ & $ {}^{\text{aa}}\hphantom{00}$ \\
2\hphantom{00} & 1x2x1x1 & abcde\ldots zABCDE \ldots Z123\ldots 0'$\sim$!@\$\ldots$>,<$ & \textyen${}^{\text{a}}\hphantom{00}$ \\
3\hphantom{00} & 1x3x1x1 & abcde\ldots zABCDE \ldots Z123\ldots 0'$\sim$!@\$\ldots$>,<$ & ${\bullet}^{\text{a}}\hphantom{00}$ \\
4\hphantom{00} & 1x4x1x1 & abcde\ldots zABCDE \ldots Z123\ldots 0'$\sim$!@\$\ldots$>,<$ & ${\copyright}^{\text{a}}\hphantom{00}$\\
\vdots \hphantom{00} & \vdots \hphantom{00} & \vdots \hphantom{00} & \vdots \hphantom{00}\\
65534 \hphantom{00}  & 16x16x16x14 & abcde\ldots zABCDE \ldots Z123\ldots 0'$\sim$!@\$\ldots$>,<$ & \"{y}\'{o}\hphantom{00}\\
65535 \hphantom{00}  & 16x16x16x15 & abcde\ldots zABCDE \ldots Z123\ldots 0'$\sim$!@\$\ldots$>,<$ & \"{y}\"{u}\hphantom{00}\\
65536 \hphantom{00}  & 16x16x16x16 & abcde\ldots zABCDE \ldots Z123\ldots 0'$\sim$!@\$\ldots$>,<$ & \"{y}\"{y}\hphantom{00}\\\hline
\end{tabular} }
\tablefootnote{a}{The actual translation table contents or \textbf{TT} file for an LDC/LDD access and management}
\tablefootnote{b}{The size of this component is approximately 8\,MB.}
\end{minipage}
\end{center}
\end{table}

\noindent Once the bit-flag addresses are identified by program ${\fbox {\bf P}}$ subroutines, thereby compared and interpreted in code, the original data is returned. This is done by bit manipulation $\pmb\mho$ from Eqs.~(\ref{eq:8}) on $\beta$ based on addresses to obtain original data (Example~\ref{4.31}).  The intersected addresses occupy in total 4 bytes for the 8-byte sample, since the number of stored chars in the ${\fbox {\bf G}}$ file is 4, or 4 bytes. Since an empty $\fbox {\bf G}$ is static in size, 64\,K of rows, all the chars stored with addresses also denote a static allocation. The addresses in this sample respectively are \texttt{7x11x1x13}, \texttt{12x14x6x13}, \texttt{6x6x4x15} and \texttt{1x13x2x7}. This is clearly specified in the four-dimensional vector space or storage subspace of bit-flags and addresses by Eq.~(\ref{eq:20}). It is denoted by dimensional contents between the angle brackets $\langle \rangle$ notation. In this example, the occupant chars occupying the specific addresses are shown as $\langle \texttt{a}, \texttt{b}, \texttt{c}, \texttt{d}\rangle$ in Eq.~(\ref{eq:20}).

\begin{figure}[th]
\begin{center}
\includegraphics[scale=1.1]{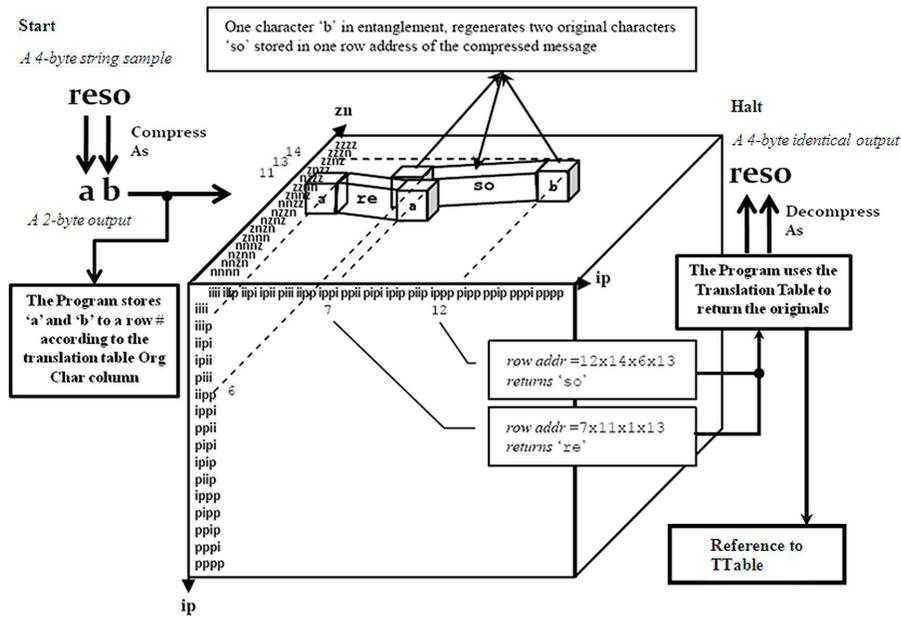}
\end{center}
\caption{An I/O $\mathcal{CC}'$ process on a `\texttt{reso}' string is given in a 4D grid constructor. This constructor shows a 50\% LDC with a DE state: the smaller inner-cubes in two places at the same time or ``characters in $\pm \pi$-entanglement." This model is a radix to higher DE-LDCs.\label{fig2}}
\end{figure}

The motive for choosing this hypercube (Fig.~\ref{fig2}) is anchored within the implementation of chars, being converted to binary as modeled back in Section~\ref{sect2.1.3}, thereby generating self-contained flags within an input char of the $\fbox{\bf G}$ grid. This results in 50\% pure compression, covering 2\,chars per entry. From Axioms~\ref{axiom1} and \ref{axiom2}, and Definitions~\ref{def4.14} to \ref{def4.26}, emitting Eqs.~(\ref{eq:18})-(\ref{eq:20}), we put all of the emerging 1-bit \textbf{znip} flags into unique combinations to obtain double-efficiency. We intersect them with other \textbf{znip}'s representing a second char input. Therefore

\begin{lemma} \label{lemma4.32}
Each character output is shared between $1\mathbf{ip}$ and $1\mathbf{zn}$ dimension, as a stored character. Containing $2 \, {\rm chars} \leftrightarrow 2 \mathbf{ip} + 2\mathbf{zn} = 4 \, {\rm dimensions}$ is done in 65,536 rows or addresses. A minimally $2 \, {\rm original \ chars}$ from $1 \, {\rm stored \ char}$ is decoded.
\end{lemma}

The analogy of Lemma~\ref{lemma4.32} is mappable to Moore's Law and Knowledge Management by Gilheany~\cite{gilheany}, stating: ``each time a bit is added to the address bus width, the amount of memory that can be addressed is doubled." Four-bit addresses allow the addressing of 16 bytes of memory, and in Lemma~\ref{lemma4.32}, are the 4 dimensions containing the 2\,chars or $xx'$. Eight bits allow the addressing of 256\,bytes of memory, whereas 16 bits can address 65,536 bytes of memory (and extra work is necessary to address 640 kilobytes of memory, as was the case on the early IBM PCs). In the FBAR case, an 8-bit $y$ can address a 16-bit $xx'$ in one of the 65,536 $\fbox{\bf G}$ file (portable memory) rows. Therefore, in terms of ``an address information sent immediately following the \emph{control byte} as a 16-bit word (65,536 possible addresses)"~\cite{smith}, here, is compressed as a 16-bit $xx'$ to an 8-bit $y$ character in $\fbox{\bf G}$ rows. So, $y$ in addition to a $\mathbf{znip}$ flag, plays the role of an 8-bit control byte for 65,536 possible addresses. Thus, we further deduce another lemma: \\

\begin{lemma} \label{lemma4.33}
Lemma~\ref{lemma4.32} gives a control double-byte $>$ a standard control byte for all intersecting addresses in $\fbox{\bf G}$. Since an input data is doubly compressed as $\ell(xx')=\ell(y)$, the static access of information in $\fbox{\bf G}$ is minimally, doubly faster than any other memory access when the compressed data is decoded.
\end{lemma}

The knowhow of these hypercube processes i.e., data access, compare, interpret after storage, is summarized in Table~\ref{tab3}, and implemented in our practical section, Section~\ref{sect5}, with performance results on the expected Hypothesis~\ref{hypo5.2} in Section~\ref{sect7}.

\subsection{Summary of Model and Theory } \label{sect4.4}

The 4D bit-flag model (hypercube) contains data for I/O transmissions. It maps contents in binary by intersecting their values in four dimensions using FBAR operators, suitable for any data type. The logic incorporated in this model, is of a combinatorial type, i.e. fuzzy, binary AND/OR logic. The rationale to the construction of this hypercube was to observe input characters, each pair of characters to be in two places at the same time. For imagery data types, an integer value is assigned instead, to satisfy an address representing two colors in two places simultaneously, out of the RGB color model for a 50\% LDC (recall Section~\ref{sect2.1.4}). Therefore, constructing $2^n$ memory addresses, in form of a grid file, gives a novel solution of how to compress data in doubles and pairs, losslessly. The fuzzy component of this logic is the middle point connecting binary with more possible states of logic. This connection of minimum to maximum number of states is defined in terms of an interrelated equation for all states of logic, and universal in all codeword representations. This was earlier introduced in Section~\ref{sect2.1.3}. By combining Eqs.~(\ref{eq:18})-(\ref{eq:20}) layout on the sample, we deduce the following components' paradigm. We later use this paradigm for the practical application of the algorithm to execute the operations held by components $\fbox{\bf TT}$, $\fbox{\bf P}$, $ \fbox{\bf G}$ and original file $\fbox{\bf O}$ as follows:
\begin{equation*}
 \begin{array}{*{20}ccc}
    \mathrm{Decompression} & \\
    & \\
   \! \! \! \! \! \! \fbox{\textbf{O}} \ \ \ \ \ \ \ \ \ \ \ \ &   \\
\! \! \! \! \! \! \uparrow \ \ \ \ \ \  \ \ \ \ \ \  &     \\
  \! \! \! \{\texttt{resolved} \} \stackrel{\mathrm{out}}{\longleftarrow} \fbox{\textbf{P}} \stackrel{\mathrm{in}}{\longleftarrow}  \fbox{\textbf{G}} &  \\
 \ \ \ \ \ \ \ \ \ \ \updownarrow &  &  \\
   \ \ \ \ \ \  \ \ \ \ \ \  \ \ \ \ \ \ \ \ \   \fbox{\textbf{TT}} \approx 8 \, \mathrm{MB} &  &
   \end{array}\! \! \! \! \! \! \! \! \! \left|
 \! \! \! \! \! \! \! \! \! \! \! \! \! \! \! \! \! \!  \begin{array}{*{20}cc}
    &  \mathrm{Compression} & \\
    & \\
  \ &  \,  \fbox{\textbf{O}}\longrightarrow \{\texttt{resolved}\} \! = \! 8 \: \mathrm{B} \ \ \ \  & \\
    & \ \downarrow^{\mathrm{in}} \ \ \ \ \ \ \ \ \ \ \ \ \ \ \ \ \ \ \ \ \ \ \ \ \ \ \ \ \ \ \ \  &   \\
     &  \ \ \ \ \ \ \ \ \ \ \   \fbox{\textbf{P}}\stackrel{\mathrm{out}}{\longrightarrow}  \fbox{\textbf{G}} \longrightarrow \{\texttt{a}, \texttt{b}, \texttt{c}, \texttt{d}\} \! = \! 4 \, \mathrm{B} \ \ \ \ \ \  & \\
 & \ \ \updownarrow \ \ \ \ \ \ \ \ \ \ \ \ \ \ \ \ \ \ \ \ \ \ \ \ \ \ \ \ \ \ \ \ \ \ \ &   &   \\
   \  & \ \ \fbox{\textbf{TT}} \approx 8 \, \mathrm{MB} \ \ \ \ \ \ \ \ \ \ \ \ \ \ \ \ \ \ \ \ \ \ \ \  &
   \end{array} \right.
\end{equation*}

\noindent Component $\fbox{\textbf{O}}$ as original file, is where the original text or string is located. The practical process and structure of all LDC/LDD components are given in Section~\ref{sect5}.

\chapter{FBAR Compression Practice}\label{sect5}
We implement the 4D model as the algorithm's prototype based on the theoretical aspects of FBAR logic on I/O data transmissions. This prototype should perform DE predictable values. To do so, DE values are enclosed as bits of information, from a $\mathcal{C}$ form to its decoded $\mathcal{C}'$ form in a lossless manner. Finally, we highlight certain details on the definiteness of future entropies supporting a growing negentropy, like Hyv\"{a}rinen {\it et~al}.~\cite{hyv}, proving a universal predictability, contrasting the popular Shannon's method of 1st order to 4th, inclusive of its general orders indeed.

\section{FBAR Components, Process and Test} \label{sect5.1}

To fully implement an algorithm, one must understand how it works in terms of its testable structure and model representation. This is illustrated in Fig.~\ref{fig3}. Furthermore, the algorithmic components must be introduced in terms of size, their process relationships, executables and data types. The $\mathcal{C}$ and $\mathcal{C}'$ phases of the algorithm, iteratively use the following components:  the $\fbox{\textbf{G}}$ as the grid file, $\fbox{\textbf{TT}}$ as the translation table file, and $\fbox{\textbf{P}}$ as the program source code for I/O executions. The $\fbox{\textbf{G}}$ file contains all compressed data representing the original characters. We call this the final compressed FBAR product or compressed file. We now introduce these components, their roles and functions for the $\mathcal{C}$ and $\mathcal{C}'$ implementation as follows: \\

\noindent As proven in theory, from Eq.~(\ref{eq:19}), the grid $\fbox{\textbf{G}}$ component consists of 8-bit blank entries or $\emptyset$ in 65,536 rows, providing a possible ASCII (256$\times$256)= 64\,K of static space for I/O data. The I/O data are processed by the $\fbox{\textbf{P}}$ component. This component deals with original contents $\fbox{\textbf{O}}$ component as original data which comprises of information built on one or more data types, given by the user. The $\fbox{\textbf{O}}$ component, is our input sample and should be tested for a lossless compression $\mathcal{C}$, as well as decompression $\mathcal{C}'$. The $\fbox{\textbf{TT}}$ component consists of combinatorial details of any data as a table on bit-flags, row number and occupant chars, available to $\fbox{\textbf{P}}$.

\begin{figure}[th]
\begin{center}
\includegraphics[scale=2.5]{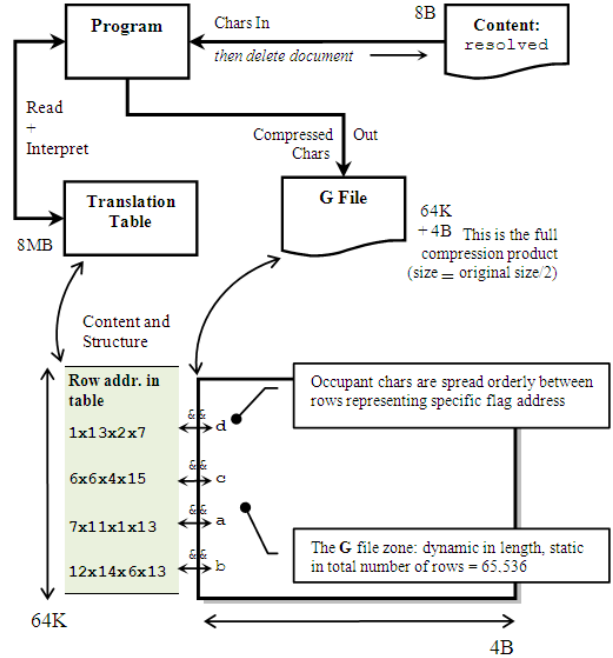}
\end{center}
\caption{The 4D logic constructor files with an 8\,B to 4\,B compression. \label{fig3} }
\end{figure}

The size of this table is static $\approx\,$8\,MB for its self-contained information. Program $\fbox{\textbf{P}}$ consists of lines of code to execute $\mathcal{C}\mathcal{C}'$ procedures. It accesses $\fbox{\textbf{O}}$ at the $\mathcal{C}$ phase, thereby constructs $\fbox{\textbf{G}}$ and puts occupant chars in a specific bit-fag and row number (a prefix address) as compressed data, using $\fbox{\textbf{TT}}$ information. At the $\mathcal{C}'$ phase, the same program accesses $\fbox{\textbf{G}}$, and by reading both its contents and addresses identified by $\fbox{\textbf{TT}}$, reconstructs $\fbox{\textbf{O}}$. This has been illustrated in Section~\ref{sect4.4}.

It is evident for each sample, at least one task $T$ is executed to perform compression parallel to decompression operations. Each conducted task allows one to evaluate the algorithm I/O's in terms of temporal measurement, here bitrate, as well as spatial measurement as bpc or entropy. Once implementation is resolved on this small scale (1\,$\fbox{\textbf{O}}$ file input), test cases are maximized or extended to the large, in number, and in scale for I/O data integration. This scalability of I/O's would guarantee the correctness of the code on FBAR logic requirements. For example, constructing an abstract release of a character reference column in the prefix $\fbox{\textbf{TT}}$ component, based on standard keyboard characters, including whitespace ``\ ", would not exceed 96 entries: 95 printable ASCII characters (decimal \#\,32-127) as shown in Table~\ref{tab2}, plus 1 control character. The latter is used to create a block or a jump character, indicated as \{\texttt{\slash a}, \texttt{\slash b},\ldots\}, between every \{1st, 2nd,\ldots\} 95-occupant char entry (or\,95\,$y$'s). At the $\mathcal{C}$ phase, this in total gives 95\,$\times$\,2\,=\,190 original char entries per block, and is denoted by the `$\mathcal{C}$(char)' column in Table~\ref{tab3}. Note that, at the $\mathcal{C}'$ phase, the program uses block chars to return \{1st, 2nd, \ldots, 95th\} pair of the original chars, hence forming words and sentences in the right order.

\begin{table}
\begin{center}
\begin{minipage}{\textwidth}
\caption{The FBAR I/O Character Process and Occupation Table \label{tab3}}
{\footnotesize\begin{tabular}{ccccc} \toprule
Row address & $\mathcal{C}$(char)\#; $\mathcal{C}_r$ &  Original chars; total & Occupant char & Size (bits) \\\hline
\multicolumn{1}{ >{\columncolor[rgb]{0.78,0.82,0.79}}c}{7x11x1x13}\hphantom{0} & \multicolumn{1}{ >{\columncolor[rgb]{0.78,0.82,0.79}}c}{1}\hphantom{0} 2:1=50\% & \multicolumn{1}{ >{\columncolor[rgb]{0.78,0.82,0.79}}c}{\hphantom{0} re \,  \hphantom{0000} 2 \ \ \ } \ \ \  &  \multicolumn{1}{ >{\columncolor[rgb]{0.78,0.82,0.79}}c}{a}  & 8 \hphantom{00} \\
\multicolumn{1}{ >{\columncolor[rgb]{0.9,0.9,0.9}}c}{12x14x6x13}\hphantom{0} & \multicolumn{1}{ >{\columncolor[rgb]{0.9,0.9,0.9}}c}{2}\hphantom{0} 2:1=50\% & \multicolumn{1}{ >{\columncolor[rgb]{0.9,0.9,0.9}}c}{\hphantom{0} so \,  \hphantom{0000} 4 \ \ \ } \ \ \    & \multicolumn{1}{ >{\columncolor[rgb]{0.9,0.9,0.9}}c}{b}  & 8 \hphantom{00} \\
\multicolumn{1}{ >{\columncolor[rgb]{0.9,0.9,0.9}}c}{6x6x4x15}\hphantom{0} & \multicolumn{1}{ >{\columncolor[rgb]{0.9,0.9,0.9}}c}{3}\hphantom{0} 2:1=50\% & \multicolumn{1}{ >{\columncolor[rgb]{0.9,0.9,0.9}}c}{\hphantom{0} lv \,  \hphantom{0000} 6 \ \ \ } \ \ \    & \multicolumn{1}{ >{\columncolor[rgb]{0.9,0.9,0.9}}c}{c}  & 8 \hphantom{00} \\
\multicolumn{1}{ >{\columncolor[rgb]{0.9,0.9,0.9}}c}{1x13x2x7}\hphantom{0} & \multicolumn{1}{ >{\columncolor[rgb]{0.9,0.9,0.9}}c}{4}\hphantom{0} 2:1=50\% & \multicolumn{1}{ >{\columncolor[rgb]{0.9,0.9,0.9}}c}{\hphantom{0} ed \,  \hphantom{0000} 8 \ \ \ } \ \ \    &  \multicolumn{1}{ >{\columncolor[rgb]{0.9,0.9,0.9}}c}{d}  & 8 \hphantom{00} \\
\multicolumn{1}{ >{\columncolor[rgb]{0.9,0.9,0.9}}c}{13x1x1x6}\hphantom{0} & \multicolumn{1}{ >{\columncolor[rgb]{0.9,0.9,0.9}}c}{5}\hphantom{0} 2:1=50\% & \multicolumn{1}{ >{\columncolor[rgb]{0.9,0.9,0.9}}c}{\hphantom{0} \: f   \hphantom{0000} 10 \ \ \ } \ \  \   & \multicolumn{1}{ >{\columncolor[rgb]{0.9,0.9,0.9}}c}{e}  & 8 \hphantom{00} \\
\multicolumn{1}{ >{\columncolor[rgb]{0.9,0.9,0.9}}c}{6x13x7x11}\hphantom{0} & \multicolumn{1}{ >{\columncolor[rgb]{0.9,0.9,0.9}}c}{6}\hphantom{0} 2:1=50\% & \multicolumn{1}{ >{\columncolor[rgb]{0.9,0.9,0.9}}c}{\hphantom{0} or   \hphantom{0000} 12 \ \ \ } \ \ \   & \multicolumn{1}{ >{\columncolor[rgb]{0.9,0.9,0.9}}c}{f}  & 8 \hphantom{00} \\
\multicolumn{1}{ >{\columncolor[rgb]{0.9,0.9,0.9}}c}{\vdots} \hphantom{0} & \multicolumn{1}{ >{\columncolor[rgb]{0.9,0.9,0.9}}p{1.1cm}}{ \ \ \ \ \ \vdots  }  \vdots \ \, \hphantom{000}  & \multicolumn{1}{ >{\columncolor[rgb]{0.9,0.9,0.9}}c}{\hphantom{0} \vdots \ \hphantom{0000}  \vdots \  \  \ } \,\hphantom{0}  & \multicolumn{1}{ >{\columncolor[rgb]{0.9,0.9,0.9}}c}{\vdots}  & \vdots \hphantom{00} \\
\rowcolor[rgb]{.0,.0,.0} \textcolor{white}{\emph{\textbf{the same as last}}} &  \textcolor{white}{\, \ \  \textbf{96}\hphantom{000} \ \ \textbf{1:1=0\%}} &  \textcolor{white}{ $\mathbf{\infty}$ \  \ \hphantom{000} \textbf{191}} \  \ \   &  \textcolor{white}{\textbf{\slash a} } &  \textcolor{white}{\textbf{16} \hphantom{00}} \\
\hline
\multicolumn{1}{ >{\columncolor[rgb]{0.9,0.9,0.9}}c}{8x12x8x12} \hphantom{0} & \multicolumn{1}{ >{\columncolor[rgb]{0.9,0.9,0.9}}c}{97}\hphantom{0} 2:1=50\% & \multicolumn{1}{ >{\columncolor[rgb]{0.9,0.9,0.9}}c}{\hphantom{0} 55 \  \hphantom{000} 193  \ \ \ } \  \ \  & \multicolumn{1}{ >{\columncolor[rgb]{0.9,0.9,0.9}}c}{a} & 8 \hphantom{00}\\
\multicolumn{1}{ >{\columncolor[rgb]{0.9,0.9,0.9}}c}{8x12x11x2} \hphantom{0} & \multicolumn{1}{ >{\columncolor[rgb]{0.9,0.9,0.9}}c}{98}\hphantom{0} 2:1=50\% & \multicolumn{1}{ >{\columncolor[rgb]{0.9,0.9,0.9}}c}{\hphantom{0} 5\$ \  \hphantom{000} 195 \ \ \ }  \ \ \  & \multicolumn{1}{ >{\columncolor[rgb]{0.9,0.9,0.9}}c}{b} & 8 \hphantom{00} \\
\multicolumn{1}{ >{\columncolor[rgb]{0.9,0.9,0.9}}c}{\vdots} \hphantom{0} & \multicolumn{1}{ >{\columncolor[rgb]{0.9,0.9,0.9}}p{1.1cm}}{  \ \ \ \ \ \vdots  }   \vdots \  \hphantom{0000} & \multicolumn{1}{ >{\columncolor[rgb]{0.9,0.9,0.9}}c}{ \vdots \ \hphantom{0000}  \vdots  } \,\hphantom{0}    & \multicolumn{1}{ >{\columncolor[rgb]{0.9,0.9,0.9}}c}{\vdots}  & \vdots \hphantom{00} \\\hline
\end{tabular} }
\tablefootnote{a}{The \textbf{TT} file is used for each \textbf{G} file-read on the compressed chars as `occupant chars' to return `original chars' in the process.}
\tablefootnote{b}{The table portrays the FBAR I/O products as original and compressed data per $\mathcal{CC}'$ operation. The program compares values in the highlighted cells to return a $\mathcal{C}$ or $\mathcal{C}'$ product.}
\end{minipage}
\end{center}
\end{table}

The process design and development of the algorithm is illustrated in Fig.~\ref{fig3}, with results listed in Table~\ref{tab3}. The process begins with encoding input data using a dictionary coder, after which a high and low-state prefix fuzzy-binary conversions occur for compression. Recalling Eqs.~(\ref{eqs:2})-(\ref{eq:4}), each level of planar projection, from a lower 2D-layer to its upper, forms a 4D \emph{quaternions plane}~\cite{arnold} or hypercube, as a 1-bit flag \emph{bi-vectors group}~\cite{lounesto}. This group in the hypercube has its own augment in identifying impure 01, 10, and pure states of 11 and 00 for each converted data byte. In return, for an LDD, the converted binary data are recalled via a translation table (Table~\ref{tab2}) as part of the dictionary or database represented by a set of occupying characters as the compressed version, denoting original data. We recall the original values from the \fbox{\textbf{TT}} file for each compressed occupying char, via a ``grid file" as a portable memory grid or \fbox{\textbf{G}} on single bit-flags to decompress data. The FBAR dictionary consists of data references parsed into the translation table, building a static size of flag information, later used by the program for string value comparisons (the highlighted cells in Table~\ref{tab3}).

\section{Methods of Double-Efficiency} \label{sect5.2}
We implement the algorithm in form of a prototype. The prototype presents the FBAR model and its encoding/decoding components for DE compressions.

\begin{figure}[th]
\begin{center}
\includegraphics[scale=2.1]{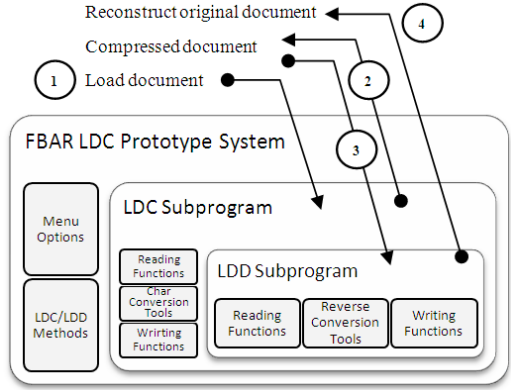}
\end{center}
\caption{The structural components of the FBAR prototype.\label{fig4}}
\end{figure}

As shown in Fig.~\ref{fig4}, the prototype representing program \fbox{\textbf{P}}, compresses data by loading a document sample. The program uses a memory grid file \fbox{\textbf{G}}, which is a portable file containing single bit-flags in 65,536 rows or addresses. The translation of addresses for original characters, is given in a \fbox{\textbf{TT}} file rows with a static size of 8\,MB, for any amount of input data manipulated by prefix code. The code interpreter decompresses data, once the flags are compared with the compression result. The decompression uses these prefix flags as compressed data, reconstructing the original document. All of these components, their processes and size are already proven in our theory, Section~\ref{sect4.3.2}. In the following sections, we implement the algorithm components with results and evaluate its DE claim on I/O samples.

\subsection{Algorithm Sample and Test} \label{sect5.3}
Assumption~\ref{assum1} holds good for the following algorithm:

\begin{proposition} From Assumption~\ref{assum1}, suppose for every $x$ character input we have a righthand character $x'$, its sequence appears as $\mathcal{s}_{\rm\, in}=xx'$. For a long sequence $\mathcal{s}$, we suppose a sumset $\sum \mathcal{s}=(x_1x'_1 + x_2x'_2 + \ldots + x_mx'_m)$ to be our information input. Our objective in the program is to compress $\mathcal{s}$ to single-byte characters or $\mathbf{F}_y=\{y_1, y_2, \ldots, y_n\}$ or recall the Proof on Proposition~\ref{prop4.5}.
\end{proposition}

\begin{algo}
Let program \fbox{\textbf{P}} get 2\,Characters from left-to-right of sequence $\mathcal{s}$. If \fbox{\textbf{P}} continues in taking 2 more Characters with respect to time $t$, it instantiates a series of tasks $T_1, T_2, T_3, \ldots, T_n$. These LDC tasks for each \emph{information processing cycle }on the sequence appear as
\begin{align*} &\sum T \times \sum \mathcal{s}=(\underbrace{x_1x'_1}_{T_1+} + \underbrace{x_2x'_2}_{T_2+} + \ldots + \underbrace{x_mx'_m}_{+T_n}) \stackrel {\mathrm{get}}{\longrightarrow} \fbox{\bf{P}}  \stackrel {\mathrm{store}}{\longrightarrow}  \fbox{\bf{G}}\stackrel{\mathrm{out}}{\longrightarrow} \mathbf{F}_y \\ &\text{such that, } \mathbf{F}_y=\{y_1, y_2, \ldots, y_n\} \ . \end{align*}

The processing cycles in our algorithm should follow
\begin{enumerate}[(1)]
\item\textbf{Input}: entering data into the program
\item \textbf{Processing}: performing operations on the data according to the \fbox{\textbf{TT}} file
\item \textbf{Output}: presenting the conversion results, in this case, the \fbox{\textbf{G}} file
\item \textbf{Storage}: saving data, or output for future use, in this case, the \fbox{\textbf{G}} file.
\end{enumerate}
\end{algo}

Now by applying the {\bf znip} operators (as 1-bit flags) on Binary Sequence Character $\beta($`1'$)= \texttt{11111111}$ as our default value in program \fbox{\textbf{P}}, to obtain the actual binary on each $xx'$ per task $T$, we then code our algorithm: \\

Algorithm~\ref{alg:sample1} comprises of LDC tasks, storing results in the \fbox{\textbf{G}} file. From the user, the program gets 2 chars, and inputs it from left-to-right of the file. The program by default contains a character `1' assigning the two concatenated input chars to the `1' (the customized $\beta$). Now, the program in line \# 6 generates a character representing the 2 chars in the correct row (corresponding row) according to ASCII standard for the same characters. This is further instructed in line \# 7 of the code, where the occupying row also represents an address of the compressed chars in the \fbox{\textbf{G}} file. The static translation table \fbox{\textbf{TT}} file is then used containing prefix addresses for every row out of 65,536 rows to translate, replacing one char with its original two chars for a 50\% compression. This is expressed in line \# 8, which gets the original 2 chars from \fbox{\textbf{TT}} per compressed char in \fbox{\textbf{G}} at the $\mathcal{C}'$ phase of the algorithm. The remaining lines of the algorithm just denote the opposite condition where the new string is again requested from the user to input from the start. \\

\SetAlFnt{\small}
\begin{algorithm}[H] \label{alg:sample1}

\KwIn{A set of LDC tasks and data conversions}
\KwOut{Storing LDC output to \fbox{\textbf{G}} file as a compressed string $y$ }
\Begin{
\SetAlFnt{\small}
\While{There are still input characters in $\mathcal{s}$}{
\ForEach{2 Characters from left-to-right of $\mathcal{s}$}{
Get 2 input characters $x_ix_i'$ \;
Pack 1-bit flags on Binary Sequence Character `1' = $x_ix_i'$ \;
Generate an occupant character $y_i$ according to the order of $x_ix_i'$ \;
Store 1 Occupant Character $y_i$ in the 1-bit flags row \# in \fbox{\textbf{G}} file\;

\uIf{$y_i$ and row \# is in the Translation Table }{
Continue getting the next 2 characters from left-to-right of $\mathcal{s}$ \;}
\uElse{Output the code for Pack as New String \;
Restart Packing as New String in \fbox{\textbf{G}} file \;
        New String = $y$ \;}
}
}
}
\caption{A lossless data compression sample}
\end{algorithm}
\vspace{8pt}

So, we can now initiate the $\mathcal{C}'$ phase of the algorithm in terms of Algorithm~\ref{alg:sample2}. This algorithm comprises of LDD tasks, reconstructing results in a new file after reading from the \fbox{\textbf{G}} file relative to the \fbox{\textbf{TT}} file. From the \fbox{\textbf{G}} file, the program reads 1 character from right-to-left and reads the row number in line \# 4, comparing it with the 65,536 available translations in the \fbox{\textbf{TT}} file (dictionary) in line \# 5. The program reconstructs a string of translated characters and adds up newcomer characters to its string to build a full word, or a sentence of the original information, in line \# 6-12, where \# 12 denotes that, Old Code = New Code. Then the program cleans up the memory at line \# 13.  \\

\SetAlFnt{\small}
\begin{algorithm}[H]\label{alg:sample2}
\KwIn{A set of LDD tasks and data conversions}
\KwOut{Decompressing \fbox{\textbf{G}} file as an LDD output $xx'$ }
\Begin{
\While{Reading characters row-by-row from end-of-file \fbox{\textbf{\emph{G}}}}{
\ForEach{1 character from right-to-left of string $y$}{
Read row \# \;

\uIf{Character $y_i$ is not in (row \# and Occupant Character) columns of  \fbox{\textbf{\emph{TT}}} file}{
New String $z$ = Get translation of Old Code \;
        New String $z$ = String $z$ + Character \;}
\uElse{Get translation of Old Code \;
Character $z_i$ = 1st or 2nd or \ldots or \emph{n}th 2 characters in
               String \;
}
               Replace Character with 2 new characters from the
   \fbox{\textbf{TT}} file \;
New String $z$ = $xx'$ \;
Delete temporary row \# and row characters \;
}
}
}
\caption{A lossless data decompression sample}
\end{algorithm}
\vspace{8pt}

Relevant to the example provided in Algorithm~\ref{alg:sample2}, we further particularize an LDD in Algorithm~\ref{alg:sample3}, which is equivalent to Algorithm~\ref{alg:sample2}.

Algorithm~\ref{alg:sample3} comprises of LDD tasks sampled from Algorithm~\ref{alg:sample2}, practicing a 16-byte (2-char) concatenation of the compressed chars $y_i =$ \texttt{d} then \texttt{c} then \texttt{b} then \texttt{a} (in lines \# 8, 10, 12, 14), reconstructing a 64-byte result after the concatenation operation is done. This reconstruction of original chars occurs in a new file after reading from the \fbox{\textbf{G}} file relative to the \fbox{\textbf{TT}} file.

From the \fbox{\textbf{G}} file, the program reads 1 char from right-to-left pre-positioned to a block char and reads the row number in line \# 4-6, comparing it with the 65,536 available translations in the \fbox{\textbf{TT}} file, in line \# 7. Finally, The program reconstructs a string of translated characters from line \# 8 up to line \# 14, and adds up (concatenate) newcomer characters to its string to reconstruct the full word as the output given in line \# 15, in this case `\texttt{resolved}'. \\

\SetAlFnt{\small}
\begin{algorithm}[H]\label{alg:sample3}
\KwIn{ A set of LDD tasks and data conversions}
\KwOut{Decompressing \fbox{\textbf{G}} file as an LDD output `\texttt{resolved}' }
\Begin{
\While{Reading characters row-by-row from end-of-file \textbf{ \fbox{\emph{G}}}}{
\ForEach{Last Block Character $y_i$}{
\uIf{Character $y_i$ is a Block Character}{
Read Character pre-positioned to Block Character \;
        Read row \# \;
        Get row address from \fbox{\textbf{TT}} file \;
\uIf {Character $y_i$ =`\emph{\texttt{d}}' and row address = `\emph{\texttt{1x13x2x7}}'}{
Output String =`\texttt{ed}'  \;}
\uElseIf {Character $y_i$ =`\emph{\texttt{c}}' and row address = `\emph{\texttt{6x6x4x15}}'}{
Output String  =`\texttt{lv}'+`\texttt{ed}' = `\texttt{lved}'  \;}
\uElseIf {Character $y_i$ =`\emph{\texttt{b}}' and row address = `\emph{\texttt{12x14x6x13}}'}{
Output String =`\texttt{so}'+`\texttt{lved}' = `\texttt{solved}'  \;}
\uElseIf {Character $y_i$ =`\emph{\texttt{a}}' and row address = `\emph{\texttt{7x11x1x13}}'}{
Output String =`\texttt{re}'+`\texttt{solved}' = `\texttt{resolved}'  \;}
\uElse{Print no data or null compressed \;
}
}
\uElse{Print no block character in range \;}
}
}
}
\caption{An LDD sample that returns the `\texttt{resolved}' string}
\end{algorithm}
\vspace{8pt}

\subsection{Maximum LDC/LDDs} \label{sect5.4}

\newcommand{\HRule}{\rule{\linewidth}{0.1mm}}

Maximum LDCs must respectively satisfy Hypotheses \ref{hypo5.2} and \ref{hypo5.3} from below, as \emph{midpoint} and \emph{maximum} LDCs for the 4D model implementation. These are the updated versions of the hypotheses H.4 and H.5 by Alipour and Ali (2010)~\cite{alipour10}, which cover discrete interval values of Eqs.~(\ref{eq:25}) and~(\ref{eq:26}), later introduced in Section~\ref{sect6}. One of which as the most radical to our model implementation is Hypothesis \ref{hypo5.2}. This hypothesis subsists on the algorithm's predecessors, which deal with minimum and middle-point LDCs. The minimum LDCs mainly project onto the dynamic memory allocation points which are of interest when optimization of the algorithm is concerned. For example, during \fbox{\textbf{G}} and \fbox{\textbf{TT}} I/O operations, the dynamic size of \fbox{\textbf{G}} must be managed by dynamic \emph{read}-and-\emph{write} of $y$'s as minimum compression, maintaining a maximum 50\% compression on all characters when only 1\fbox{\textbf{TT}} file is being read. We implement maximum LDCs by self-containing the static memory allocation points as mid-points in a \fbox{\textbf{TT}} file bit-flag addresses, specified back in Sections~\ref{sect4.1} and \ref{sect4.2}, as follows: \\

The following represents ``midpoint LDCs" on the version-to-version 4D model \smallskip

\begin{hypothesis} \label{hypo5.2}
A sequence of bit-flags representing double-efficient compressed data in FBAR, once reused by its translation table adjacent to other purely compressed data, results in a decompressed message.
\end{hypothesis}

\noindent whereas its null hypothesis would be \smallskip

\noindent \HRule \vspace{2pt}
{\noindent \textbf{Hypothesis 4.1}$_{\bf 0}$ The sequential recall and reuse of bit-flags from memory/grid, is firstly minimum-compression dependent, and secondly, unachievable for an identical data reconstruction.} \HRule \vspace{6pt}

The following represents ``maximum LDCs" on the version-to-version 4D model

\begin{hypothesis} \label{hypo5.3}
A sequence of compressed data in form of four-dimensional 1-bit flags, when partitioned into memory or confined in information space/grid, results in a maximum LDC possible $\ge$ 87.5\% with optimal bitrates.
\end{hypothesis}

\noindent whereas its null hypothesis would be \smallskip

\noindent \HRule \vspace{2pt}
{\noindent \textbf{Hypothesis 4.2}$_{\bf 0}$ The compression of any data length into one single-byte is firstly minimum-compression dependent, and secondly, unmanageable and irreversible for data reconstruction like Hypothesis~\ref{hypo5.2}.} \HRule \\[0.1cm]

According to Sections~\ref{sect4}-\ref{sect5.3}, Hypothesis~\ref{hypo5.2} is by now achieved, which addresses the 4D model implementation, independent of its null hypothesis limitations due to the 4D model characteristics i.e., the \fbox{\bf TT} and \fbox{\bf G} components and their relationships. These dimensional relationships on I/O $\mathcal{CC}'$ data are discussed as follows: \\

\noindent For an 87.5\%, obviously, the column with 96 characters will not change, however, the `$i \times j \times k \times l$' column in its configuration becomes `$i \times j \times k \times l$ \ \ $i \times j \times k \times l$', and the last column with 2 characters, becomes 8 characters, since the cubic representation of the `1st $i \times j \times k \times l$' with the `2nd $i \times j \times k \times l$' has a second \emph{non-commutative symmetric} format: `2nd $i \times j \times k \times l$' with the `1st $i \times j \times k \times l$', giving four distinct addresses simultaneously. So, for the former, this means, 2 original chars result in 1\,char in compression (2:1 or 50\%), and for the latter, 8 original chars result in 1 compressed char ($100\% - 12.5\% = 87.5\%$ or 8:1 bytes) as an `occupant char' (see, Table~\ref{tab2}), occupying a row in the compressed file \fbox{\bf G} in Fig.~\ref{fig3}. The symmetry `2nd $i \times j \times k \times l$' with the `1st $i \times j \times k \times l$', altogether, gives four distinct double-char addresses simultaneously, i.e., an 8:1 LDC. This satisfies $65536^4$ \textbf{TT}ables = $1.84 \times 10^{19}$ unique combinations, or, 16 exabytes (EB) of grid rows. In case of columnar symmetry in two translation tables, $65,536^2$ = 4.1\,GB, handles the 16\,EBs when column values are co-intersected by a comparator matrix in our code, residing in the LDD subprogram comparator (Sections~\ref{sect4.3.1} and \ref{sect4.3.2}). So, four 64\,K grid row combinations, handle the same EB values in four parallel tables. This requires a complex compression matrix coding as symmetric and antisymmetric access of \fbox{\bf TT} data on current machines equipped with dual CPUs. The reason is having an optimized version by creating multi-threads on the four parallel \fbox{\bf TT} files (\emph{prefix data}) per LDC operation. To this account, we pose a formulation:

\subsection{Complex Matrix Coding} \label{sect5.4.1}

Let $\mathcal{C}_{\rm matrix}$ be a variable for a compression matrix with decision nodes in FBAR code on either spatial or temporal measurements made by Alipour and Ali (2010)~\cite{alipour10}, where its values to process \fbox{\bf TT} data for \fbox{\bf G} read/write operations would establish
\begin{equation} \label{eq:21}
\mathcal{C}_{\rm matrix}  \propto \mathcal{C}_{\max } \left( {\beta {\rm  },{\rm  }t_L } \right) \ ,
\end{equation}

\noindent where $\mathcal{C}_{\rm \max}$ is a data function for the highest possible layer of lossless compression (HLLC) by FBAR, and $t_L$ is the time taken to process \fbox{\bf TT} and \fbox{\bf G} files for an I/O binary sequence $\beta$. This deduces
\begin{equation} \label{eq:22}
\therefore \mathcal{C}_{\rm matrix}  = M \times \mathcal{C}_{\max } \ .
\end{equation}

In fact, time length $t_L$ corresponds to the current time where present pseudocode decision points would not exceed the limit of `if-else statements', even in case of extending them into the 16\,EB scenarios. In other words, a 64-bit microprocessor, in principle, handles at most, 18\,EBs of space~\cite{ebbers}, if based solely on 1\,{\bf TT}able. So, we program 4{\fbox{\bf TT}}'s to just have a 32\,MB table with our FBAR package. In our later results in Section~\ref{sect7}, Eq.~(\ref{eq:21}) becomes evident in terms of {\it cyclomatic complexity M}~\cite{mccabe}, with a conjecture of just including the concatenation operator `+' in its stateful extension. This makes FBAR as efficient as possible in its $\mathcal{CC}'$ product results. The more {\fbox {\bf TT}}s included in Eq.~(\ref{eq:21}), the more complexity or decision nodes of code loops. For the `if-else' statements satisfying a $\mathcal{CC}'$ (Algorithms~\ref{alg:sample1} and \ref{alg:sample3}), we deduce that an efficient complex matrix code dedicated to the 4{\fbox {\bf TT}}-read per {\fbox{\bf G}}-write content, requires $M =M_{\rm previous} + 3$ (as sampled below in Algorithm~\ref{alg:sample4}). The reason compared to the previous pseudocodes is that, the concatenation `+' operator, triples on decision points in the new complex version in terms of Algorithm~\ref{alg:sample4}. Apart from each short-circuit `AND' operator adding a 1 to the $M$~\cite{watson}, the number of if-statements shall remain the same for the LDC/LDD codes (recall Algorithm~\ref{alg:sample2}). The simulated results on Eq.~(\ref{eq:21}) are listed in Table~\ref{tab4}.

Algorithm~\ref{alg:sample4}, below, comprises of LDD tasks and is analogous to Algorithms~\ref{alg:sample2} and~\ref{alg:sample3}, but with the ability to manage large amounts of reconstructable data through complex coding ($\mathcal{C}_{\rm matrix} $ code). The algorithm hypothetically returns 64\,bytes (8\,original characters) represented by 1 single byte (1\,compressed character), as standardized in the \fbox{\textbf{TT}} file, after reading all translatable bit-flag combinations in 4$\times$65536 conjoint rows (Fig.~\ref{fig5}).  \smallskip

\SetAlFnt{\small}
\begin{algorithm}[H]\label{alg:sample4}
\KwIn{A set of LDD tasks and data conversions}
\KwOut{Decompressing \fbox{\textbf{G}} file as an LDD output `\texttt{resolved}' }
\Begin{
\While{Reading characters row-by-row from end-of-file \fbox{\textbf{\emph{G}}}}{
\ForEach{Last Block Character $y_i=y_n$}{
\uIf{Character $y_i$ is a Block Character}{
Read Character(s) pre-positioned to Block Character \;
        Read row \# \;
        Get row address from \fbox{\textbf{TT}} file \;
\uIf {Character $y_i$ = 1 Occupant Character and row address = `$i\times j\times k\times l$'+`$i\times j\times k\times l$'+`$i\times j\times k\times l$'+`$i\times j\times k\times l$'}{
Output String = `1st 2 characters' + `2nd 2 characters'
          + `3rd 2 characters' + `4th 2 characters' =
          `8 original characters' \;
          }
\uElse{Print no data or $\emptyset$ compressed \;}
}
\uElse{Print no block character in range \;}
}
}
}
\caption{An 87.5\% LDD sample: 1\,compressed character returning 8\,original characters}

\end{algorithm}
\vspace{8pt}

The character reconstruction method, from line \# 8 and 9, results in a new file after reading from the \fbox{\textbf{G}} file relative to its \fbox{\textbf{TT}} file. From the \fbox{\textbf{G}} file, in line \# 2-5, the program reads $n \geq 1$ character from left-to-right until it reaches a block character $y_n$, and then reads its row number in line \# 6. It compares the character(s) located between two block characters $y_n$ and $y_{n-1}$, with the 65,536 available translations in the \fbox{\textbf{TT}} file, in line \# 5-7. The program reconstructs a string of translated characters and adds up newcomer characters to its string to reconstruct the full word, sentence or original information, from line \# 8-10. The program cleans up the memory, being well-aware that if any input character is not given according to line \# 11 to the algorithm, null ($\emptyset$) is returned, specifying no code block in range, in line \# 14.

\subsection{Component Relationships} \label{sect5.4.2}
To perform double, or even quadruple efficiencies, the program must refer to a \fbox {{\bf TT}} file comprised of bit-flag addresses from the 4D model, with their corresponding original chars as well as occupant chars in code representing their original char positions. These I/O references have been shown in Table~\ref{tab3}. For quadruple efficiencies, the comparator's if-else statements are expanded in terms of reading multiple {\fbox{\bf TT}}s in parallel. This returns for each unique address, 4 original chars denoting a 75\%, and 8 original chars denoting an 87.5\% compression. The 4 original char version, requires the complex matrix to read data from 2{\fbox {\bf TT}}s correspondingly, since each {\fbox{\bf TT}}, according to the 50\% LDC, returns 2 original chars per 1 compressed char as an occupant char. Therefore, 2{\fbox{\bf TT}}s return 4 original chars at the $\mathcal{C}'$ phase. So, for performing the $\mathcal{C}_r=$ 8:1, or the 8 original char version, we require 4{\fbox{\bf TT}}'s for each data-read by the comparator to succeed 4 address translations, or
\begin{equation}
 4{\fbox{\bf TT}} \times  4 \ occupant \ char \ translations = 8 \ original \ chars.  \label{eq:23}\end{equation}

\noindent In Eq.~(\ref{eq:23}), the number of translations is $n$ in $n${\fbox{\bf TT}}, and applies to an $n$-hypercube, $2^n\,n!$, by Coxeter \emph{et al.} (2006)~\cite{coxeter}, as $2^{16}\times n$ $=16^{\rm 4D}\times n{\fbox{\bf TT}}$ flag combinations.

\begin{figure}[th]
\begin{center}
\includegraphics[scale=2.7]{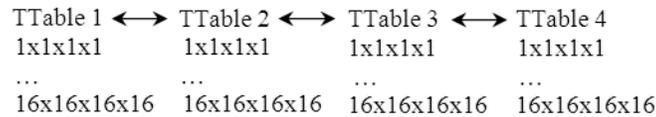}
\end{center}
\caption{Translation tables in parallel intersections for an 8:1 LDC. The $\mathcal{C}_{\rm matrix}$ code accesses data for read + write operations from all {\fbox{\bf TT}}s with a configuration of addresses only, building a (16$\times$16$\times$16$\times$16)$\times$4 or a $16^{\rm 4D}\times$4 matrix for a {\fbox{\bf G}} file's $\mathcal{CC}'$ write operation.\label{fig5}}
\end{figure}

In Fig.~\ref{fig5}, Eq.~(\ref{eq:23}) is illustrated in form of a flag-char relationship diagram, corresponding to bit-flag values (char address) of the 4D model. As a result, we return the original chars exactly as expected at the $\mathcal{C}'$ phase. For highest DEs, we therefore extend the number of \textbf{znip} columnar combinations from the previous \fbox{\bf {TT}} in terms of row-by-row intersections. This is called a 4\,table-based algorithm. It delivers double DEs, and thereby, quadrupled efficiencies as well. We described this in terms of fulfilling 4.1\,GB and 16\,EB combinations in the above paragraphs, respectively. In other words, on all occasions, the program's interpreter/comparator matrix must be able to handle 1, 2 and 4\fbox{\bf {TT}}s for all intersections between them, needing just 8, 16 and 32\,MB static size on an x86 machine, instead of the EB barrier denoting no columnar interactions whatsoever. For example, an intersection of \texttt{1x1x1x1} with \texttt{1x2x1x1} with \texttt{16x16x16x15} with \texttt{1x4x1x1}, from {\fbox{\bf {TT}}}s 1-to-4 in Fig.~\ref{fig5}, returns ${}^{\text{aa}}$\textyen${}^{\text{a}}{\bullet}^{\text{a}}{\copyright}^{\text{a}}$ original chars. Hence, the length of 64 bits is thus \emph{self-contained} and \emph{fixed} by the program's comparator efficiently, using just 8 bits out of the 32\,MB of the traversed tabular space, denoting an 87.5\% compression.

\chapter{Simulation Results, Contribution and Analysis} \label{sect6}

\section{Contribution}
The main contribution in this paper is presenting a new model on self-embedded flags from Section~\ref{sect4.3.2}. It allows an LDC algorithm to conduct a new encoding technique i.e., doubling the efficiency between two points of data transmission (DE). The key component of an FBAR algorithm is the \fbox{\bf TT} file or translation table whereby double-efficiency is conducted. Moreover, the \fbox{\bf G} file, as another component, holds sizes denoting double compression ratios after each full I/O write of contents. This established a key difference in techniques, observed between the FBAR algorithm and other LDC algorithms. According to the ``theory of data compression"~\cite{shannon93, shannon48}, we conclude that almost every LDC uses Shannon entropy as its `logic base' in conducting a lossless compression. In fact, repetition of characters in a certain frequency based on the theory of probability is embedded in such LDCs. In layman's terms, information entropy is the same as ``randomness." A string of random letters and numbers along the lines of ``5f78HJ2Z2Xp4V7Vb6" can be said to have high information entropy, or, large amounts of entropy, while the complete works of Shakespeare can be said to have low information entropy. Their LDC products are quite variant, which depend on content pattern probability or character rate of recurrence. FBAR LDC, however, deals with the computation of binary logic regardless of content size and type, whereas other techniques are not bothered about. Binary logic in FBAR deals with individual bits, their combination, repetition, cubic conservation and not character repetition. This means, based on a fixed size character reference table, Table~\ref{tab2}, we derive a more certain equation (least zero order $H$ values), which is logarithmically the least probabilistic with discrete entropy (in bpc), compared to Shannon's entropy rate $H$ on English source alphabet $\mathbb{A} = \{ a, b, c, d, \ldots, z, \texttt{space} \}$ given by \begin{equation} \label{eq:24}
H_{\mathbb{A}}  = \log _2 m   = 4.75 \, {\rm  bpc} \ , \ {\rm where} \ m= 27 \, ,
\end{equation}

\noindent and for higher orders of $H$, for a given text source made up of English alphabet letters, becomes 4.07, 3.36, 2.77 and 2.3 bpc, respectively. In FBAR, however, fixed values of $\mathcal{C}$ for every double-efficient order remain
\begin{equation}\label{eq:25}
H_{ \wedge  \vee \left( b \right)}  = \log _b \left| \beta  \right|  = \left[ {0,2} \right]\, {\rm  bpB}
\end{equation}

\noindent and for a binary sequence $\beta$, the binary probability of two states, $b = 2$, constructing 1 char, entropy $H$ becomes 2, 1 and 0 bits per byte (bpB), regardless of source for a given fixed size binary reference code (compare this with Mackay (2005)~\cite{mackay}). This makes the algorithm to compute information reliably based on fuzzy-binary, rather than string characters. The DE process in Eq.~(\ref{eq:25}), evaluates every character by using and-or, pure and impure logic, and from there, further LDCs between bits of information. Equation~(\ref{eq:24}), however, deals with the random process to evaluate the whole sequence of characters using probability theory for an LDC result. Equation~(\ref{eq:25}), by comparison, improves less dependency on symbolic representations, and has a firm dependency on binary logic, thereby, fuzzy, and finally, DE logic. The latter, however, remains quite intact with higher orders of probability equations promoting Shannon zero-order through third-order and general models, in simplistic sizes of LDC. Reasoning that, DE logic by itself is based on probability behavior over bit states. We define the relationships between logical events, ``bit states" of the FBAR algorithm, as LDC causality in form of supreme states of compression. For any data type at a DE level, the current model (Fig.~\ref{fig2}) providing DE compressions holds good for superdense coding operators by Bennett~\cite{bennett}.

In our next report, we improve our model design, reconfiguring \textbf{znip} flags in an extended translation table in aim of super-compressing an encoded message, thereby decode and decompress. The FBAR logic would then be called a DE-negentropic and-or logic (DENAR) in its ultimate performance of LDCs. Hence, a \emph{negentropy} $<0$\,bpB of Eq.~(\ref{eq:25}), denoting DE's above 87.5\% compression for a universal predictability, is not farfetched in reality.

\section{The FBAR Entropic Comparisons} \label{sect7}
The following bar chart, Fig.~\ref{fig6}, gives a compression ratio comparison for our chosen algorithms. In this case, we chose WinZip, GZip, WinRK algorithms based on their respective ranks (see, e.g., Bergmans (1995)~\cite{bergmans}). The detailed empirical and statistical analyses of these algorithms compared to FBAR, based on the non-parametric Friedman test, have been initially reported by Alipour and Ali (2010)~\cite{alipour10}.

\begin{figure}[th]
\begin{center}
\includegraphics[scale=1.7]{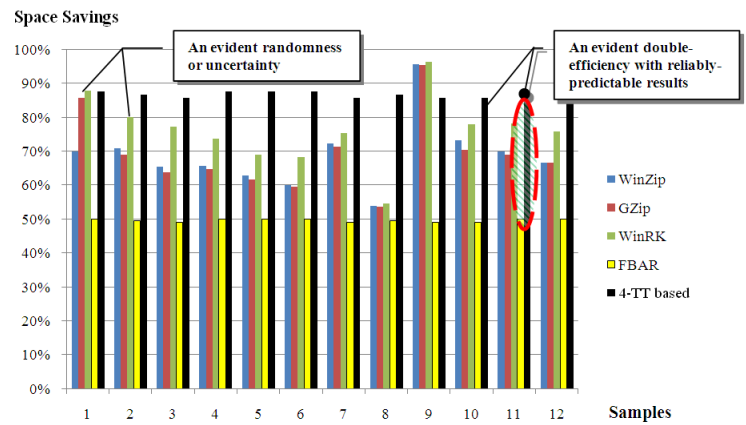}
\end{center}
\caption{LDC ratio comparisons between FBAR/4-{\bf TT} based and other algorithms on the 12 char-based documents selected by random. The ranking of the algorithms and their non-parametric Friedman comparisons were motivated and defended in the initial thesis of Alipour and Ali (2010)~\cite{alipour10}.  \label{fig6}}
\end{figure}

Figure~\ref{fig6}, further shows the difference between uniformity of LDC values on FBAR, compared to the random performance (or uncertainty) of others with a highly-ranked algorithm, WinRK. Bitrate and memory usage comparisons between FBAR and WinRK algorithms are given in Fig.~\ref{fig7}.

\begin{figure}[th]
\begin{center}
\includegraphics[scale=2.25]{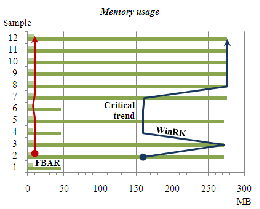}
\hspace{2.5in}\footnotesize (a) \vspace{6pt} \\
\includegraphics[scale=2.25]{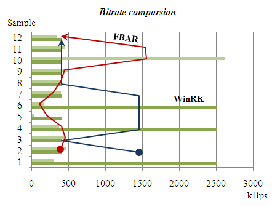}
 \hspace{2.5in} (b)
\end{center}
\caption{Algorithmic performances of FBAR and WinRK. (a) Memory usage; (b) bitrate.\label{fig7}}
\end{figure}

The selection of an LDC algorithm depends on the following criteria as applicable characteristics to all LDC algorithms:

\begin{enumerate}
\item The ability to compress input data losslessly regardless of type, size and complexity. If data type matters, e.g., being of textual type, must compress textual data losslessly, i.e., the $\mathcal{C}'$ data after $\mathcal{C}$ must be identical to the original.
\item Use memory for data access and management issues efficiently, e.g., data rate and spatial occupation of bits during $\mathcal{C}$, i.e., when encoded and referenced upon.
\item Must have a dictionary coder for validating data, referencing and de-referencing them during the reconstruction phase of data i.e. $\mathcal{C}'$.
\end{enumerate}

Our test samples were char-based, suitable for any ASCII string conversion, e.g., *.txt, *.tex and *.htm, and were selected by random, in terms of content size, content characters i.e. quantity relative to the supported char data types, which were also different, no matter how limited our choice for the current version. The data type is random on the FBAR's evolutionary grade (Fig.~\ref{fig7}) due to its explicit behavior in converting chars to binary and vice versa, via the 4D translation table.

The present FBAR supports char-based data types. However, the conversions of diverse binary types are achievable due to the universality of the ASCII table associated with its translation table, which is a midpoint $\mathcal{C}$ converter between char and binary, giving a new data type standard for flags. Meaningly, the random selection of samples (documents) were all ASCII-based as char-to-binary by our algorithm. The selection of packages or LDC algorithms, however, was not random, and was based on the three criteria given above relative to their ranks.

Figure~\ref{fig6} shows a distinctive alignment and correlation of $\mathcal{C}_r$'s of the FBAR and DE versions to others. Comparatively, the new algorithm is more reliable in LDC results with consistence in spatial efficiency values performed on compression, which is due to having fixed-size components like the {\fbox{\bf {TT}}} file, and contrasts other algorithms that create a new dictionary code for each I/O load.

\begin{figure}[th]
\begin{center}
\includegraphics[scale=2.1]{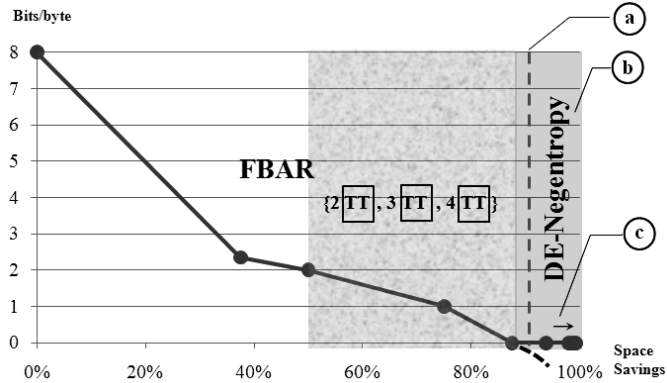} 
\end{center}
\caption{The pure FBAR in its version-to-version evolution, mutates to a DE-Negentropic and-or (DENAR) logic via its hybrid version on x86 machines; a) the EB barrier; b) quantum machines; c) negentropy leading to a universal predictability.\label{fig8}}
\end{figure}

Based on the three characteristics criteria, Fig.~\ref{fig8} portrays the selection of FBAR type as oriented to DE-negentropic type during implementation. Its simulation grade on x86 machines, reaches 87.5\% fixed LDC scenarios. The 87.5\% LDC indicates the lower-bound interval of Eq.~(\ref{eq:25}). The zone indicating x86 limits for the hybrid version, shown as FBAR$\thicksim$DENAR in Fig.~\ref{fig8}, inclusive of the regular FBAR versions (1\fbox{\textbf{TT}} to 4\fbox{\textbf{TT}} usage), continues to expand within the DENAR territory. This means, the structural integrity of the FBAR dictionary (or the 4D grid model) at $H = 0$\,bpB final version on x86, is significantly changed in favor of superdense coding or a quantum machine territory. In one word, \emph{FBAR mutates from version-to-version with uniformly-fixed values on space savings}.

The negative entropy of Eq.~(\ref{eq:25}), as Eq.~(\ref{eq:26}), denotes universal predictability, giving values $\ge$  93.75\% compression, as estimated. This model could be considered as a solution to \emph{complex negentropy problems}~\cite{hyv} in signal processing and information theory, making the current model universal for negative and positive ranges of Eq.~(\ref{eq:25}). Alternatively, we constrict Eq.~(\ref{eq:25}) in terms of
\begin{equation} \label{eq:26}
 - H_{ \wedge  \vee \left( b \right)}  = \log _b \left| \beta  \right|  < 0 \, {\rm  bpB} \ , \ {\rm where} \ b=2 \, .
\end{equation}

For the positive range of Eq.~(\ref{eq:26}), as Eq.~(\ref{eq:25}), twelve documents were given to four different LDC compressors (in random order), relative to their bitrate performance for each LDC execution. The spatial and temporal estimates are given in Table~\ref{tab4}. Process time of a test, and percentages of compression, were also measured. The resulted data on both spatial and temporal performances are expanded from 1\fbox{\textbf{TT}}, to 4\fbox{\textbf{TT}} inclusions. From Eqs.~(\ref{eq:21}) and (\ref{eq:22}), in Table~\ref{tab4}, the $t_L$ on 4\fbox{\textbf{TT}}s, without parallel processing, is hypothesized for $\mathcal{C}_{\rm matrix}$= 3$\mathcal{C}_{\max}$= 5.5\,s, satisfying the `extended if-else' nodes in form of a $16^{4\rm D}\times$4 address format (Section~\ref{sect5.4.2}). The nodes are quite local in the $\mathcal{C}_{\rm matrix}$ code, and the HLLC merely concentrates on read/write operations, constituting the current LDD$\times$1\fbox{\textbf{TT}} vs. LDD$\times$4\fbox{\textbf{TT}} scenarios.

\begin{table}[th]
\begin{center}
\begin{minipage}{\textwidth}
\caption{Estimates on $\mathcal{CC}'$ phases with rate performance on FBAR via 1\textbf{TT} against 4\textbf{TT} \label{tab4}}
{\footnotesize \begin{tabular} {@{}ll r c rrr cc@{}} \toprule
No. & File  & Size   & \multicolumn{3}{c} {$t_L =$ CPU time/s} & Compressed size  & bpc  \\
&  & (KB) & \multicolumn{3}{c} {\underline{LDC/LDD $\times$ 1{\bf TT} vs. 4{\bf TT}}} & (KB)  &   \\
&  &  &  \multicolumn{3}{c} {\scriptsize $\mathcal{C}_{\rm matrix}$  \ \ \ \ LDC \ \ \ \ \ \ \ \  LDD} &  &    \\\hline
1 & text &  60.16 &  0.06 & $0.2:0.26$ & $0.24:0.3$ & $30.39:7.52$ & \{0,2\}  \\
2 & book1 & 662.34 & 0.42 & $1.45:1.87$ & $1.40:1.82$ & $334.62:82.79$ & \{0,2\}   \\
3 & book2 & 1730.54 & 2.25 & $3.95:6.2$ & $3.43:5.68$ & $874.28:216.31$ & \{0,2\}  \\
4 & paper1 & 51.28 & 0.03 & $0.14:0.17$ & $0.20:0.23$ & $25.9:6.41$ & \{0,2\}  \\
5 & paper2 & 114.73 & 0.345 & $0.375:0.72$ & $0.32:0.665$ & $57.96:14.34$ & \{0,2\}  \\
6   & paper3 & 10.02 & 0.03 & $0.06:0.09$ & $0.10:0.13$ & $5.06:1.25$ & \{0,2\}   \\
7   & web1 & 730.24 & 0.66 & $1.85:2.51$ & $1.71:2.37$ & $368.92:91.28$ & \{0,2\}   \\
8  & web2 & 584.1 & 0.6 & $1.49:2.09$ & $1.36:1.96$ & $295.09:73.01$ & \{0,2\}   \\
9   & log & 1797.77 & 0.81 & $3.7:4.51$ & $3.58:4.39$  & $908.25:224.72$ & \{0,2\}   \\
10  & cipher & 759.42 & 0.12 & $0.29:0.41$ & $0.32:0.44$ & $383.66:94.92$ & \{0,2\}   \\
11   & latex1 & 204.3 & 0.09 & $0.46:0.55$ & $0.49:0.58$ & $103.21:25.53$ & \{0,2\}   \\
12 & latex2 & 151.99 & 0.09 & $0.45:0.54$ & $0.92:1.01$ & $76.78:18.99$ & \{0,2\}   \\
\rowcolor[rgb]{.0,.0,.0} \textcolor{white}{0} & \textcolor{white}{\textbf{TT} file} &  \textcolor{white}{$\approx$\,8\,MB} &  \hphantom{0}\textcolor{white}{N/A} \ &   \textcolor{white}{N/A}\hphantom{000} &  \textcolor{white}{N/A}\hphantom{000} &  \textcolor{white}{N/A}\hphantom{0000} & \textcolor{white}{\{0,2\} read}  \\
   & Total & 6856.89 & 5.505 & $14.41:19.92$ & $14.07:19.57$ & $3464.12:857.07$ & \{0,2\}   \\\hline
\end{tabular} }
\tablefootnote{a}{Estimates on compression with rate performance on FBAR's LDC and LDD using 1\textbf{TT} vs. 4\textbf{TT} file-set. The main columns representing the LDC $\times$ 1\textbf{TT} vs. 4\textbf{TT}, and LDD $\times$ 1\textbf{TT} vs. 4\textbf{TT} ratios are shown in the three sub-columns of the third column, $\mathcal{C}_{\rm matrix}$, LDC and LDD.}
\tablefootnote{b}{The base file accessed by the program for read/write operations is the \textbf{TT} file $\approx$ 8\,MB.}
\end{minipage}
\end{center}
\end{table}

The bitrate results, are the least random compared to other LDCs, and thus are in conformity with the spatial results in Fig.~\ref{fig7}. Figure~\ref{fig7} shows the bitrate performance on the 12 test documents, with their critical and optimal trends. The adapted version focusing on HLLC results for a simulated 50\% DE-LDC is given in Table~\ref{tab4}. In it, we further computed the time factor as CPU time in seconds, as well as LDD and LDC results. The bitrate, relative to memory usage, was observed between the algorithms on `space savings': WinRK vs. FBAR. As we can see, for higher bitrate performances, WinRK has a critical usage of memory per input sample. In some cases, even having 10\,kBps for encoding and decoding data, required 800\,MB memory on a 2\,GHz Athlon CPU. This drawback ranks WinRK's memory performance lower than expected compared to FBAR.

When we associate the left chart values with the right chart in Fig.~\ref{fig7}, it is evident that the empirical data relative to memory usage on FBAR is optimal, and uniformly correlated except a bitrate jump on sample \# 10. This is due to the excessive repetition of characters within the sample. The original input chars were ignored due to their pattern simplicity for storing data. Hence, the algorithm is not forced to take in too much information, thus its computation. The average bitrate was estimated 475\,kBps for FBAR, and 925\,kBps for WinRK on 12 samples. For the EB barrier in Fig.~\ref{fig8}, we refer to the explanations provided in Section~\ref{sect5.4}, which concern $\mathcal{C}$'s $>$ 87.5\% scenarios, addressing fixed $\mathcal{C}_r$ values of a DENAR-LDC.

The evolutionary grade of FBAR in Fig.~\ref{fig8}, further illustrates the elicited $\mathcal{C}_r$ ratios, respectively giving 8:8 for 0\%, 8:4 for 50\%, 8:2 for 75\%, 8:1 for 87.5\%, 8:$-1$ for 93.75\%, 8:$\ldots-\infty$ for $\approx$100\%. These ratios correspond to Eqs.~(\ref{eq:25}) and~(\ref{eq:26}) discrete intervals, from present to FBAR future versions, supporting the spatial boundary limits of Hypothesis~\ref{hypo5.2}, for a universal predictability in Hypothesis~\ref{hypo5.3}. To avoid ratio confusion on Eq.~(\ref{eq:25}), for $b=2$, we rather state, $2^H =B$ or \emph{bytes} for 8:$B$, hence Fig.~\ref{fig8} on Eq.~(\ref{eq:26}), does not scale the vast negative space 0\,B$\,<2^{-H}\!<\,$1\,B for 8:$-H$ on DE-negentropy. The predictable $H$-limit is indicated by a dotted curve, descending within the negative space of Fig.~\ref{fig8}. For example, according to Eq.~(\ref{eq:26}), $b$\,=\,$2$, and for a 93.75\% space savings, a $2^{-1}$= 0.5\,B is gained. Meaningly, 0.5 \emph{compressed}\,B is gained against a 100\%-read on 8\,\emph{original}\,B, or an 8\,B\,$input$\,:\,0.5\,B\,$output$, as $\mathcal{C}_r\!=\!\{100\!:\!6.25\}\% = 93.75\%$. This is a highly compressed version of data, which ascends to attain an ultimate LDC.

\section{Costs and Future Work} \label{sect8}
Nowadays, compressors accumulate much more memory space, even more than 250\,MBs, e.g., WinRK in Fig.~\ref{fig7}. This is significant when \emph{overhead information} and \emph{memory caching} issues are studied from the usability aspect of the algorithm. By employing cache memory, the \textbf{TT} file is temporarily loaded into memory and accumulate much lesser space. This is imperative for huge data transmissions above TB limits on the network and elsewhere, satisfying the EB limits explained above.

There would be an additional cost in terms of the 4\textbf{TT} process and management issues discussed in Section~\ref{sect5.4}, or see, its relevant code complexity relations in the same section. In fact, 32\,MB in size, as part of the algorithm's package, is affordable for users. The more the users pay, the greater guaranteed LDCs they get. For example, for 50\%, 1\textbf{TT} shall suffice, which is 8\,MB as part of the package or program, and not something being generated each time we load a document to the program (unlike other compressors). The nature of these \textbf{TT}s is being robust in their dimensions, content, and always static in length.

Such algorithm components were discussed to prove the spatial and temporal limits of Hypothesis~\ref{hypo5.2}, such as DE logic and model representation, following its applicability and usage in code. Future works shall focus on Hypothesis~\ref{hypo5.3} addressing ultimate DE-LDCs, its robustness, complexity, reliability, confidence, etc., relative to the universality of the 4D model. We have, however, showed confidence on predictability, such as LDC values being predicted before compression. This was done by satisfying Hypothesis~\ref{hypo1} via the \textbf{TT} and \textbf{G} components, as if any randomness is ``self-contained" within their code.

The FBAR algorithm, based on our current analysis and results, addresses 1\,terabyte (TB) and beyond the EB limits (Sections \ref{sect2.1.4} and \ref{sect5.4}), hence its components, as a whole, are applicable to databases as VLDBs. From the \emph{qualitative} aspect, suppose within the context of \emph{knowledge modeling}, we employ the translation table with the FBAR interpreter from Section~\ref{sect5.2}. Since the table is ASCII-based, the code interpreter classifies data for a unique set of conversions (a compression or encoding). When decoding, the interpretation of input knowledge by the computer, e.g., the English language, gives a resultant output satisfying the compactness of first-order logic~\cite{dawson, godel}. Thus, this information product is content-based and semantically familiar to the human knowledge. This promotes the usability aspect of FBAR on databases, information retrieval systems, their design and architecture.

From the \emph{quantitative} viewpoint, Eqs.~(\ref{eq:25}) and~(\ref{eq:26}) efficiency against Shannon relations like Eq.~(\ref{eq:24}), becomes evident in future publications. The current work briefly discussed Shannon, and proves the relatedness of different types of logic. The current algorithm is based on FBAR logic and does not use Shannon. Shannon in this work, however, is referred to make performance measurements and demonstrate the uniqueness of FBAR compared to randomness, i.e., redundant symbols of information stored and observed in other algorithms, with relatively well-known entropy orders to Shannon codeword. (Recall, Section~\ref{sect2.1.1}.)

In our future reports, we propose an empirical resolution over any complexity issues raised in Section~\ref{sect5}, beyond the 4D model representation. We formulate an LDC in terms of \emph{linear codes} to reduce the length of current FBAR code by \emph{self-dual coding}, which further obtains optimal DE ratios that fall inside the $[50, 100)\%$ interval. As initially conceived by Eq.~(\ref{eq:26}), the upper-bound of the interval denotes an entropy norm \begin{equation}\|f\|_{H}=  \sqrt{\log_{b} p\left(\ell(|x\cdot x'|)\right)} \geq \sqrt{\sum_{i=0}^{b} {b \choose i}  \dot{\delta}} \ ; \  \dot{\delta} \in (\imath^2\infty^2,\imath^2]\ \text{bpB} \label{eq:5.4}\end{equation}

\noindent where $b$ is the total number of fuzzy-bit states of an input stream $xx'$, assuming length $\ell(x) =\ell(x')$, and $p$ is the predictable information output carrying any $b$. The focus is to gain efficient temporal and spatial lossless outputs for a $b \in [2, 4]$ as fuzzy-bit, and a $b=8$ for qubit type compressions. To conduct these LDCs, we employ \textbf{znip} operators from an \emph{octonion}~\cite{sabinin} field, transforming its information content into a bivector code. The method is to encode data from an 8-dimensional vector space (octonion) into a 4-dimensional (quaternion), where each dimension stores $2^n$ bits, noting $n$ as the number of pairable bits. Using \emph{parity check} and \emph{quaternary linear coding} such as \emph{puncturing matrix}~\cite{wan}, the Hamming distance of the encoded data is reduced and further compacted for each hypercube node holding compressed data. From there, the shortened code (deleted bits) stored in the 4D \textbf{G} field is accessed and looked up in the \textbf{TT}able for decoding. Therefore, specific addresses are decoded for each matched code from the field. The \emph{fuzzy-quantum binary} type LDC, is by combining the hypercube with a Bloch sphere~\cite{churuscinski}, projecting an $N$-point probability data from its surface onto the cube's field, thereby producing an $n$4D-superdense model for a maximum DEN-LDC. This results in a new hypothesis, yet to be proven as follows:

\begin{hypothesis} Compressing all of our Universe's data into one single bit or a near-zero byte within an infinitesimal time-frame losslessly, is by combining the 4D FBAR hypercube with a Bloch sphere, giving an \emph{n}4D-superdense model. Its translation table will contain all logic states of information, emitting a complete $\mathcal{C}'$ product.
\end{hypothesis}

This combinatorial model, transforms any other type, like Shannon inequality into FBAR. For example, as compression $\mathcal{C}$ approaches $\infty$ for input $x$, a normalized vector quantization is achieved i.e., establishing a predictable density curve, close to a value of Dirac delta $\delta$ identity~\cite{gelfand}, such that \begin{equation}\int_{\dot{\delta}\rightarrow\mathcal{C}}^{\infty}\delta(x) \ \mathrm{d}x= \frac{\ell(xx')}{V_{Q_y}}\gtrapprox1 \  \label{eq:5.5}\end{equation}

\noindent where function $\ell$ returns the total number of input bits, and $V$ is the volume of the cube containing the bits' compressed product as given by Eq.~(\ref{eq:2b}). This product for either 4D or \emph{n}4D model is four dimensions, self-containing all logic states originated from the $2^n$ field dimensions, as our main base indicator in terms of $\sum_{i=0}^{b} {b \choose i} = 2^b$, by Eq.~(\ref{eq:5.4}). The outcome is a compression system $S(\delta)$ usable in current and future generation computers, where its performance is measured by Eqs.~(\ref{eq:5.4})  and~(\ref{eq:5.5}) that accompany temporal bounds on the size of the generated linear codes in form of evaluation tables. Of course, other spatial and temporal relations such as Hamming rate from previous chapters shall contribute results to these tables.

\chapter*{Conclusion} \label{sect9}
In this paper, we have demonstrated that the 4D bit-flag model, in its logic base, is self-contained for lossless data compaction and compression. Thus, it is more reliable for data read and write, compared to probabilistic methods in implementing a lossless compression, due to having predictable values in its LDC results.

By using this model, an LDC program converts any ASCII character to its compressed version, double-efficiently in an exponential manner. Based on our results, the 4D model proves universal predictability from one version to another via a universal translation table for data conversions. Thus, it gives reliably fixed results in every data conversion output per se. Therefore, this makes all algorithmic components of the model, and its self-contained bit-flag values universal as well.

Perceivably, the present FBAR compresses data with fixed compression ratios, where other compressors do not. Almost every lossless compressor uses probabilistic Shannon entropy as its `logic base' in conducting LDCs. FBAR, however, achieves higher space savings, above 50\% as estimated, simulated and discussed in theory from its DE coding technique, as well as a 4D model representation. The FBAR products were studied from an LDC and LDD viewpoint in terms of delivering DE-$\mathcal{C}$'s $>$ 87.5\%, or, a DEN $<$ 0 bpB. It is conclusive that, this algorithm contains predictable values for every double-character input. The predictable fixed value, allows a user to know how much physical space is available within a reasonable time, before and after compression. This confidence in predictability makes FBAR a reliable version compared to the probabilistic LDCs available on the market. Optimal and greater efficiencies based on linear code integration for ultimate FBAR LDCs is also proposed and aimed for in our future reports.

The FBAR algorithm is novel in most aspects such as encryption, binary, fuzzy and information-theoretic methods such as probability. To this account, the fields of interest encompass the newly-born FBAR model useful to information theory mathematicians, electrical and computer engineers as well as computer scientists for its logic, and software engineers for its applications.

\begin{acknowledgements}
I gratefully thank my supervisor, Dr. T. A. Gulliver, at the Department of Electrical and Computer Engineering, University of Victoria, Canada, for his constructive remarks in improving the model and theoretical aspects of this work for a better presentation. I also like to thank the anonymous reviewers for their valuable comments and suggestions.
\end{acknowledgements}





\begin{notationsandacronyms}

In this section, we present the main acronyms and notations used in this article with recognition of those notations and definitions formulated in Table~\ref{tab1}, and elsewhere. We finally outline in a separate table, the key interpretation conceived for these acronyms to avoid any confusion when studying the article.

\begin{table}
\begin{minipage}{\textwidth}
\caption{Main acronyms and notations that have been used throughout this article.}
\begin{center}\footnotesize
\begin{tabular*}{\textwidth}{@{}cll@{}}\hline
Acronym &  Meaning  & Employed notations \\
& & from Table~\ref{tab1} and elsewhere  \\ \hline
LDC & Lossless Data Compression & \multirow{1}{5.5cm}{$\mathcal{C}_r, \mathcal{C}, H, \mathcal{s}, \ell$, {\bf O}, {\bf TT}, {\bf P}, {\bf G}  }\\
LDD & Lossless Data Decompression & \multirow{1}{5.5cm}{$\mathcal{C}_r, \mathcal{C}', H, \mathcal{s}, \ell$, {\bf O}, {\bf TT}, {\bf P}, {\bf G}  } \\
4D & Four-Dimensional & \multirow{1}{5.5cm}{$\mathbf{F}_{xx'}, \mathbf{F}_{y}, \mathbb{R}^{2^n}, \mathbb{C}\ell_4\mathbb{R}^{4}$, $h^2\mathbf{e}^2_{ij}, Q_{xx'}$, $Q_y, A_{r\times4}$} \\
& & \\
FBAR & Fuzzy Binary AND-OR & \multirow{1}{5.5cm}{$\Phi_\wedge \, , \,  \Phi_\vee$, \, $\bigcap, \, \bigcup$, $\mathcal{A}, \mathcal{\tilde{A}}$  }  \\
\textbf{O} & Original & $xx'$ \\
\textbf{TT} & Translation Table ${}^{\mathrm{a}}$ & $i \times j \times k \times l$, $xx'$, $y$  \\
\textbf{G} & Grid ${}^{\mathrm{b}}$ & $\mathcal{s}_{\mathrm{out}}$, $y$\\
\textbf{P} & Program  & \multirow{1}{5cm}{$\mathcal{s}_{\mathrm{in}}$, $xx'$,  $\mathcal{s}_{\mathrm{out}}$, $y$, $\pmb\mho$, \textbf{znip}, $\hat{{\bm\mu}}$, {\bf O}, {\bf TT}, {\bf G}, $\mathcal{C}_{\rm matrix}$}  \\
& & \\
I/O & Input/Output &$\mathcal{s}_{\mathrm{in}}$, $xx'$,  $\mathcal{s}_{\mathrm{out}}$, $y$, {\bf O}, {\bf TT}, {\bf G}, {\bf P}  \\
ASCII &  \multirow{1}{5cm}{American Standard Code for Information Interchange} & char $\leftrightarrow\beta$ \\
& & \\
ISO & \multirow{1}{5cm}{International Organization for Standardization} & N/A\\
& & \\
CPU & Central Processing Unit & $\sum T_i, t_L$   \\
VLDB & Very Large Database ${}^{\mathrm{c}}$ & $\sum \mathbf{O}_i \gg \mathbf{G}$ \\
RAM & Random Access Memory & $A$ \\
B & Byte & $\hat{{\bm v}}, \beta$ \\
Bps & Bytes per second & $R, R_{\mathbb{H}}$ \\
b & Bit & $\hat{{\bm v}}, \beta, b$ \\
bps & Bits per second & $R, R_{b}$ \\
KB & Kilobyte & $\ell(A_{r\times4})$ \\
MB & Megabyte & $\ell(A_{r\times4})$  \\
EB & Exabyte ${}^{\mathrm{c}}$  & $\sum\ell(A_{r\times4})_i$, $\sum \mathbf{O}_i$ \\
DE & Double Efficient & $H_{\wedge\vee(b)}$ \\
DENAR & \multirow{1}{5cm}{Double Efficient Negentropic AND-OR} & $-H_{\wedge\vee(b)}$, $-H_{\wedge\vee(fb)}$ \\
& & \\\hline
\end{tabular*}
\end{center}
\tablefootnote{a}{printed in bold to represent a field of data vectors as readable I/O dictionary code.}
\tablefootnote{b}{printed in bold to represent a field of data vectors as I/O storable/compressed data.}
\tablefootnote{c}{briefly discussed in this article, however, mainly considered for future articles covering issues related to VLDB compression, transactions and data management issues.}
\end{minipage}
\end{table}

\begin{table}
\begin{minipage}{\textwidth}
\caption{Notations and acronyms in Table 5.2 as interpreted relevant to the current topic.}
\begin{center}\footnotesize
\begin{tabular*}{\textwidth}{clcl}\hline
Acronym & As \\ \hline
LDC & \multirow{1}{8cm}{an encoding phase, algorithm, code or program} \\
LDD & \multirow{1}{8cm}{a retrieving phase, algorithm, code or program} \\
4D &  \multirow{1}{8cm}{a partitioned field, hypercube, vector or operator}\\
FBAR & logic or algorithm \\
\textbf{O} &  \multirow{1}{8cm}{an input file or component when an LDC phase initiated} \\
\textbf{TT} & a dictionary coder, addresses or reference code \\
\textbf{G} & an output file or component at the LDC phase\\
\textbf{P} & a source code, file or component \\
I/O &  a data operation \\
ASCII &  a standard table of codes \\
ISO &  \multirow{1}{8cm}{a responsible organization for management standards}  \\
CPU & \multirow{1}{10cm}{the portion of a computer system that carries out the instructions of a computer program} \\
& &\\
VLDB & \multirow{1}{10cm}{a database that contains a high number of tuples (database rows), which occupies
large physical filesystem storage space, usually more than 1 terabyte = $10^{12}$ bytes} \\
& &\\
& &\\
RAM & a form of computer data storage \\
B & \multirow{1}{10cm}{a unit of digital information mostly consists of eight bits} \\
Bps & \multirow{1}{10cm}{Byte-rate as the number of bytes that are conveyed or processed per unit of time, for evaluating algorithm spatial and temporal performance}  \\
& &\\
b & \multirow{1}{10cm}{a unit of digital information as a contraction of binary digit with a value of 0 or 1 logic}  \\
& &\\
bps & \multirow{1}{10cm}{Bit-rate  as the number of bits that are conveyed or processed per unit of time, for evaluating  algorithm spatial and temporal performance}  \\
& &\\
KB &  \multirow{1}{10cm}{a multiple of the unit byte for digital information with a value of either $10^{24} (2^{10})$ bytes or $1000 (10^{3})$ bytes} \\
& &\\
MB & \multirow{1}{10cm}{a multiple of the unit byte for digital information storage or transmission with two different values, such that 1048576 bytes $(2^{20})$ generally for computer memory} \\
& &\\
& &\\
EB &  \multirow{1}{10cm}{a unit of information or computer storage equal to one quintillion bytes or 1 EB = $10^{18}$ bytes = 1073741824 gigabytes = 1048576 terabytes} \\
& &\\
DE & \multirow{1}{10cm}{spatial inclusive of temporal double-efficiency, measured by FBAR entropy rate $H_{\wedge\vee(b)}$ with CPU time $t_L$, and Hamming rate $R$ for proper compression} \\
& &\\
DENAR &  \multirow{1}{10cm}{a negative entropy (NE) measurement of information based on FBAR logic,  measured in rate $-H$ or $H$ $= |\log_{b} p(x)|$ that falls inside the $(-\infty,0)$ interval} \\
& &\\\hline
\end{tabular*}
\end{center}
\end{minipage}
\end{table}

\end{notationsandacronyms}

\bibliographystyle{plain}	
\bibliography{fbar}		

\begin{thebibliography}{10}

\bibitem{dembo}
{A. Dembo, T. M. Cover, and J. A. Thomas}.
\newblock Information theoretic inequalities.
\newblock {\em IEEE Trans. Info. Theo.}, 37(6):1511--1517, 1991.

\bibitem{alipour10}
P.~B. Alipour and M.~Ali.
\newblock An introduction and evaluation of a fuzzy binary and/or compressor.
\newblock Master's thesis, School of Computing, Blekinge Institute of
  Technology, Sweden, 2010.

\bibitem{anton}
H.~Anton.
\newblock {\em Calculus: A New Horizon}.
\newblock Wiley, New York, 6th edition, 1999.

\bibitem{arnold}
V.~I. Arnold.
\newblock Relatives of the quotient of the complex projective plane by the
  complex conjugation.
\newblock {\em Tr. Mat. Inst. Steklova}, 224:56--67, 1999.

\bibitem{balakrishnan}
V.~K. Balakrishnan.
\newblock {\em Theory and Problems of Combinatorics}.
\newblock McGraw-Hill, New York, 1995.

\bibitem{bartle}
R.~G. Bartle.
\newblock {\em The Elements of Integration and Lebesgue Measure}.
\newblock Wiley Interscience, New York, 1995.

\bibitem{bennett}
C.~Bennett and S.~J. Wiesner.
\newblock Communication via one- and two -particle operators on
  einstein-podolsky-rosen states.
\newblock {\em Phys. Rev. Lett.}, 69:2881–--2884, 1992.

\bibitem{bergmans}
W.~Bergmans.
\newblock {\it Maximum Compression.} software ranks published online in maximum
  compression benchmark.
\newblock Phys. Rev. Lett., 2011.

\bibitem{boole}
G.~Boole.
\newblock The calculus of logic.
\newblock {\em Cambridge and Dublin Math. J.}, 3:183–--98, 1848.

\bibitem{bowen}
J.~P. Bowen.
\newblock Hypercubes.
\newblock {\em Pract. Comput. J.}, 5(4):97–--99, 1982.

\bibitem{bush}
S.~F. Bush.
\newblock {\em Nanoscale Communication Networks}, chapter 1.6, 1.6.2, pages
  19--29.
\newblock Artech House Publishers, USA, 2010.

\bibitem{shannon93}
{C. E. Shannon, A. D. Wyner, and N. J. A. Sloane}.
\newblock {\em Claude E. Shannon Collected Papers}.
\newblock Wiley IEEE Press, New York.

\bibitem{churuscinski}
D.~Chru\'{s}ci\'{n}ski.
\newblock Geometric aspect of quantum mechanics and quantum entanglement.
\newblock In {\em J. Phys. Conf. Ser.}, volume~35, pages 9–--16, 2006.

\bibitem{cockle}
J.~P. Cockle.
\newblock On systems of algebra involving more than one imaginary.
\newblock {\em Phil. Magaz.}, 35(3):434--435, 1849.

\bibitem{conway}
A.~W. Conway.
\newblock On the application of quaternions to some recent developments in
  electrical theory.
\newblock In {\em Royal Irish Acad. Proc.}, volume~29, pages 1–--9, 1911.

\bibitem{coxeter}
H.~S.~M. Coxeter.
\newblock {\em The Coxeter Legacy: Reflections and Projections}.
\newblock American Mathematical Society, USA, 2006.

\bibitem{dawson}
J.~W. Dawson.
\newblock The compactness of first-order logic: From g$\ddot{\mathrm{o}}$del to
  lindstr$\ddot{\mathrm{o}}$m.
\newblock {\em Histor. and Phil. of Logic}, 14(1):15--37, 1993.

\bibitem{godel}
K.~G$\ddot{\mathrm{o}}$del.
\newblock {\em The first proof of the completeness theorem}.
\newblock PhD thesis, University of Vienna, Austria, 1929.

\bibitem{gelfand}
I.~M. Gelfand and G.~E. Shilov.
\newblock {\em Generalized functions}, chapter 1.1, pages 1--5.
\newblock Academic Press, New York and London, 1st edition, 1968.

\bibitem{gilheany}
S.~Gilheany.
\newblock Moore's law and knowledge management.
\newblock Archive Builders, 2011.

\bibitem{girard}
P.~R. Girard.
\newblock The quaternion group and modern physics.
\newblock {\em Europhys. J.}, 5:25--32, 1984.

\bibitem{gottwald}
S.~Gottwald.
\newblock {\em A Treatise on Many-Valued Logics (Studies in Logic and
  Computation)}.
\newblock Research Studies Press, Baldock, Hertfordshire, 2001.

\bibitem{green}
G.~D. Green and D.~Newth.
\newblock Towards a theory of everything? - grand challenges in complexity and
  informatics.
\newblock {\em Europhys. J.}, 8:1--12, 2001.

\bibitem{hamilton1}
W.~R. Hamilton.
\newblock {\em Lectures on Quaternions}, page 730.
\newblock Dublin University Press, Utah, 1853.

\bibitem{hamming}
R.~W. Hamming.
\newblock Error detecting and error correcting codes.
\newblock {\em Bell Sys. Tech. J.}, 29(2):147–--160, 1950.

\bibitem{hyv}
A.~Hyv$\ddot{\mathrm{a}}$rinen.
\newblock {\em Measuring non-Gaussianity by negentropy. Independent component
  analysis}.
\newblock Cambridge University Press, 2001.

\bibitem{jacas}
J.~Jacas and L.~Valverde.
\newblock {\em On fuzzy relations, metrics and cluster analysis}, volume~I.
\newblock Approx. Reasoning Tools for AI, 1990.

\bibitem{jackson}
P.~Jackson and D.~Sheridan.
\newblock Clause form conversions for boolean circuits.
\newblock In {\em 7th Int. Conf. on Theo. and Applic. of Satisfiability
  Testing, Springer}, pages 183--189, Vancouver, Canada, 2005.

\bibitem{zadeh96}
{L. A. Zadeh, G. J. Klir, and B. Yuan}.
\newblock {\em Fuzzy Sets, Fuzzy Logic, Fuzzy Systems}.
\newblock World Scientific Publishing Co., Singapore.

\bibitem{sabinin}
I.~P.~Shestakov L.~V.~Sabinin, L.~Sbitneva.
\newblock {\em Non-associative algebra and its applications}, chapter 17.2,
  pages 235--243.
\newblock CRC Press, Florida, 2006.

\bibitem{lanczos}
C.~Lanczos.
\newblock {\em The Variational Principles of Mechanics}.
\newblock Toronto Press University, Ontario, 4th edition.

\bibitem{leff}
L.~S. Leff.
\newblock {\em PreCalculus the Easy Way}, page 296.
\newblock Barron's Educational Series, Cambridge, 7th edition, 2005.

\bibitem{lidl}
R.~Lidl and H.~Niederreiter.
\newblock {\em Finite Fields}.
\newblock Cambridge University Press, 2nd edition, 1997.

\bibitem{lounesto}
P.~Lounesto.
\newblock {\em Clifford algebras and spinors}.
\newblock Cambridge University Press, 2001.

\bibitem{ebbers}
{M. Ebbers, W. O'\,Brien, and B. Ogden}.
\newblock {\em Introduction to the New Mainframe: z/OS Basics}, chapter 2.4.8
  and 2.4.9, pages 15--19.
\newblock Vervante, Utah, 2006.

\bibitem{mackay}
D.~J.~C. MacKay.
\newblock {\em Information Theory, Inference, and Learning Algorithms}, pages
  2, 15, 67--74, 91--100.
\newblock Cambridge University Press, 7th edition, 2005.

\bibitem{maini}
A.~K. Maini.
\newblock {\em Digital Electronics: Principles, Devices and Applications}.
\newblock Wiley, West Sussex, 2007.

\bibitem{makarychev}
K.~Makarychev.
\newblock A new class of non-shannon-type inequalities for entropies.
\newblock {\em Commun. in Info. and Sys.}, 2:147--166, 2002.

\bibitem{mccabe}
T.~J. McCabe.
\newblock A complexity measure.
\newblock {\em IEEE Trans. Soft. Eng.}, 2(4):308--320, 1976.

\bibitem{murdocca}
M.~J. Murdocca and V.~P. Heuring.
\newblock {\em Information Theory, Inference, and Learning Algorithms}, pages
  50--55, 364--366.
\newblock Prentice Hall, New Jersey, 2000.

\bibitem{papoulis}
A.~Papoulis.
\newblock {\em Probability, Random Variables, and Stochastic Processes}.
\newblock McGraw-Hill, New York.

\bibitem{passino}
K.~M. Passino and S.~Yurkovich.
\newblock {\em Fuzzy Control}, chapter 1 and 2, pages 1--117.
\newblock Addison Wesley Longman, Inc, New Jersey, 1998.

\bibitem{sayood}
K.~Sayood, editor.
\newblock {\em Lossless Compression Handbook}, chapter 1, 2, 6, pages 165--243.
\newblock Academic Press, Elsevier Science, USA, 2003.

\bibitem{scarborough}
C.~T. Scarborough and A.~H. Stone.
\newblock Products of nearly compact spaces.
\newblock {\em Trans. of Amer. Math. Soc.}, 124(3):131–--147, 1966.

\bibitem{shannon40}
C.~E. Shannon.
\newblock A symbolic analysis of relay and switching circuits.
\newblock Master's thesis, Dept. of Elect. Eng, Massachusetts Inst. of Tech.,
  USA, 1940.

\bibitem{shannon48}
C.~E. Shannon.
\newblock A mathematical theory of communication.
\newblock {\em Bell Sys. Tech. J.}, 27:379--423, 623--656, 1948.

\bibitem{shukla}
S.~K. Shukla.
\newblock Transitive closure.
\newblock Dictionary of Algorithms and Data Structures, 2004.

\bibitem{sloane}
N.~J.~A. Sloane.
\newblock Bounds for binary codes of length less than 25.
\newblock {\em IEEE Trans. Inf. Theo.}, 24:81--93, 1978.

\bibitem{smith}
J.~R. Smith.
\newblock {\em Programming the PIC Microcontroller with Mbasic}.
\newblock Newnes Publishers, Elsevier, Oxford.

\bibitem{symonds}
P.~Symonds.
\newblock Part 2: Hamming distance, math 32031: Coding theory.
\newblock Lecture Paper, 2007.

\bibitem{tribus}
M.~Tribus.
\newblock {\em Thermodynamics and Thermostatics: An Introduction to Energy,
  Information and States of Matter}.
\newblock Van Nostrand Company Inc., New York.

\bibitem{wan}
Z.~X. Wan and C.~H. Wan.
\newblock {\em Quaternary Codes}.
\newblock World Scientific Publishing, Singapore, 1997.

\bibitem{watson}
A.~H. Watson and T.~J. McCabe.
\newblock {\em Structured Testing: A Testing Methodology Using the Cyclomatic
  Complexity Metric}.
\newblock National Institute of Standards and Technology, USA.

\bibitem{zadeh65}
L.~A. Zadeh.
\newblock Fuzzy sets.
\newblock {\em Info. and Control}, 8(3):338--353, 1965.

\bibitem{zalta}
E.~N. Zalta.
\newblock Terms: i- fuzzy logic, ii- logical consequence.
\newblock Stanford Encyc. of Phil., 2009.

\bibitem{ziv}
J.~Ziv and A.~Lempel.
\newblock Compression of individual sequences via variable-rate coding.
\newblock {\em IEEE Trans. Info. Theo.}, 24(5):530--536, 1978.

\end{thebibliography}

\end{document}